\documentclass[11pt, letterpaper]{amsart}
\usepackage[left=1in,right=1in,bottom=0.5in,top=0.8in]{geometry}
\usepackage{amsfonts}
\usepackage{amsfonts,amsbsy,amssymb,amsmath,graphicx,amsthm, bm} 
\usepackage{graphicx}
\usepackage[font=small,labelfont=bf]{caption}
\usepackage{epstopdf}
\usepackage{xcolor}

\usepackage{float}
\usepackage{pgfplots}
\usepackage{listings}
\usepackage{longtable}
\usepackage{mathrsfs}

\usepackage{subcaption}
\usepackage{lipsum}

\usepackage{epsfig}
\usepackage{dsfont}
\usepackage{natbib}
\usepackage{mathrsfs}
\usepackage{geometry}

\newtheorem{thm}{Theorem}[section]
\newtheorem{cor}{Corollary}[section]

\newtheorem{lem}{Lemma}[section]

\theoremstyle{definition}
\newtheorem{defn}{Definition}[section]

\newtheorem{exmp}[thm]{Example}

\theoremstyle{remark}
\newtheorem{rem}{Remark}[section]

\definecolor{energy}{RGB}{114,0,172}
\definecolor{freq}{RGB}{45,177,93}
\definecolor{spin}{RGB}{251,0,29}
\definecolor{signal}{RGB}{203,23,206}
\definecolor{circle}{RGB}{217,86,16}
\definecolor{average}{RGB}{203,23,206}

\colorlet{shadecolor}{gray!20}
\pgfplotsset{compat=1.9}

\usepgflibrary{fpu}

\newcommand{\argmin}{\mathop{\mathrm{arg\,min}}}

\newcommand{\G}{\mathcal{G}}
\newcommand{\F}{\mathcal{F}}

\newcommand{\E}{\mathbb{E}}

\def\1{1\kern-.20em {\rm l}}

\numberwithin{equation}{section}

\author[Abdelbasset Djeniah]{Abdelbasset Djeniah}
\address{(A. Djeniah) Laboratory of Stochastic Models, Statistics and Applications. University of Saida Dr. Moulay Tahar. B.P. 138, En-Nasr, Saida, 2000 Algeria.}
\email{abdelbasset.djeniah@univ-saida.dz}

\author[Mohamed Chaouch]{Mohamed Chaouch}
\address{(M. Chaouch) Statistics Program, Department of Mathematics and Statistics, College of Arts and Sciences, Qatar University, Doha, Qatar}
\email[Corresponding author]{mchaouch@qu.edu.qa}

\author[Amina Angelika Bouchentouf ]{Amina Angelika Bouchentouf}
\address{(A. A. Bouchentouf) Laboratory of Mathematics, Djillali Liabes University of Sidi Bel Abbes, Sidi Bel Abbes, Algeria}
\email{bouchentouf\_amina@yahoo.fr}

\title[Functional conditional volatility modeling with missing data]{Functional conditional volatility modeling with missing data: inference and application to energy commodities }

\begin{document}
\begin{abstract}
This paper explores the nonparametric estimation of the volatility component in a heteroscedastic scalar-on-function regression model, where the underlying discrete-time process is ergodic and subject to a missing-at-random mechanism. We first propose a simplified estimator for the regression and volatility operators, constructed solely from the observed data. The asymptotic properties of these estimators, including the almost sure uniform consistency rate and asymptotic distribution, are rigorously analyzed. Subsequently, the simplified estimators are employed to impute the missing data in the original process, enhancing the estimation of the regression and volatility components. The asymptotic behavior of these imputed estimators is also thoroughly investigated. A numerical comparison of the simplified and imputed estimators is presented using simulated data. Finally, the methodology is applied to real-world data to model the volatility of daily natural gas returns, utilizing intraday EU/USD exchange rate return curves sampled at a 1-hour frequency.
\end{abstract}
\maketitle

\section{Introduction}\label{intro}

In financial market analysis, understanding and modeling the volatility of financial assets, and predicting it if possible, is of significant interest to investors aiming to make informed decisions. Additionally, financial institutions are not only concerned with short-term forecasting of asset prices but also with quantifying the uncertainty associated with these predictions,  often captured through the volatility component. Two major classes of volatility models introduced in the literature are the Generalized Autoregressive Conditionally Heteroskedastic (GARCH) models and Stochastic Volatility (SV) models. The univariate ARCH model was first proposed by \citet{engle} and subsequently extended to the GARCH model by \citet{boll87}. GARCH models, along with their extensions, have proven to be effective for modeling the conditional volatility of financial returns observed at monthly or higher frequencies, facilitating the study of the intertemporal relationship between risk and expected return. However, early GARCH models, despite their widespread use, have been criticized for their rigidity, particularly when modeling return series over long time spans.

SV models offer an alternative approach by positing that volatility is driven by its own stochastic process. Unlike GARCH models, where the volatility at time  t  is fully determined by past information, SV models treat volatility as a latent random variable. These models provide several advantages, such as offering a natural economic interpretation of volatility and facilitating connections with continuous-time diffusion models. They are also considered more flexible for modeling financial returns (see the handbook by \citet{Andersen2009}).

To relax the parametric assumptions inherent in GARCH models, nonparametric autoregressive models with ARCH-type errors have been proposed (see \citet{LA05} and \citet{FY98}). While these models offer flexibility in capturing nonlinear patterns in volatility, they are limited by the curse of dimensionality. To address this issue, semi-parametric approaches such as additive and single-index volatility models have been introduced. Additive models reduce the dimensionality by assuming that the target volatility function can be expressed as a sum of functions of the covariates, thereby improving convergence rates. Single-index models, on the other hand, simplify conditional variance modeling by reducing multivariate regression to a single index. For further details, readers are referred to \citet{Su}.

A key limitation of the aforementioned approaches is their assumption that the time series is fully observed and that the predictor and response variables are recorded at the same frequency. In practice, even with modern data collection technologies, financial time series often contain missing observations. For instance, stock price data may be missing due to regular holidays (e.g., Thanksgiving, Christmas) or technical issues such as recording device failures or system breakdowns. In financial data analysis, it is commonly assumed that data are completely observed, which may not always be realistic. Addressing the problem of missing data is crucial, particularly when observations in the time series are interrupted. Furthermore, in many scenarios, predictors and responses may be recorded at different time frequencies. For example, one might assess the impact of intraday (1-hour frequency) EU/USD exchange rate returns on the volatility of daily (1-day frequency) natural gas price returns. Figure \ref{sample} illustrates such a scenario, where the response variable (daily natural gas log returns) is scalar-valued, while the predictor (intraday EUR/USD exchange rate log returns at 1-hour frequency) is functional. In these cases, GARCH and SV models are not applicable, making nonparametric functional models particularly relevant.
\begin{figure}[h!]
\begin{center}
\includegraphics[width=15cm,height=5cm]{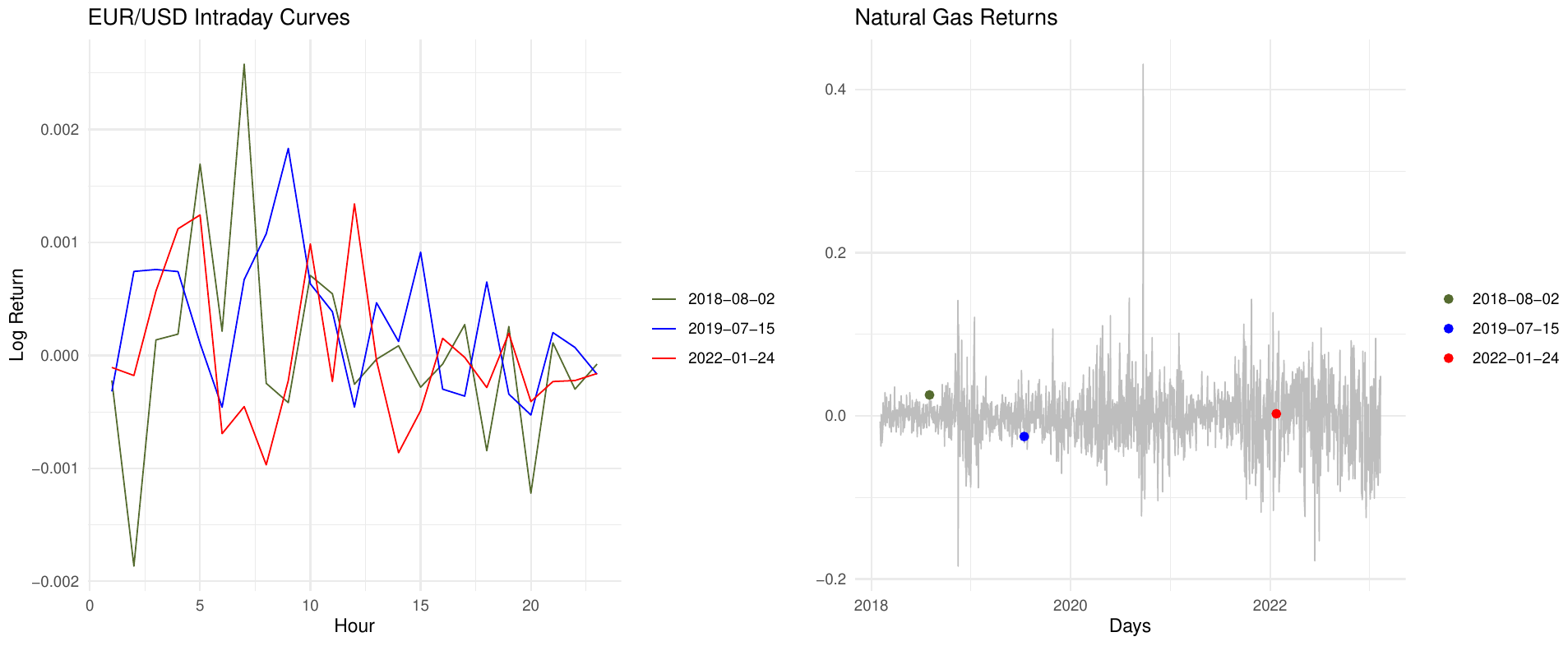}
\end{center}
\caption{(a) Sample curves of intraday (1-hour frequency) EUR/USD exchange rate log returns. (b) Time series of daily natural gas log returns, with dots indicating the corresponding values for the three preselected days.}
\label{sample}
\end{figure}
Over the last two decades, functional data analysis has gained substantial attention due to its wide-ranging applications in fields such as engineering, biology, medicine, climatology, and economics. Early research in this domain focused on parametric (see \citet{Bosq2000}; \citet{Ramsay2002,Ramsay2005}), nonparametric (see \citet{Ferraty2006}; \citet{Geenens2011}; \citet{Ling2018}; \citet{Chaouch2013,Chaouch2015,Chaouch2016}), and semi-parametric models (see \citet{Goia2014}; \citet{Vieu2018}; \citet{Chaouch2020} and others). More recently, nonparametric functional models have gained popularity in econometric studies. For instance, \citet{Ferraty2016} estimated two risk measures, namely value-at-risk and expected shortfall, conditioned on a functional variable. \citet{muller11} introduced a functional volatility process for modeling volatility trajectories in high-frequency financial markets. \citet{Hormann13} proposed a functional ARCH model for high-resolution tick data, modeled as a continuous-time process. \citet{Wang2014} applied functional principal component analysis to identify relevant patterns in the Shanghai Stock Exchange 50 index, while \citet{caldeira} used nonparametric functional methods to forecast the US term structure of interest rates. Recently, \citet{Chaouch2019} investigated volatility estimation in a heteroscedastic scalar-on-function regression model with complete data.

\subsection{Contributions and organization of the paper}
This paper generalizes the results of \citet{perez2010} to accommodate infinite-dimensional predictors and residual-based estimator. Unlike \citet{LING15}, this study considers a heteroscedastic functional regression model in which both the regression and conditional variance operators must be estimated under the missing-at-random assumption. 

Initially, we define simplified estimators using the available data, excluding contributions from missing observations. These simplified estimators are subsequently used to impute the missing data in the original process. In contrast to \citet{LING15}, we employ the initial estimator to impute the missing data and rigorously study the properties of the imputed estimator, both asymptotically and numerically. Finally, we re-estimate the model parameters using the imputed data and compare the results with the simplified estimators. The asymptotic properties of both simplified and imputed estimators are rigorously analyzed, including pointwise and uniform consistency rates, as well as their asymptotic distributions. Confidence intervals are also derived.

The paper is organized  as follows. Section \ref{setting} introduces the heteroscedastic scalar-on-function regression model along with the missing data mechanism and data dependence structure. Section \ref{estimators} defines the simplified and nonparametric imputed estimators. Section \ref{main} outlines the main assumptions of the study and discusses the asymptotic properties of the estimators, including uniform almost-sure convergence rates and asymptotic distributions. To validate these results, simulations and real-data analyses are provided in Sections \ref{sec4} and \ref{sec5}, respectively. Finally, Section \ref{sec6} discusses the findings and offers perspectives for future research. Technical proofs of the main asymptotic results are detailed in the Appendix.

\subsection{Notations}
We define $o_{a.s.}(v)$ as a real random function $z$ such that $z(v)/v$ converges to zero almost surely as $v\rightarrow 0$. Similarly, $\mathcal{O}_{a.s.}(v)$ is defined  as a real random function $z$ such that $z(v)/v$ is almost surely bounded. Throughout this work, $C$ denotes a positive generic constant that may vary in value. The notation $\overset{\mathcal{D}}{\longrightarrow}$ is used to indicate convergence in distribution, and the normal distribution is denoted by $\mathcal{N}(\cdot, \cdot).$

\section{Settings}\label{setting}

In this section, we present the heteroscedastic scalar-on-function regression model, define the missing-at-random (MAR) assumption, and outline the ergodic assumption, which we assume is satisfied by our functional data.

\subsection{Model}

Let $(X_t, Y_t)_{t=1, \dots, n}$ be a sample of discrete time ergodic random processes taking values in $\mathcal{E}\times \mathbb{R}$ and distributed as $(X,Y)$. Here, $Y$ represents the variable of interest, and $X$ is a functional covariate taking values in an infinite-dimensional space $\mathcal{E}$ equipped with a semi-metric $d(\cdot, \cdot)$\footnote{A semi-metric (sometimes called pseudo-metric) $d(\cdot, \cdot)$ is a metric which allows $d(x_{1}, x_{2}) = 0$ for some $x_{1}\neq x_{2}$.}, which defines a topology to measure the proximity between two elements of $\mathcal{E}.$ This semi-metric is disconnected of the definition of $X$ in order to avoid measurability problems. The data generation process is assumed to follow the heteroscedastic functional regression model:
 \begin{eqnarray}\label{model}
 Y_t = m(X_t) + \sqrt{U(X_t)}\;\varepsilon_t, \quad\quad t=1, \dots, n,
 \end{eqnarray}
 where $m(\cdot) = \mathbb{E}(Y | X=\cdot)$ is the regression operator and $U(\cdot) = \text{var}(Y|X=\cdot)$ is the conditional variance operator, both of which are assumed to be unknown.  
 
 The sequence $\varepsilon_1, \varepsilon_2, \ldots$ is assumed to form a martingale difference sequence, satisfying:
 \begin{align}\label{epsilon}
 \E(\varepsilon_t | \G_{t-1})=0\;\text{ a.s.}\quad \mbox{and}\quad \mbox{var}(\varepsilon_t | \G_{t-1}) = 1\;\; \mbox{a.s.},
 \end{align}
where $\G_{t-1}$ is the $\sigma$-field generated by $\{(X_1, Y_1), \dots, (X_{t-1}, Y_{t-1}), X_t \}$. Additionally, let $\F_{t-1}$ denote the $\sigma$-field generated by $\{(X_1, Y_1), \dots, (X_{t-1}, Y_{t-1}) \}.$

It is worth noting that model \eqref{model}, where the response is real-valued and the covariate is infinite dimensional, encompasses several notable volatility models studied in the literature. In the following, we discuss some particular cases.

\vskip 2mm

\noindent\textit{Case 1 - Parametric autoregressive model with ARCH errors}: Let $\mathcal{E} = \mathbb{R}^d$, where $d\geq 1$, and consider $X_{t-1} \equiv (Y_{t-1}, \dots, Y_{t-d})^\top.$ In this case, model \eqref{model} simplifies to:
\begin{eqnarray*}\label{AR}
Y_t = m(Y_{t-1}, \dots, Y_{t-d}) + U(Y_{t-1}, \dots, Y_{t-d})\;\varepsilon_t.
\end{eqnarray*}
Furthermore, if $m(X_{t-1}) \equiv m(X_{t-1}; {\bm \alpha}) = \sum_{j=1}^q \alpha_j Y_{t-j}$ and $U(X_{t-1}; {\bm \beta}) = 1+\sum_{j=1}^d \beta_j Y_{t-j}^2$, then model \eqref{model} corresponds to the AR-ARCH model introduced in \cite{Bor}.

\vskip 2mm

\noindent\textit{Case 2 - Nonparametric autoregressive model with ARCH errors}: To mitigate potential misspecification of the parametric forms assumed for the regression and volatility functions in the AR-ARCH model, \cite{FY98} proposed a local linear estimator for the conditional variance function within a time series regression framework where $Y\in \mathbb{R}$ and  $\mathcal{E} = \mathbb{R}.$ Subsequently, \cite{LA05} extended this approach by investigating a local constant estimator for the parameters in model \eqref{model} considering the case where $Y\in \mathbb{R}$ and $\mathcal{E}=\mathbb{R}^d.$ 

\vskip 2mm


\vskip 1mm
When both the response variable and the predictor are fully observed, \citet{Chaouch2019} proposed nonparametric estimators for $m(\cdot)$ and $U(\cdot)$. Specifically, for any $x\in \mathcal{E}$, these estimators are defined as follows:
 \begin{equation}\label{NWV}
m_{n,c}(x) = \frac{\displaystyle\sum_{t=1}^n Y_t K\left(h_1^{-1} d_1(x, X_t)\right)}{\displaystyle\sum_{t=1}^n K\left(h_1^{-1} d_1(x, X_t)\right)} \quad\mbox{and} \quad U_{n,c}(x) = \frac{\displaystyle\sum_{t=1}^n \big\{Y_t - m_n(X_t)\big\}^2 W\left(h_2^{-1} d_2(x, X_t)\right)}{\displaystyle\sum_{t=1}^n W\left(h_2^{-1} d_2(x, X_t)\right)},
\end{equation}
where $K$ and $W$ are kernel functions, and the sequences $h_1 := h_{1,n}$ and $h_2 := h_{2,n}$ are positive real numbers that decrease to zero as $n \to \infty$. Two distinct semi-metrics, $d_1$ and $d_2 $, are employed for the regression and conditional variance estimators, respectively, to measure the similarity between two functional random variables $X_t$ and $X_s$ in $\mathcal{E}$, for $t\neq s.$

In practice, however, the response variable may be missing at random due to various factors, including data loss, non-response, or data entry errors. \citet{LING15} studied the estimation of the regression operator $m(\cdot)$ under the MAR assumption, where the predictor is a completely observed functional random variable, and the error term in the model has constant variance. Similarly, \citet{Crambes} investigated regression operator estimation under a homoscedastic functional linear model. To the best of our knowledge, no studies have addressed the estimation of the variance operator  $U(\cdot)$  when the response is missing at random. In finite-dimensional case ($\mathcal{E} = \mathbb{R}^d$, $d\geq 1$), \cite{perez2010} explored nonparametric estimation of the difference-based conditional variance function under the MAR assumption, assuming a fixed design for the fully observed predictor.


 \subsection{Missing at random assumption}\label{sec2}

In practice, it is common for data to be incomplete. In this study, we allow the response variable $Y_t$ to be missing at random at any time $t=1, \dots, n$, while assuming that the predictor $X$ is fully observed. To determine whether an observation is complete or missing, we introduce an indicator function $\delta_t,$ where $\delta_t=1$ if $Y_t$ is observed, and $\delta_t=0$ if $Y_t$ is missing, for any $t=1, \dots, n.$ The Bernoulli random variable $\delta$ is assumed to satisfy:
 \begin{eqnarray}\label{pieq}
 \mathbb{P}(\delta=1 | X=x, Y=y) = \mathbb{P}(\delta=1 | X=x) =: \pi(x).
 \end{eqnarray}
  Here, $\pi: \mathcal{E}\rightarrow [0,1]$ represents the conditional probability of observing the response variable and is typically unknown. This assumption implies that $\delta$ and $Y$ are conditionally independent given $X.$  In other words, assumption \eqref{pieq} states that the response variable $Y$ does not provide additional information, beyond what is already given by the explanatory variable $X$, to predict whether an individual observation will be missing.

\subsection{Ergodic processes: definition and examples}
We now introduce the ergodic condition under which the asymptotic theory of our estimators is developed.

\begin{defn}[\cite{R72}]
Let $\{ Z_n, n\in \mathbb{Z} \}$ be a stationary sequence. Define the backward field $\mathcal{B}_n=
\sigma(Z_k; k\leq n)$ and the forward field $\mathcal{A}_m=\sigma(Z_k; k\geq m)$. The sequence is said to be strongly mixing if
$ \sup_{A\in \mathcal{B}_0, B\in \mathcal{A}_n}\Big| \mathbb{P} (A \cap B)-\mathbb{P} (A) \mathbb{P} (B)\Big|=\varphi(n)\to 0 \quad \mbox{as}\quad n\to \infty.$
The sequence is called ergodic if
 $\lim_{n\to\infty} \frac {1}{n} \sum_{t=0}^{n-1} \Big|\mathbb{P} \left(A \cap \tau^{-t}B\right) - \mathbb{P}(A) \mathbb{P}(B)\Big| = 0,$
where $\tau$ represents the time-evolution or shift transformation. 
\end{defn}

The term "strong mixing" in the above definition refers to a more stringent condition than what is typically referred to as strong mixing in the context of measure-preserving dynamical systems. In the latter context, strong mixing is defined as $\lim_{n\to\infty} \mathbb{P} (A \cap \tau^{-n}B) = \mathbb{P}(A)\mathbb{P}(B)$  for any two measurable sets $A, B$ (see \cite{R72}). Consequently, strong mixing implies ergodicity; however, the converse is not true. There exist ergodic sequences that are not strong mixing.

The ergodicity condition is a natural and less restrictive assumption compared to any form of mixing conditions, under which standard nonparametric estimators (e.g., density and regression estimators) are consistent. Ergodicity can be viewed as a condition ensuring the validity of the law of large numbers. Specifically, the ergodic theorem states that for a stationary ergodic process $Z$, we have
$
\lim_{n\to\infty}\frac{1}{n}\sum_{t=1}^n Z_t=\mathbb{E}(Z_1), \ \mbox{almost surely (a.s.)}.$ 
For further details and results on ergodic theory, we refer the reader to the book by \citet{K85}. In the following, we present examples of ergodic processes that do not satisfy the commonly used mixing conditions.

\begin{exmp}
\textit{k-factor Gegenbauer process}

Let $(\epsilon_t)_{t\in\mathbb{Z}}$ be a white noise process with variance $\sigma^2$. The k-factor Gegenbauer process is defined as $\prod_{i=1}^k (I - 2\varpi_i B + B^2)^{d_i} X_t = \epsilon_t,
$ where $I$ is the identity operator, $B$ is the backshift operator, $0 < d_i < \frac{1}{2}$ if $|\varpi_i| < 1$, or $0 < d_i < \frac{1}{4}$ if $|\varpi_i| = 1$, for $i = 1, \ldots, k$.
Here, $\varpi_i$ denotes the Gegenbauer frequencies that control the location of spectral poles, while $d_{i}$ determines the persistence strength at each frequency.
\cite{Giraitis1995} demonstrated that this process exhibits long-memory, is stationary, causal, and invertible. Furthermore, under specified conditions on $d_i$ and $\varpi_i$, the process can be equivalently expressed in a moving average (MA) form $ X_t = \sum_{j=0}^{\infty} \psi_j \epsilon_{t-j}, \quad \text{where } \sum_{j=0}^{\infty} \psi_j^2 < \infty.
$ Here, $\psi_j$ represents the moving average coefficients, describing the influence of past white noise terms
$\epsilon_{t-j}$ on the current value $X_t.$
Alternatively, \cite{Guegan2001} showed that when $(\epsilon_t)_{t\in\mathbb{Z}}$ is a Gaussian process, the $k$-factor Gegenbauer process is not strong mixing. However, its moving average representation confirms that it remains stationary, Gaussian, and ergodic process.
\end{exmp}
\begin{exmp}
\textit{Non-mixing Markov chain}

Consider the linear Markov AR(1) process $X_t = \frac{1}{2} X_{t-1} + \epsilon_t,$ where $(\epsilon_t)$ are independent, symmetric Bernoulli random variables taking values $-1$ and $1$. As noted by \cite{Andrews1984}, the stationary solution of this process is not $\alpha$-mixing. Nevertheless, the process $(X_t)$ is Markovian, stationary, and ergodic.
\end{exmp}

\section{Definition of the estimators}\label{estimators}
In this section, we define the simplified and imputed estimators of the regression and conditional variance operators. We also investigate their asymptotic properties.

\subsection{Simplified estimator}
By multiplying \eqref{model} by $\delta_t$ we obtain $\delta_t Y_t = \delta_t m(X_t) + \delta_t U^{1/2}(X_t)\varepsilon_t.$
Taking the conditional expectation of both sides given $X_t=x$, and using \eqref{pieq}, we have:
\begin{eqnarray*}
\mathbb{E}\big(\delta_t Y_t | X_t =x\big) &=& \mathbb{E}\big(\delta_t m(X_t) | X_t=x \big) + \mathbb{E}\big(\delta_t U^{1/2}(X_t)\varepsilon_t | X_t=x \big)\\
&=& m(x) \mathbb{E}(\delta_t |X_t=x).
\end{eqnarray*}
Thus, the regression operator for the case where the response variable is missing at random can be expressed, for any $x\in \mathcal{E},$ as:
\begin{eqnarray}\label{reg}
m(x) = \dfrac{\mathbb{E}\left(\delta_t Y_t | X_t =x\right) }{\mathbb{E}(\delta_t |X_t=x)}.
\end{eqnarray}
Similarly, consider the expression $\delta_t \left(Y_t - m(X_t) \right)^2=  \delta_t U(X_t)\,\varepsilon_t^2.$
Taking the conditional expectation given $X_t=x$ on both sides, we get:
$$
\mathbb{E}\left(\delta_t \left(Y_t - m(X_t) \right)^2 | X_t=x \right) = U(x)\, \mathbb{E}(\delta_t |X_t=x).
$$
Consequently, the variance operator can be written, for any $x\in \mathcal{E},$ as:
\begin{eqnarray}\label{var}
U(x) = \dfrac{\mathbb{E}\left(\delta_t \left(Y_t - m(X_t) \right)^2 | X_t=x \right)}{\mathbb{E}(\delta_t |X_t=x)}.
\end{eqnarray}

Given a random sample $(X_t, Y_t, \delta_t)_{t=1, \dots, n}$ we define the simplified estimators of $m(x)$ and $U(x)$, as given in \eqref{reg} and \eqref{var}, respectively, as follows:
 \begin{eqnarray*}
m_{n,0}(x) = \frac{\displaystyle\sum_{t=1}^n \delta_t\,Y_t K\left(h_1^{-1} d_1(x, X_t)\right)}{\displaystyle\sum_{t=1}^n \delta_t\,K\left(h_1^{-1} d_1(x, X_t)\right)} \quad\mbox{and} \quad U_{n,0}(x) = \frac{\displaystyle\sum_{t=1}^n \delta_t\,\big(Y_t - m_{n,0}(X_t)\big)^2 W\left(h_2^{-1} d_2(x, X_t)\right)}{\displaystyle\sum_{t=1}^n \delta_t\,W\left(h_2^{-1} d_2(x, X_t)\right)}.
\end{eqnarray*}

\subsection{Nonparametric imputed estimator}
To construct a nonparametric imputed estimator of the regression operator, we first impute the missing values in the original response process using the initial estimator. Specifically, the imputed response is defined as $\widehat{Y}_t = \delta_t Y_t + (1-\delta_t)m_{n,0}(X_t).$ Using the imputed sample $( X_t, \widehat{Y}_t)_{t=1, \dots, n},$ a nonparametric imputed estimator of $m(x)$ is obtained by substituting $Y_t$ with $\widehat{Y}_t$ in \eqref{NWV}. Consequently, for any fixed $x \in \mathcal{E},$ the imputed estimator is given by:
 \begin{eqnarray}\label{imputedm}
m_{n,1}(x) = \frac{\displaystyle\sum_{t=1}^n \widehat{Y}_t K\left(h_1^{-1} d_1(x, X_t)\right)}{\displaystyle\sum_{t=1}^n K\left(h_1^{-1} d_1(x, X_t)\right)}.
\end{eqnarray}
 An imputed estimator of the conditional variance can be constructed in two steps. First, the missing residuals are imputed nonparametrically as follows: $\widehat{r}_t = \delta_t r_t + (1-\delta_t)U_{n,0}(X_t)$, where $r_t = (Y_t - m_{n,0}(X_t))^2$ when $\delta_t=1$ and $r_t$ is unobserved when $\delta_t=0$. Second, using the sample $(X_t, \widehat{r}_t)_{t=1, \dots, n}$, a nonparametric imputed estimator of $U(x)$ is defined, for any $x\in \mathcal{E},$ as:
 \begin{eqnarray}\label{imputedv}
 U_{n,1}(x) =\frac{\displaystyle\sum_{t=1}^n \widehat{r}_t W\left(h_2^{-1} d_2(x, X_t)\right)}{\displaystyle\sum_{t=1}^n W\left(h_2^{-1} d_2(x, X_t)\right)}.
 \end{eqnarray}
 
\section{Assumptions and main results}\label{main}
 For  $k\in \{1, 2\}$ , let  $F_{x,k\color{black}}(u)=\mathbb{P}(d_{k\color{black}}(x, X) \leq u)=\mathbb{P}(X \in$ $B(x, u))$ and $F_{x,k}^{\mathcal{F}_{t-1}}(u)=\mathbb{P}\left(d_{k}(x, X) \leq u \mid \mathcal{F}_{t-1}\right)=\mathbb{P}\left(X \in B(x, u) \mid \mathcal{F}_{t-1}\right)$ represent the marginal distribution and the conditional marginal distribution of $X$ given the $\sigma$-field $\mathcal{F}_{t-1},$ respectively.

Our primary results involve establishing the  almost sure uniform consistency, along with convergence rates, for both the simplified and the imputed estimators of the regression operator and the conditional variance.  To this end, let $\mathcal{C}$ denote a class of elements in the functional space $\mathcal{E},$ and consider, for any $\eta > 0$,
$$\begin{array}{ll}
 N\left(\eta, \mathcal{C}, d_{\mathcal{C}}\right)=\min \{n: & \text{there exist} \quad   c_{1}, \ldots, c_{n} \in \mathcal{C} \quad \text{such that} \quad  \forall x \in \mathcal{C} \\  & \text{there exists}  \quad k \in\{1, \ldots ., n\} \quad  \text{such that} \quad   d_{\mathcal{C}}\left(x, c_{k}\right)<\eta \},
\end{array}$$
a number which measures how full is the class $\mathcal{C}.$ Our results are established under the assumptions listed below. In the subsequent study, the notation $\mathcal{K}$ is used to refer to either the kernel $K$ or $W.$
\begin{itemize}
\item[\textbf{(A1)}]
\begin{itemize}
\item [(i)] $\mathcal{K}$ is a nonnegative bounded kernel of class $\mathcal{C}^{1}$ over its support $[0, 1]$ with $\mathcal{K}(1)>0$. The derivative $\mathcal{K}'$ exists on $[0, 1)$ and satisfies the condition
$\mathcal{K}^{\prime}(v)<0 \text { for all } v \in[0,1) \text { and }$ \begin{center}
$\left|\displaystyle\int_{0}^{1}\left(\mathcal{K}^{j}\right)^{\prime}(v) d v\right|<\infty \text { for } j \in \{1,2\}.$ \end{center}
\item [(ii)]$\mathcal{K}$ is a H\"{o}lder function of order $ \gamma $ with a constant $a_{0}$.
\item [(iii)] There exist constants $a_{1}$ and $a_{2}$ such that $0 < a_{1} \leq \mathcal{K}(v) \leq a_{2} < \infty$ for all $v \in \mathcal{C}$.
\end{itemize}
\item[\textbf{(A2)}]
For $x \in \mathcal{E}$, there exists a sequence of nonnegative  random variables $(f_{t,1})_{t\geq 1} $  almost surely bounded
by a sequence of deterministic quantities ($b_{t}(x))_{t\geq 1}$, a sequence of random functions
$(\psi_{t,x})_{t\geq 1} $, a deterministic nonnegative bounded function $f_{1}$ and a nonnegative real function $\phi(.)$ that tend to zero, as its
argument tends to 0,
\begin{itemize} 
 \item[(i)] $F_{x,k}(u)=\phi_{k}(u) f_{1}(x)+o(\phi_{k}(u)) \text { as } u \rightarrow 0$ and $k\in \{1, 2\},$ where $o(\phi_{k}(u))$ is uniform in $x$. 
\item [(ii)] For $k\in \{1, 2\} $ and $\forall  t \in N, F_{x,k}^{\mathcal{F}_{t-1}}(u)=\phi_{k}(u) f_{t, 1}(x)+\psi_{t, x}(u)$ with $\psi_{t, x}(u)=o_{a.s}(\phi_{k}(u))$ as $u \rightarrow 0, \frac{\psi_{t,x}(u)}{\phi_{k}(u)}$ almost surely bounded for any $x \in \mathcal{C}$ and  $n^{-1} \displaystyle\sum_{t=1}^{n} \psi_{t . x}(u)=o_{a.s}\left(\phi_{k}(u)\right)$ as $n \rightarrow \infty$ and $u \rightarrow 0$.
where $o_{a . s}\left(\phi_{k}^{j}(u)\right)$ is uniform in $x$.
\item [(iii)] $\displaystyle{\lim_{n \rightarrow \infty} \sup_{x\in \mathcal{C}}}\left|f_{1}(x)-n^{-1} \sum_{t=1}^{n} f_{t, 1}(x)\right|=0$ almost surely.
\item [(iv)] For $k \in \{1, 2\} $, there exists a nondecreasing bounded function $\tau_{0,k}$ such that, uniformly in $u \in[0,1]$, $\frac{\phi_{k}(h u)}{\phi_{k}(h)}=\tau_{0,k}(u)+o(1)$, as $h \downarrow 0$, and we have
\begin{center}
$\displaystyle\int_{0}^{1}\left({\mathcal{K}}^{j}(u)\right)^{\prime} \tau_{0,k}(u) d t<\infty$  for  $j \geq 1$.
\end{center}
\item [(v)] $n^{-1} \sum_{t=1}^{n} b_{t}(x) \rightarrow D(x)<\infty$ as $n \rightarrow \infty$, and $0 <\displaystyle{\sup _{x \in \mathcal{C}}} D(x)<\infty$.%
\item [(vi)] $\displaystyle{0<\theta_{0}\leq \inf_{x\in \mathcal{C}}f_{1}(x)\leq \sup_{x\in \mathcal{C}}f_{1}(x)<\infty}$ for some nonnegative real number $\theta_{0}$.
\item [(vii)] $\displaystyle{\inf_{x\in \mathcal{C}}}\pi(x)> \theta_{1}$ for some positive real number $\theta_{1} \in [0,1].$
\end{itemize}
 \item[\textbf{(A3)}]
\begin{itemize}
\item[(i)] There exists $\rho > 1$ such that  $\mathbb{E}\{U(X_{t})^{\rho^{2}/2(\rho-1)}  \}<\infty $ and $\underset{1\leq t\leq n}{\max}\mathbb{E}(|\varepsilon_{t}|^{\rho}|\G_{t-1})<\infty$.
\item[(ii)] $\mathbb{E}\{(\varepsilon^{2}_{t}-1)^{2}|X_{t})=\mathbb{E}\{(\varepsilon^{2}_{t}-1)^{2}|\G_{t-1}\}=\omega(X_{t}) $ is continuous in a neighborhood of $x$ as $h \rightarrow 0,$  that is
\begin{equation*}
\sup _{\{u: d(x, u) \leq h\}}|\omega(u)-\omega(x)|=o(1).
\end{equation*}
\item[(iii)] For some $\kappa  > 0$, $\underset{1\leq t\leq n}{\max} \mathbb{E}\{(\varepsilon^{2}_{t}-1)^{2+\kappa }|\G_{t-1}\} < \infty$.
\item[(iv)] $\mathbb{E}\{|U(X_{t})|^{\rho^{2}/(\rho-1)}\}<\infty $ and $\underset{1\leq t\leq n}{\max}\mathbb{E}(|\varepsilon_{t}^{2}-1|^{\rho\color{black}}|\G_{t-1})<\infty$.
\item[(v)] For any $t\geq 1$, $\mathbb{E}(\delta_{t}|\G_{t-1})= \mathbb{E}(\delta_{t}|X_{t})=\pi(X_{t}) $ is continuous in a neighborhood of $x$ as $h \rightarrow 0,$ that is
\begin{equation*}
\sup_{\{u: d(x, u) \leq h\}}|\pi(u)-\pi(x)|=o(1).
\end{equation*}
\end{itemize}
\item[\textbf{(A4)}]
\begin{itemize}
\item[(i)] For any $(u, v)\in \mathcal{E}^{2}$, $|m(u)-m(v)|<c_{1}d_{1}^{\alpha}(u,v)$, for some  constant $c_{1}>0$ and $\alpha > 0,$
\item [(ii)] For any $(u, v)\in \mathcal{E}^{2}$, $|U(u)-U(v)|<c_{2}d_{2}^{\beta}(u,v)$, for some  constant $c_{2}>0$ and $\beta > 0.$
\item[(iii)] $U^{\kappa+2}(.)$ is continuous in a neighborhood of $x$ as $h \rightarrow 0$, that is
\begin{equation*}
\sup _{\{u: d(x, u) \leq h\}}|U^{\kappa+2}(u)-U^{\kappa+2}(x)|=o(1).
\end{equation*}
\end{itemize}

\end{itemize}

Assumption (A1)(i) pertains to the choice of the kernel $\mathcal{K}$, a standard requirement in nonparametric functional estimation.  It is important to note that the Parzen symmetric kernel is not suitable in this context because the random process $d(x,X_t)$ is positive. Consequently, we restrict $\mathcal{K}$ to have a support on $[0, 1]$. This represents a natural generalization of the typical assumption in the multivariate setting, where $\mathcal{K}$ is often taken to be a spherically symmetric density function. The assumption $\mathcal{K}(1)>0$ and $\mathcal{K}^\prime <0$ ensure that $M_{1,W,2} >0$ for all limiting functions $\tau_0.$ The requirement $\mathcal{K}(1) >0$ is specifically necessary to define the moments $M_{j,W,2}$, which, in this context, depend on the value $\mathcal{K}(1).$ Assumption (A1)(ii) introduces a H\"older-type condition, imposing a certain degree of smoothness on the kernel, which is common in nonparametric setting. Assumption (A1)(iii) requires that the kernel $\mathcal{K}$ is bounded away from zero, a relatively standard condition in nonparametric functional data estimation.

Conditions (A2)(i)-(ii) reflect the ergodicity property assumed for the discrete-time functional process, which is crucial for establishing the asymptotic properties of the estimator. The functions $f_{t,1}$ and $f_1$ serve a similar purpose to the conditional and unconditional densities in finite dimensional case, while $\phi(u)$ characterizes the effect of the radius $u$ on the small ball probability as $u \rightarrow 0$. Several examples of processes satisfying these conditions can be found in \cite{LL10}. Conditions (A2)(iii) and (A2)(v) are primarily formulated to ensure the applicability of the ergodic Theorem. Assumptions (A3)(i), (A3)(iii) and (A3)(iv) are technical and impose bounds on higher-order moments of the errors and the conditional variance. Assumptions (A3)(ii) and (A3)(v) require the continuity of the operators $\pi(\cdot)$ and $\omega(\cdot).$ Finally, assumption (A4) imposes smoothness conditions on the regression and conditional variance operators.

\subsection{Asymptotic properties of the simplified estimator}

In this subsection, we examine the asymptotic properties of the simplified estimators, including their uniform consistency rate, asymptotic distribution, and the construction of confidence intervals.
\subsubsection{Uniform consistency}
\begin{thm}\label{T1}
Suppose that assumptions  (A1), (A2), (A3)(i),(v), (A4)(i) hold true, along with the following conditions:
\begin{equation}\label{1}
 \lim _{n \rightarrow \infty} n \phi_{1}\left(h_{1}\right)=\infty \,\ \hbox{and} \,\
 \lim _{n \rightarrow \infty}\left(n \phi_{1}\left(h_{1}\right)\right)^{-1} \log n=0
\end{equation}
Additionally, for a sequence of positive real numbers $\lambda_{n}$ tending to zero, as $n \rightarrow \infty,$ and  $\eta=\eta_{n}= o(h_{1})$, assume that the following conditions are satisfied:
\begin{equation}\label{2}
\begin{aligned}
\lim _{n \rightarrow \infty} \frac{\log N\left(\eta, \mathcal{C}, d_{\mathcal{C}}\right)}{n\lambda_{n}^{2} \phi_{1}\left(h_{1}\right)\ell_{n}^{-2}}=0 \,\ &  \text{and} &
 \sum_{n \geq 1} \exp \left[-\lambda_{n}^{2}\ell_{n}^{-2} \mathcal{O}\left(n \phi_{1}\left(h_{1}\right)\right)\right]<\infty\text{,}
\end{aligned}
\end{equation}
where $\ell_{n}$ is a sequence of positive numbers tending to infinity as $n \rightarrow \infty,$ defined by\\ $\ell_{n}=\left(\frac{\log n}{\lambda_{n}\phi_{1}(h_{1})^{(\rho-1)/\rho}}\right)^{1/(\rho-1)},$ and
$\rho$ is specified in (A3)(i).
Then, the following result holds:
\begin{equation*}
\sup _{x \in \mathcal{C}}\left|m_{n,0}(x)-m(x)\right|=\mathcal{O}_{a . s .}\left(h_{1}^{\alpha}\right)+\mathcal{O}_{a \cdot s .}\left(\lambda_{n}\right),
\end{equation*}
where $\alpha$ is defined in (A4)(i).
\end{thm}
\begin{rem}
 The uniform consistency rate for the regression operator is identical to that derived in \cite{Chaouch2019} when the response variable is fully observed. Moreover, the same rate has been obtained in \cite{LL10} for the homoscedastic scalar-on-function regression model. Additionally, if $\lambda_n = \mathcal{O}\left( \sqrt{ \log n/(n\phi_1(h_1)})\right)$, condition (\ref{2}) is satisfied, and consequently, the uniform consistency rate for $m_{n,0}(x)$ becomes $\mathcal{O}(h_1^\alpha) + \mathcal{O}\left( \sqrt{ \log n/(n\phi_1(h_1)})\right).$ This matches the one established in \cite{Ferraty2006} for independent data.
\end{rem}
\begin{thm}\label{T2}
Assume that assumptions  (A1), (A2), (A3)(i),(iv)-(v),  (A4)(i)-(ii),  conditions (\ref{1})-(\ref{2}), and the following conditions hold:
\begin{equation}\label{3}
 \lim _{n \rightarrow \infty} n \phi_{2}\left(h_{2}\right)=\infty \,\ \hbox{and} \,\
 \lim _{n \rightarrow \infty}\left(n \phi_{2}\left(h_{2}\right)\right)^{-1} \log n=0.
\end{equation}
Furthermore, for a sequence of positive real numbers $\lambda_{n}^{\prime}$ tending to zero as $n \rightarrow \infty$, suppose the following conditions are satisfied:
\begin{equation}\label{4}
\begin{aligned}
\lim _{n \rightarrow \infty} \frac{\log N\left(\eta, \mathcal{C}, d_{\mathcal{C}}\right)}{n\left(\lambda_{n}^{\prime}\right)^{2} \phi_{2}\left(h_{2}\right)\left(\ell_{n}^{\prime}\right)^{-2}}=0 & & \text{and} & &
 \sum_{n \geq 1} \exp \left[-\left(\lambda_{n}^{\prime}\right)^{2}\left(\ell_{n}^{\prime}\right)^{-2} \mathcal{O}\left(n \phi_{2}\left(h_{2}\right)\right)\right]<\infty,
\end{aligned}
\end{equation}
where $\ell_{n}^{\prime}$ is a sequence of positive numbers tending to $\infty$, as $n \rightarrow \infty,$ defined as\\ $\ell_{n}^{\prime}=\left(\frac{\log n}{\lambda_{n}^{\prime}\phi_{2}(h_{2})^{(\rho-1)/\rho}}\right)^{1/(\rho-1)},$ with
$\rho$ specified in (A3)(i).  Then, the following result holds:
\begin{equation*}
\sup _{x \in \mathcal{C}}\left|U_{n,0}(x)-U(x)\right|=\mathcal{O}_{a . s .}\left(h_{1}^{2\alpha}+h_{2}^{\beta}\right)+\mathcal{O}_{a \cdot s .}\left(\lambda_{n}^{\prime}+\lambda_{n}^2\right),
\end{equation*}
where $\alpha$ and $\beta$ are given in (A4)(i) and (A4)(ii), respectively.
\end{thm}
\begin{rem}
The initial estimator $U_{n,0}(x)$ exhibits the same uniform almost sure consistency rate as that derived in \cite{Chaouch2019}, where the response variable is fully observed. Furthermore, by setting $h_1 = h_2 = h$, $d_1(\cdot, \cdot) = d_2(\cdot, \cdot)$, $\alpha = \beta$, $\phi_1(\cdot) = \phi_2(\cdot) = \phi(\cdot)$, and $\lambda_n = \lambda_n' = \mathcal{O}\left( \sqrt{ \log n/(n\phi(h)})\right)$, condition  \eqref{3} is satisfied. Consequently, the uniform convergence rate of  $U_{n,0}(x)$ is given by $\mathcal{O}(h^\beta) + \mathcal{O}\left( \sqrt{ \log n/(n\phi(h)})\right)$. In the finite-dimensional case ($\mathcal{E}=\mathbb{R}^d$), the rate simplifies to $\mathcal{O}(h^\beta) + \mathcal{O}\left( \sqrt{ \log n/n h^d}\right),$ which aligns with the rate established in \cite{LA05} for a nonlinear autoregressive model with ARCH errors.
\end{rem}

\subsubsection{Asymptotic distribution}
\begin{thm}\label{T3}
Assume that assumptions (A1)-(A4) and conditions (\ref{1})-(\ref{4}) hold. Additionally, suppose that the following conditions are satisfied:
\begin{eqnarray}
    \label{extracond}
    \sqrt{n\phi_2(h_2)} h_1^{2\alpha} \rightarrow 0, \;  \sqrt{n\phi_2(h_2)} \lambda_n^2 \rightarrow 0, \; \text{and}\; \sqrt{n\phi_2(h_2)} \lambda_n^\prime \rightarrow 0\; \text{as}\; n\rightarrow \infty.
\end{eqnarray}
Then, for any fixed $x\in \mathcal{E}$ such that $f_1(x)>0$, the following results hold:
\begin{itemize}
\item[$(i)$] $ \sqrt{n \phi_{2}\left(h_{2}\right)}\left(U_{n,0}(x)-U(x) - B_{n,0}(x)\right) \overset{\mathcal{D}}
{\longrightarrow} \mathcal{N}\left(0, \sigma_{0}^{2}(x)\right),
$
\end{itemize}
where the ``bias'' term $B_{n,0}(x) = \mathcal{O}(h_2^\beta)$ with $\beta$ as defined in assumption (A4)(ii). The asymptotic conditional variance is $\sigma_{0}^{2}(x):=\frac{M_{2, W, 2} U^{2}(x)\omega(x)}{M_{1, W, 2}^{2}  
 \pi(x)f_{1}(x)}$ and 
 \begin{eqnarray}\label{MJ}
 M_{j, W, 2}=W^{j}(1)-\int_{0}^{1}\left(W^{j}\right)^{\prime}(u) \tau_{0,2}(u) d u ,\,\ for \,\ j\in \{1,2\}.
 \end{eqnarray}
\begin{itemize}
\item[$(ii)$] If, in addition, $\sqrt{n\phi_2(h_2)} h_2^\beta \rightarrow 0\; \text{as}\; n\rightarrow \infty,$ then the estimator satisfies: 
$$ \sqrt{n \phi_{2}\left(h_{2}\right)}\left(U_{n,0}(x)-U(x)\right) \overset{\mathcal{D}}
{\longrightarrow} \mathcal{N}\left(0, \sigma_{0}^{2}(x)\right).
$$
\end{itemize}
\end{thm}

Theorem \ref{T3} extends the asymptotic distribution result established in Theorem 3 in \citet{Chaouch2019}, which considers the case where data are fully observed. Notably, the asymptotic conditional variance $\sigma_0^2(x)$ depends on the conditional probability of observing data, $\pi(x)$. A higher (resp. lower) missing-at-random rate - corresponding to smaller (resp. larger) $\pi(x)$ - results in a larger (resp. smaller) value of $\sigma_0^2(x)$. Consequently, the efficiency of the initial estimator decreases (resp., increases) with an increasing MAR rate. If the data are fully observed (i.e. $\pi(x) = 1, \forall x\in\mathcal{E}$), the asymptotic conditional variance  $\sigma_0^2(x)$ reduces to the form derived in  \citet{Chaouch2019}, highlighting the optimal efficiency under a complete data scenario. 
 


\subsubsection{Asymptotic confidence intervals}
Our objective here is to construct asymptotic confidence intervals for $U(x)$ for any fixed curve $x \in \mathcal{E} $ based
on a normal approximation. It is important to note that the asymptotic variance in Theorem \ref{T3} involves several unknown
quantities: $\pi(\cdot),$ $\omega(\cdot),$  $U(\cdot),$ $M_{1, W, 2},$ $M_{2, W, 2}$ and $\tau_{0,2}(u).$ These parameters are replaced with their empirical counterparts as follows. The conditional variance $U(x)$ is replaced by its estimator $U_{n,0}(x)$, while $\omega(x)$ and  $\pi(x)$ are estimated using:
\begin{equation}\label{w_n}
 \omega_{n,0}(x)=\frac{\displaystyle\sum_{t=1}^{n}\delta_t(\hat{\varepsilon}^{2}_{t,0}-1)^{2}H\left(h_3^{-1} d_{3}(x,X_{t})\right)}{\displaystyle\sum_{t=1}^{n} \delta_t H\left(h_3^{-1} d_{3}(x,X_{t})\right)},
\end{equation}
where $\hat{\varepsilon}_{t,0}=[Y_{t}-m_{n,0}(X_{t})]/[U_{n,0}(X_{t})]^{1/2},$  and
 \begin{equation}\label{pi_n}
 \pi_{n}(x)=\frac{\displaystyle\sum_{t=1}^{n}\delta_{t}\tilde{H}\left(h_4^{-1}d_{4}(x,X_{t})\right)}{\displaystyle\sum_{t=1}^{n} \tilde{H}\left(h_4^{-1}d_{4}(x,X_{t})\right)}.
\end{equation}
Here, $H$ and $\tilde{H}$ are kernel functions, $d_3(\cdot, \cdot)$ and $d_4(\cdot, \cdot)$ are semi-metrics, and $h_3$ and $h_4$ are bandwidths used to estimate $\omega(x)$ and $\pi(x)$, respectively. Additionally, by applying assumptions (A2)(i) and (A2)(iv), the estimator of $\tau_{0,2}$ is defined as 
$
\widehat{\tau}_{0,2}(u)=\widehat{F}_{x,2}(uh_{2})/\widehat{F}_{x,2}(u),
$
where
$
\widehat{F}_{x,2}(u)=n^{-1}\sum_{t=1}^{n}\mathds{1}_{\{d_{2}(x,X_{t})\leq u\}}.
$
The value of $\tau_{0,2}(u)$ is replaced by its estimator $\widehat{\tau}_{0,2}(u)$ in (\ref{MJ}) to obtain the plug-in estimator $\widehat{M}_{j, W, 2}$ of  $M_{j, W, 2}$, for $j=\{1,2\}.$\\

\begin{cor}\label{Cor}
  Under the conditions of Theorem \ref{T3}, we have:
 \begin{equation}\label{Cor+1}
  \left(\frac{\color{black} \widehat{M}_{1, W, 2}\color{black}}{\color{black}\sqrt{\widehat{M}_{2, W, 2} \color{black}} }  \sqrt{\frac{n \widehat{F}_{x, 2}(h_2)\pi_{n}(x)}{\omega_{n,0}(x)\color{black} U^{2}_{n,0}(x) }}\right) \left(U_{n,0}(x)-U(x)\right) \overset{\mathcal{D}}
{\longrightarrow} \mathcal{N}\left(0, 1\right),\,\ \hbox{as}\,\ n\rightarrow\infty.
\end{equation}
\end{cor}

This result is pivotal for constructing confidence intervals of $U(x)$. Specifically, equation (\ref{Cor+1}) leads to the following asymptotic $100(1-\nu)\%$ confidence interval for the conditional variance $U(x):$
 \begin{equation}\label{CIS}
 \mathrm{CI}^{\text{S}}_{\nu}=U_{n,0}(x)\left(1 \pm q_{\nu / 2} \frac{\sqrt{\widehat{M}_{2, W, 2}}}{\widehat{M}_{1, W, 2}} \sqrt{\frac{\omega_{n,0}(x)}{n \widehat{F}_{x, 2}(h_2)\pi_{n}(x)}}\right) \text {, }
\end{equation}
where $q_{\nu/ 2}$ denotes the upper $\nu / 2$ quantile of the standard normal distribution.

The confidence interval in \eqref{CIS} reveals that as the missing rate increases ($\pi(x)\rightarrow 0$), the interval becomes wider, reflecting reduced precision in the interval estimation.

\subsection{Asymptotic properties of the nonparametric imputed estimator}

In this subsection, we explore the asymptotic properties of the nonparametric imputed estimators for the regression and conditional variance operators.

\subsubsection{Uniform consistency}
\begin{thm}\label{T5}
Under the same conditions stated in Theorem \ref{T1}, we have:
\begin{equation*}
\sup _{x \in \mathcal{C}}\left|m_{n,1}(x)-m(x)\right|=\mathcal{O}_{a . s .}\left(h_{1}^{\alpha}\right)+\mathcal{O}_{a \cdot s .}\left(\lambda_{n}\right),
\end{equation*}
where $\alpha$ is given in (A4)(i).
\end{thm}
\begin{thm}\label{T6}
Under the conditions outlined in Theorem \ref{T2}, we have:
\begin{equation*}
\sup _{x \in \mathcal{C}}\left|U_{n,1}(x)-U(x)\right|=\mathcal{O}_{a . s .}\left(h_{1}^{2\alpha}+h_{2}^{\beta}\right)+\mathcal{O}_{a \cdot s .}\left(\lambda_{n}^{\prime}+\lambda_{n}^2\right),
\end{equation*}
where $\beta$ is as defined in (A4)(ii).
\end{thm}
\begin{rem} The almost sure uniform convergence rates for the estimators $m_{n,1}(x)$ and $ U_{n,1}(x)$ are identical to those derived in Theorems \ref{T1} and \ref{T2}, respectively.
\end{rem}

\subsubsection{Asymptotic distribution}
\begin{thm}\label{T7}
Suppose the assumptions of Theorem \ref{T3} are satisfied, the conditions in \eqref{extracond} hold, and $\sqrt{n\phi_2(h_2)}h^\beta \rightarrow 0$, as $n\rightarrow\infty.$  Then, the following result holds:
$$ \sqrt{n \phi_{2}\left(h_{2}\right)}\left(U_{n,1}(x)-U(x)\right) \overset{\mathcal{D}}
{\longrightarrow} \mathcal{N}\left(0, \sigma^{2}_1(x)\right),
$$
where 
$
\sigma^{2}_1(x)=\frac{M_{2, W, 2} U^{2}(x)\omega(x)\pi(x)}{M_{1, W, 2}^{2}  f_{1}(x)}=\sigma_{0}^{2}(x)+\frac{M_{2, W\color{black}, 2} U^{2}(x)(\pi^2(x)-1)\omega(x)}{\pi(x)M_{1, W, 2}^{2}  f_{1}(x)},$ and $\sigma_0^2(x)$ is as defined in Theorem \ref{T3}.
\end{thm}
When $\pi(x) \rightarrow 1$ (i.e. low MAR rate) $\sigma^{2}_1(x)$ and $\sigma_0^2(x)$ converge to the asymptotic conditional variance obtained in \citet{Chaouch2019} for the case of complete data. Conversely, under high MAR rate ($\pi(x) \rightarrow 0$), $\sigma_0^2(x) \rightarrow \infty$ while $\sigma^{2}_1(x)\rightarrow 0.$ This indicates that the nonparametric imputed estimator is more efficient that the simplified estimator when the missing rate is high.

\subsubsection{Asymptotic confidence intervals}
Similar to the simplified estimator, the following corollary is used to construct asymptotic confidence intervals for $U(x).$

To achieve this, we define the imputed estimator of $\omega(x)$. For any fixed $x\in \mathcal{E}$, the estimator is given by:
\begin{equation}
\omega_{n,1}(x) = \frac{\displaystyle\sum_{t=1}^{n} (\hat{\varepsilon}^{2}_{t,1}-1)^{2}H\left(h_3^{-1} d_{3}(x,X_{t})\right)}{\displaystyle\sum_{t=1}^{n}  H\left(h_3^{-1} d_{3}(x,X_{t})\right)},
\end{equation}
where $\hat{\varepsilon}_{t,1}=[\widehat{Y}_{t}-m_{n,1}(X_{t})]/[U_{n,1}(X_{t})]^{1/2}.$ 

\begin{cor}\label{CorIm}
  Under the conditions of Theorem \ref{T7}, the following result holds:
  \begin{equation}\label{Cor+2}
  \left( \frac{\widehat{M}_{1, W, 2}}{\color{black}\sqrt{\widehat{M}_{2, W, 2} \color{black}} }  \sqrt{\frac{n \widehat{F}_{x, 2}(h_2)}{\omega_{n,1}(x)\pi_{n}(x) U^{2}_{n,1}(x)  }}\right) \left(U_{n,1}(x)-U(x)\right) \overset{\mathcal{D}}
{\longrightarrow} \mathcal{N}\left(0, 1\right),\,\ \hbox{as}\,\ n\rightarrow\infty.
\end{equation}
\end{cor}

An asymptotic $100(1-\nu)\%$ confidence interval for $U(x)$ is given by:
 \begin{equation}\label{CINI}
 \mathrm{CI}^{\text{NPI}}_{\nu} = U_{n,1}(x) \left( 1 \pm q_{\nu / 2} \frac{\sqrt{\widehat{M}_{2, W, 2}}}{\widehat{M}_{1, W, 2}} \sqrt{\frac{\omega_{n,1}(x)\pi_{n}(x)}{n \widehat{F}_{x, 2}(h_2)}} \right),
\end{equation}
where $q_{\nu / 2}$ denotes the upper $\nu / 2$ quantile of the  standard normal distribution.

Comparing the simplified-based confidence interval \eqref{CIS} and the nonparametric imputation-based confidence interval \eqref{CINI}, it becomes evident that higher missing rates ($\pi(x) \rightarrow 0$) result in significantly wider $\mathrm{CI}_\nu^{\text{S}}$ compared to $\mathrm{CI}_\nu^{\text{NPI}}.$ This highlights that imputing  missing data provides a more precise interval estimation of $U(x)$.

\section{Finite sample performance}\label{sec4}

In this section, we conduct a simulation study to evaluate the performance and effectiveness of the proposed estimation methods. 

\subsection{Data generating process}

Consider the process $(X_t, Y_t, \delta_t)_{t=1,\cdots,n}$ as a strict stationary process such that, for any $t=1, \dots, n,$ the functional covariate $\{X_t(\lambda): \lambda \in [-1,1]\}$ is  sampled at $100$ equally spaced points within the interval $[-1,1]$. the covariate is generated according to the following model: 
\begin{eqnarray}\label{simX}
X_t(\lambda)=A(2-\cos(\pi\lambda \omega))+(1-A)\cos(\pi \lambda \omega),
\end{eqnarray}
where $\omega \sim \mathcal{N}(0,1)$, and $A \sim \text{Bernoulli}\bigg(\dfrac{1}{2}\bigg).$ A sample of 100 simulated curves is presented in Figure \ref{funcurv}.
\begin{figure}[h!]
\begin{center}
\includegraphics[width=10cm,height=5.5cm]{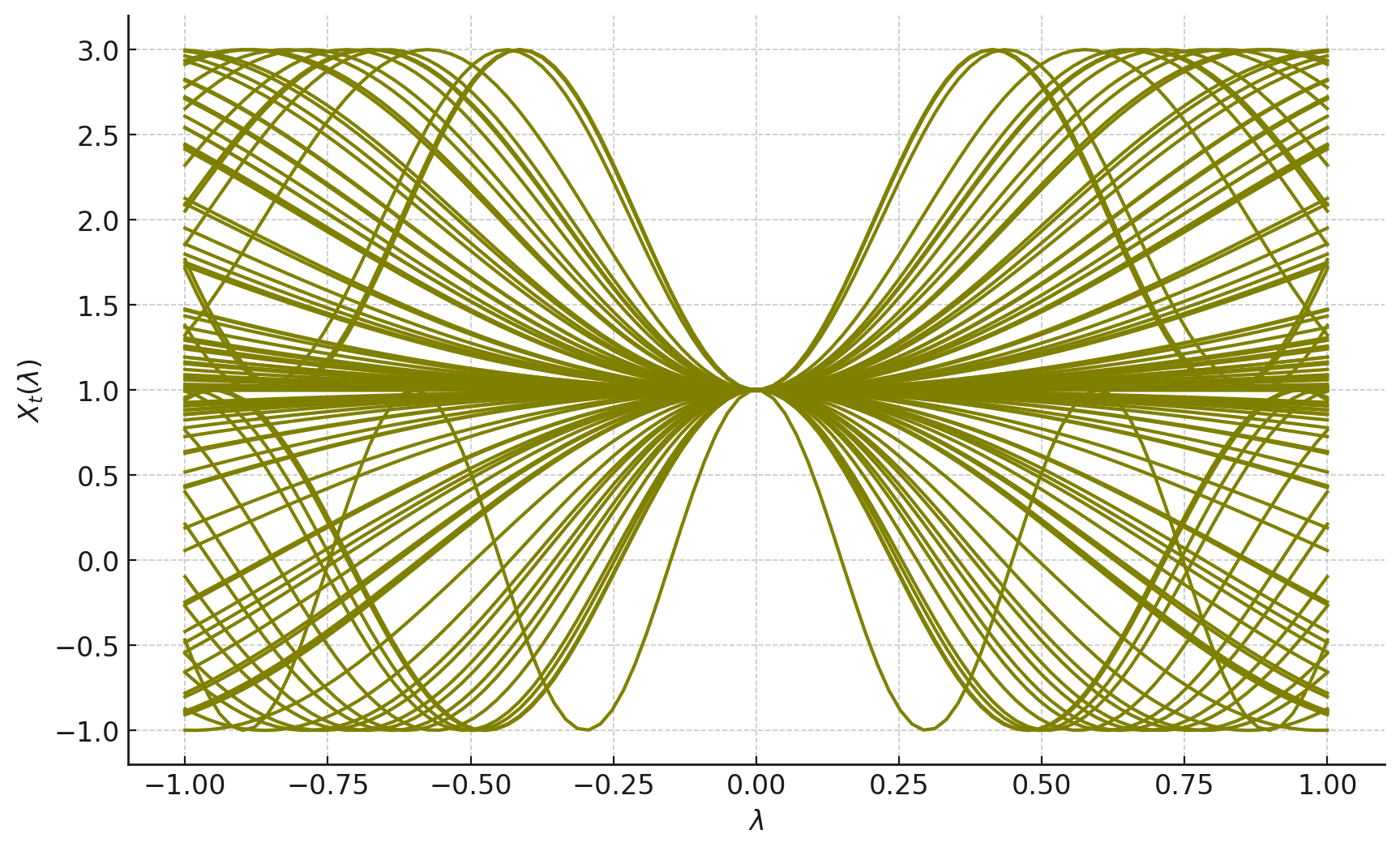}
\end{center}
\caption{A sample of simulated curves $X_t(\lambda)$.}
\label{funcurv}
\end{figure}

The response variable $Y$ is generated using the following heteroscedastic functional regression model:
$$ Y_t=m(X_t)+\sqrt{U(X_t)}\varepsilon_t,$$
\noindent where the regression and variance operators are defined at any fixed point $x$ as:
\begin{equation}\label{Reg&VarFun}
    m(x)=\int_{-1}^1 \lambda x(\lambda) d\lambda, \quad U(x)=\int_{-1}^1 |\lambda| x^2(\lambda) d\lambda.
\end{equation}
\noindent The errors $\varepsilon_t$ are generated according to one of the following models:\\
\noindent {\bf Model 1:} The $\varepsilon_t$'s are i.i.d, distributed according to $\mathcal{N}(0,1)$.\\
\noindent{\bf Model 2:} $\varepsilon_t=0.5\;\varepsilon_{t-1}+\xi_t$, where $\xi_t \sim \mathcal{N}(0,1)$.\\
\noindent {\bf Model 3:} $\varepsilon_t=-0.25\;\varepsilon_{t-1}+\xi_t$, where $\xi_t \sim \mathcal{N}(0,1)$.\\
\noindent {\bf Model 4:} $\varepsilon_t=\dfrac{1}{2}\varepsilon_{t-1}+\xi_t$, where $\xi_t \in\{ -1, 1\}$ and $\xi_t \sim \texttt{Bernoulli}\bigg(\dfrac{1}{2}\bigg)$.\\

The four models correspond to varying dependence structures. Model 1 assumes independent and identically distributed data. Models 2 and 3 consider $\alpha$-mixing processes. Finally, Model 4 provides an example of an ergodic process that is not mixing. 

    

We assume that missing at random observations in  the response variable $Y$ are generated by the following conditional probability distribution:

$$\pi(x)=\mathbb{P}(\delta=1|X=x)=\text{expit}\bigg(2\eta\int_{-1}^1 x^2(\lambda) d\lambda \bigg),$$
where $\text{expit}(u)=\dfrac{e^u}{1+e^u}$ and $\eta \in \{0.2,0.8\}$.

The value of $\eta$ determines the MAR rate. A higher value of $\eta$ leads to a higher probability $\pi(x),$ and consequently, a lower missing data rate. Specifically, for $\eta=0.8,$ the MAR rate is approximately $10\%,$ while for $\eta = 0.2,$ the MAR increases to $30\%.$ Figure \ref{Y_t MAR} illustrates the process $Y_t$ under the MAR mechanism for $\eta = 0.2$ and $\eta = 0.8$, respectively.

\begin{figure}[h!]%
    \centering
    \subfloat[\centering ]{{\includegraphics[width=7.5cm,height=5.5cm]{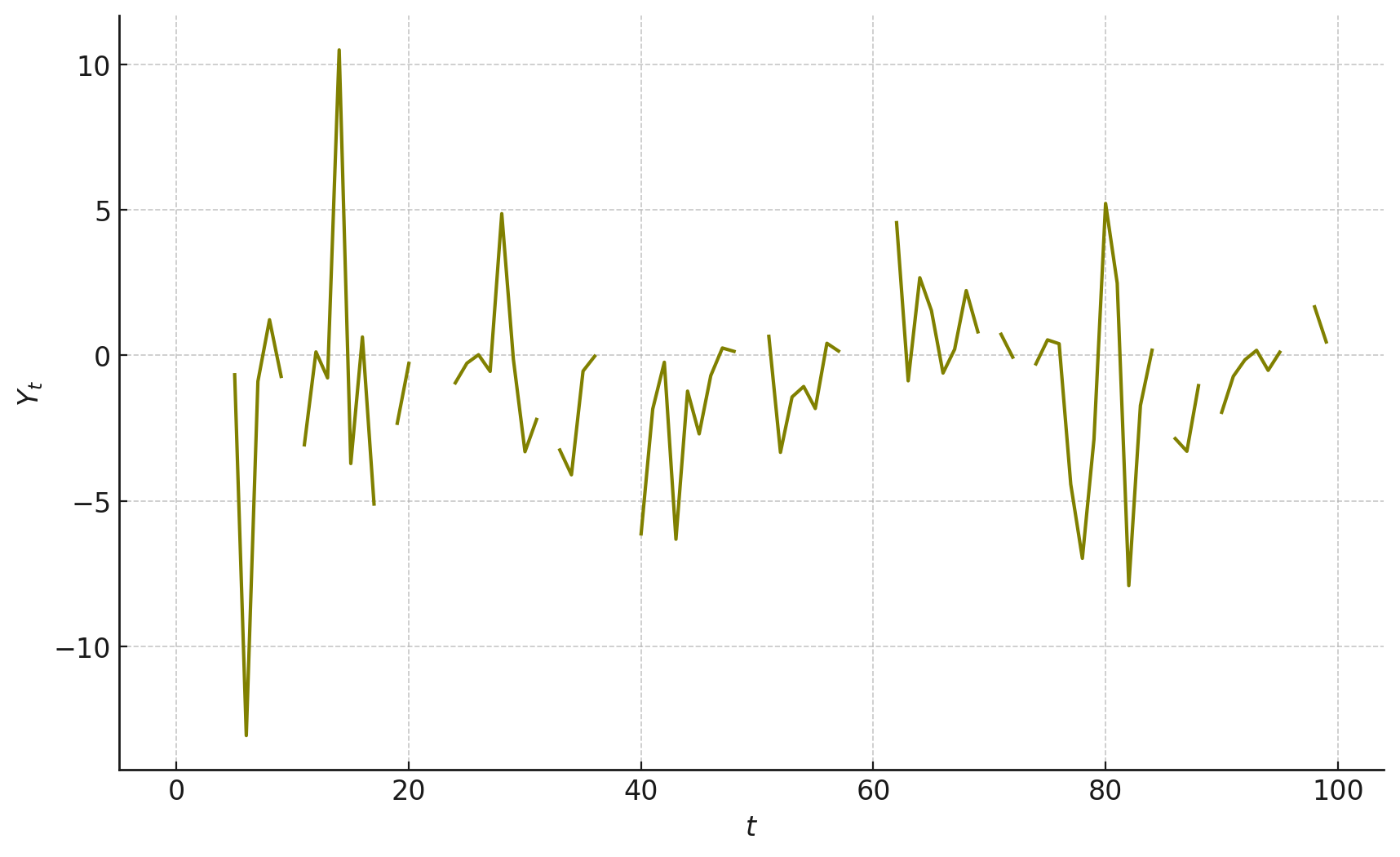} }}%
    \qquad
    \subfloat[\centering ]{{\includegraphics[width=7.5cm,height=5.5cm]{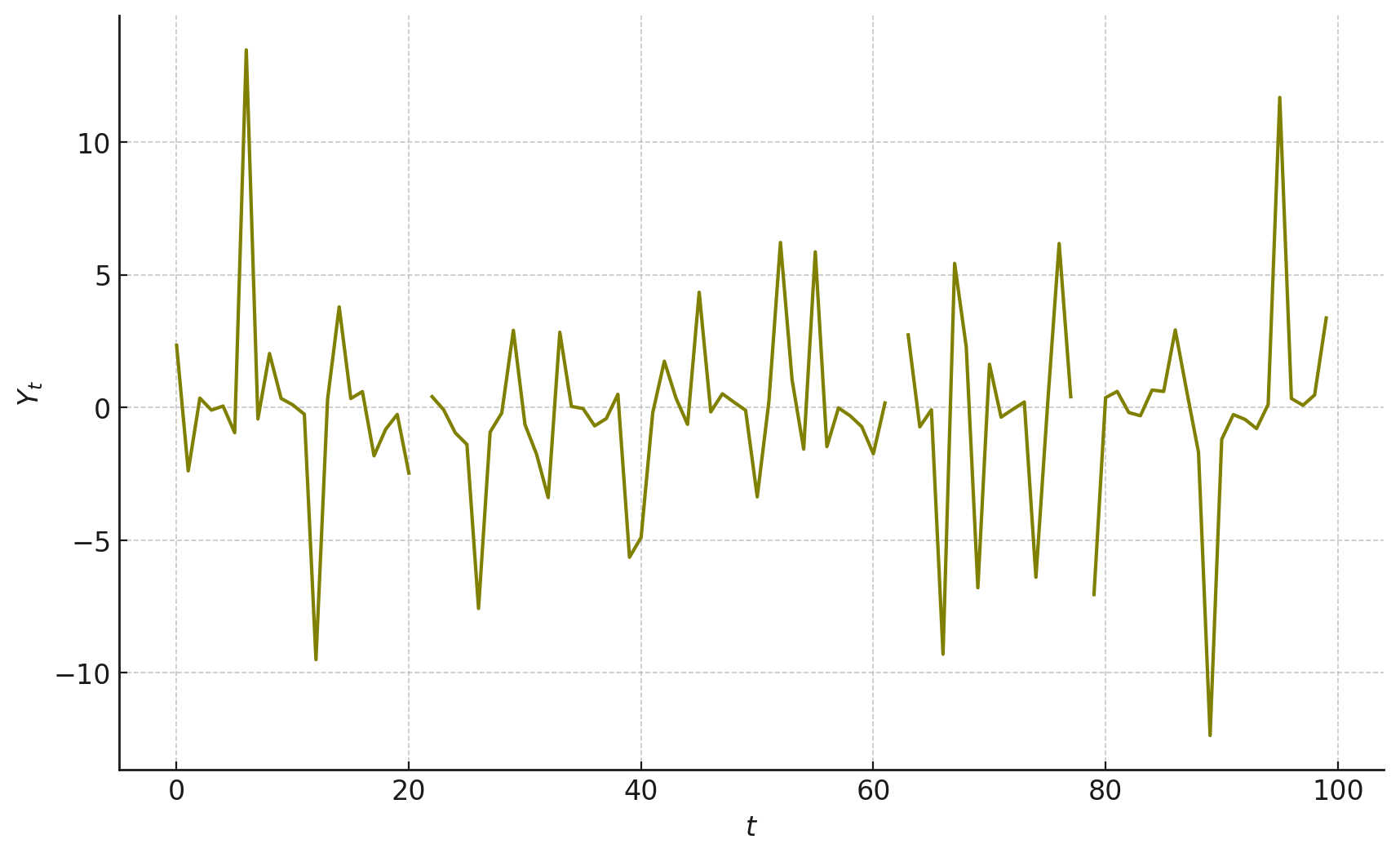} }}%
   
    \caption{Missing at random mechanism applied to the generated process $(Y_t)_{t=1, \dots, 100}$ with $\eta=0.2$ ($30\%$ MAR) in (a) and $\eta =0.8$ ($10\%$ MAR) in (b).}%
    
    \label{Y_t MAR}%
\end{figure}

\subsection{Tuning parameters selection}\label{tunpar}
The estimators discussed in Section \ref{main} depend on three key tuning parameters: the kernel, the semi-metric and the smoothing parameter $h.$
 \subsubsection{Choice of the kernel:} we use the quadratic kernel $K(u) = \frac{3}{2}(1-u^2)\1_{[0,1]}(u)$ to calculate both the simplified and the nonparametric imputed estimators. 
 
 \subsubsection{Choice of the semi-metric:} The choice of the semi-metric has been extensively discussed in \citet{Ferraty2006} (chapter 3), where several options are proposed depending on the structure of the functional data. In this study, due to the smoothness of the curves $X_t(\lambda)$, we employ a semi-metric based on the $L_2$-norm of the first derivatives of the curves, defined as:
$$d(X_t,X_s)=\bigg[\int_{-1}^1\bigg\{X^{(1)}_t(\lambda)-X^{(1)}_s(\lambda)\bigg\}^2 d\lambda\bigg]^{1/2}, \quad \forall t\neq s.$$

\subsubsection{Cross-validation based bandwidth selection:}
For the smoothing parameters, we adapt the well-known cross-validation technique to accommodate the missing data framework. 

The optimal bandwidth for the simplified regression operator estimator is determined by minimizing the following criterion:
\begin{eqnarray}
    \label{h_m}
    h^{\text{opt}}_{1} = \argmin_{h_1} \sum_{t=1}^n \delta_t \left(Y_t - m_{n,0}^{(-t)}(X_t; h_1)\right)^2,
\end{eqnarray}
where $ m_{n,0}^{(-t)}(X_t; h_1)$ represents the simplified estimator of the regression operator evaluated at $X_t$ with a given bandwidth $h_1$ excluding the $t^{th}$ observation.

Similarly, the optimal bandwidth for the simplified conditional variance component is obtained by:

\begin{eqnarray}
    \label{h_u}
    h^{\text{opt}}_{2} = \argmin_{h_2} \sum_{t=1}^n \delta_t \left(r_t - U_{n,0}^{(-t)}(X_t; h_2)\right)^2,
\end{eqnarray}
where $r_t = (Y_t - m_{n,0}(X_t; h_1^{\text{opt}}))^2$ and $U_{n,0}^{(-t)}(X_t; h_2)$ is the simplified conditional variance estimator obtained after excluding the $t^{th}$ observation. 

To construct the confidence intervals, we estimate $\pi(x)$ and $\omega(x)$ using formulas \eqref{w_n} and \eqref{pi_n}, respectively. The optimal bandwidth for these estimators is chosen as follows:

\begin{equation}\label{h_3}
h_3^{\mathrm{opt}}=\underset{h_3}{\arg \min } \sum_{t=1}^n \delta_t\left(\left(\hat{\varepsilon}_{t, 0}^2-1\right)^2-\omega_{n, 0}^{(-t)}\left(X_t ; h_3\right)\right)^2,
\end{equation}
where $\omega_{n, 0}^{(-t)}\left(X_t ; h_3\right)$ is the value of the simplified estimator of $\omega\left(X_t\right)$ obtained after removing the $t^{\text {th }}$ residual $\hat{\varepsilon}_{t, 0}$ from the sample.

The optimal bandwidth for the estimator of $\pi(x)$ is determined by:
$$
h_4^{\mathrm{opt}}=\underset{h_4}{\arg \min } \sum_{t=1}^n\left(\delta_t-\pi_n^{(-t)}\left(X_t ; h_4\right)\right)^2.
$$

 Finally, the optimal bandwidths for the nonparametric imputed estimators of $m(\cdot), U(\cdot)$ and $\omega(\cdot)$ are derived similarly to \eqref{h_m}, \eqref{h_u} and \eqref{h_3}, with the adjustment of removing $\delta_t$ from the formulas and replacing $m_{n, 0}^{(-t)}\left(X_t ; h_1\right), U_{n, 0}^{(-t)}\left(X_t ; h_2\right), \omega_{n, 0}^{(-t)}\left(X_t ; h_3\right)$ and $\hat{\varepsilon}_{t, 0}$ by $m_{n, 1}^{(-t)}\left(X_t ; h_1\right), U_{n, 1}^{(-t)}\left(X_t ; h_2\right)$,  $\omega_{n, 1}^{(-t)}\left(X_t ; h_3\right)$ and $\hat{\varepsilon}_{t, 1}$, respectively.

\subsection{Performance assessment}
\subsubsection{Study of the consistency of the estimators}

The objective is to evaluate the performance of the conditional variance estimators over a grid of $J=100$ randomly selected curves, denoted as $\mathfrak{J}:= \{x^\star_1, \dots, x^\star_J\}$, simulated using equation \eqref{simX}. To assess the consistency of the different estimators, we conduct $B=500$ replications. For each replication, we estimate the conditional variance at each curve in the grid $\mathfrak{J}$ and compute, at each iteration $b\in \{ 1, \dots, B\},$ the Mean Square Error ($\text{MSE}_b$) defined as:
\begin{equation}\label{MSE}
\text{MSE}_b = \dfrac{1}{J}\sum_{j=1}^J\left(\mathcal{U}_{n,b}(x^\star_j) - U(x^\star_j) \right)^2,
\end{equation}
where $\mathcal{U}_{n,b}(x^\star_j)$ represents one of the three estimators of the conditional variance- namely, the complete estimator $U_{n,c}(x^\star_j)$, the simplified estimator $U_{n,0}(x^\star_j)$, or the nonparametric imputed estimator $U_{n,1}(x^\star_j)$- evaluated at the $j^{\text{th}}$ point in the grid $\mathfrak{J}$, during the $b^{\text{th}}$ replication.

The Mean Integrated Square Error (MISE) is then computed as: $\text{MISE} := B^{-1}\sum_{b=1}^B \text{MSE}_b.$ Figures \ref{MISE_300} and \ref{MISE_50} present a heatmaps of the MISE values across models and MAR rates, with sample sizes of $n=300$ and $n=50$, respectively. To provide deeper insights into the variability of the MSE, these figures also include the first quartile $Q1$, the median, and the third quartile $Q3$ of the distribution of $(\text{MSE}_b)_{b=1, \dots, B}.$ This additional information allows us to assess not only the average performance of the estimators (through the MISE) but also the dispersion and variability of the MSE as the sample changes across replications.
\begin{figure}[h!]%
    \includegraphics[width=15cm,height=9cm]{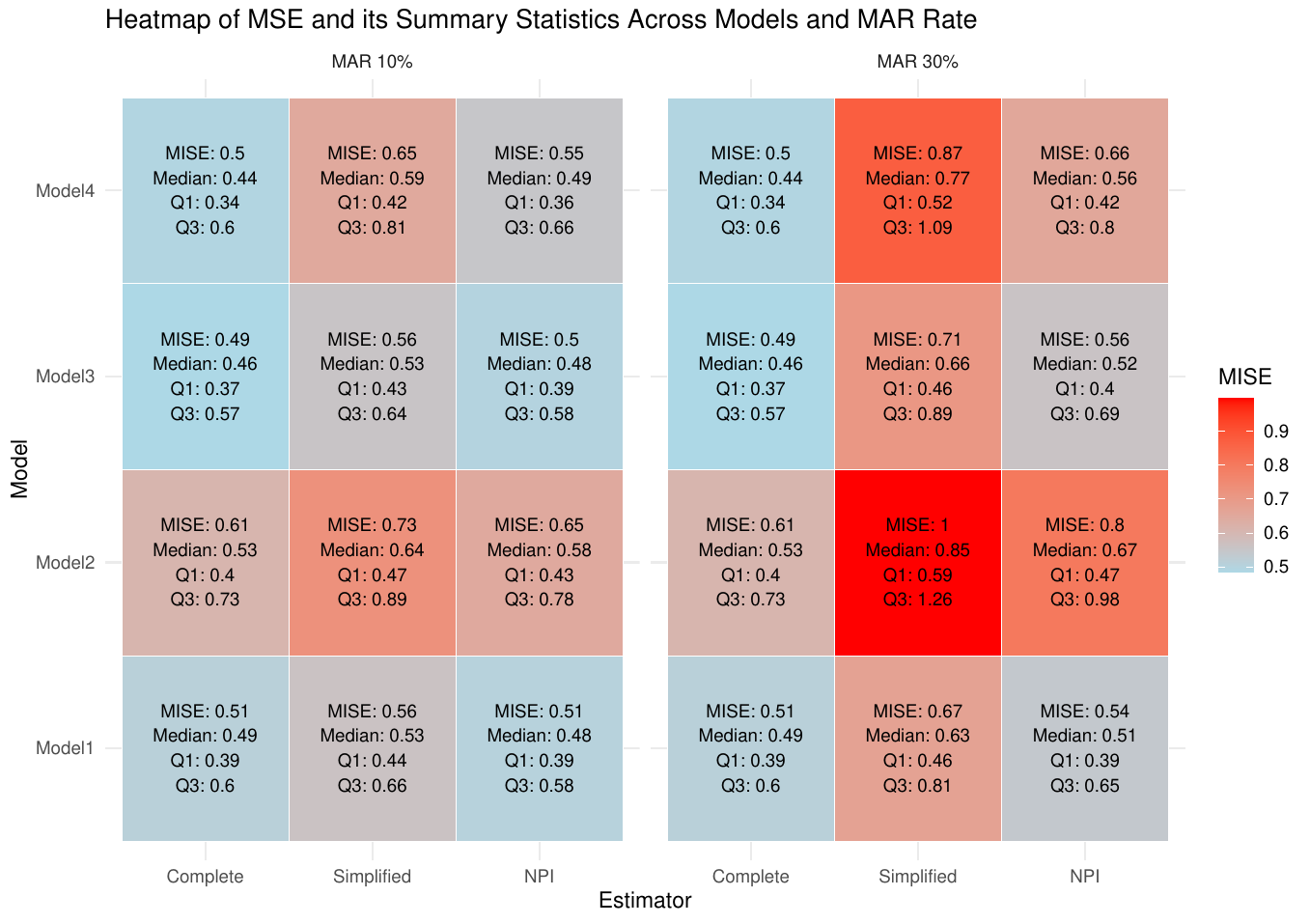} %
    \caption{Heatmaps of Mean Squared Error (MSE) along with summary statistics (median, first quartile $Q1$, and third quartile $Q3$) across models and missing at random rates for $n = 300.$}%
    
    \label{MISE_300}%
\end{figure}

\begin{figure}[h!]%
    \includegraphics[width=15cm,height=9cm]{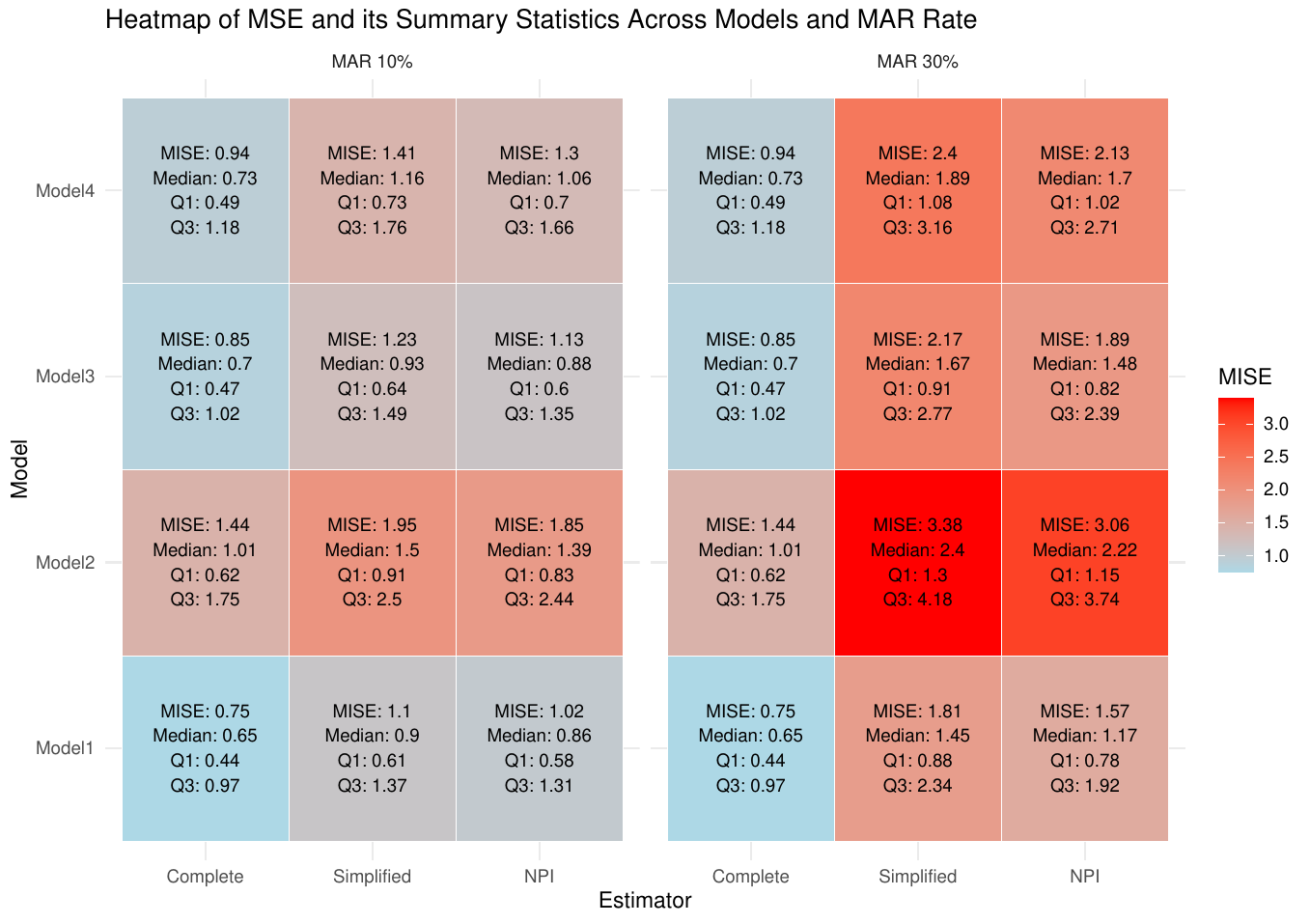} %
    \caption{Heatmaps of Mean Squared Error (MSE) along with summary statistics (median, first quartile $Q1$, and third quartile $Q3$) across models and missing at random rates for $n = 50.$}%
    
    \label{MISE_50}%
\end{figure}

The colors in Figures \ref{MISE_300} and \ref{MISE_50} represent the MISE values for each model, missing at random rate, and estimator type (complete, simplified, or nonparametric imputed). Across all four models, the complete estimator consistently outperforms the simplified and nonparametric imputed estimators. For a low MAR rate ($\text{MAR}=10\%$) and a larger sample size $n=300$, the nonparametric imputed estimator emerges as a strong competitor to the complete estimator. On the other hand, the nonparametric imputed estimator exhibits better performance than the simplified estimator across all models and sample sizes. The performance gap between these two estimators becomes more pronounced as the MAR rate increases.

To quantify the relative improvement of the nonparametric imputed estimator over the simplified estimator, we calculate its efficiency using the formula:
$$
\text{Eff.}(\%) = \dfrac{\text{MISE}_{\text{Simp}} - \text{MISE}_{\text{NPI}}}{\text{MISE}_{\text{Simp}}}\times 100.
$$

Positive values of $\text{Eff.(\%)}$ indicate that the MISE for the simplified estimator is higher than that for the nonparametric imputed estimator. A higher efficiency percentage reflects a greater advantage of the imputed estimator over the simplified one.

\begin{table}[h!]
\caption{MISE and efficiency for each model under MAR = 30\% and MAR = 10\% (values in parentheses).}\label{summary20}
\centering
\begin{tabular}{ccccc||ccccc}
    \hline
    \texttt{Model} & \multicolumn{4}{c||}{$n=300$} & \multicolumn{4}{c}{$n=50$}\\
    \hline
  & Comp & Simp. & NPI. & Eff.($\%$) & Comp & Simp. & NPI. & Eff.($\%$)\\
    \hline
   \texttt{Model1}    & 0.51 & 0.67 & 0.54 & 19.40 & 0.75 & 1.81 & 1.57 & 13.25 \\
    & & (0.56) & (0.51) & (8.90) & & (1.10) & (1.02) & (7.27)\\
    \hline
    
     \texttt{Model2}    & 0.61 & 1.00  & 0.80 & 20.00 & 1.44 & 3.38 & 3.06 & 9.46\\
     & & (0.73) & (0.65) & (10.9) & & (1.95) & (1.85) & (5.12)\\ 
     
    \hline
    
     \texttt{Model3}     & 0.49 & 0.71 & 0.56 & 21.12 & 0.85 & 2.17 & 1.89 & 12.90 \\
      & & (0.56) & (0.50) & (10.71) & & (1.23) & (1.13) & (8.13)\\ 
    \hline
     \texttt{Model4}    &  0.5 & 0.87 & 0.66 & 24.13 & 0.94 & 2.40 & 2.13 & 11.25  \\
     & & (0.65) & (0.55) & (15.38) & & (1.41) & (1.30) & (7.80)\\ 
    \hline
    \label{eff}
\end{tabular}
\end{table}

Table \ref{eff} summarizes the MISE values and efficiency scores for each data-generating model, considering different MAR rates and sample sizes ($n=300$ and $n=50$). The results consistently show that the nonparametric imputed estimator is more efficient than the simplified estimator across all models and settings. The efficiency of the imputed estimator increases substantially as the MAR rate and sample size grow, highlighting its robustness and improved accuracy under challenging scenarios.

\subsubsection{Assessment of confidence intervals estimation}

Given the formulas for the asymptotic confidence intervals of the simplified and imputed estimators, provided in \eqref{CIS} and \eqref{CINI}, respectively, one can observe that, asymptotically, as the MAR rate increases (i.e. $\pi_n(x)$ decreases), the length of $\text{CI}_\nu^{\text{S}}$ increases, whereas the length of $\text{CI}_\nu^{\text{NPI}}$ decreases. This observation is numerically validated in Figure \ref{CIres} (top-right). Consequently, Figure \ref{CIres} (top-left) illustrates that, unlike $\text{CI}_\nu^{\text{S}}$, the coverage rate of $\text{CI}_\nu^{\text{NPI}}$ decreases as the MAR rate increases. For comparison, we employ the asymptotic confidence interval for complete data introduced in Section 3.3.1 of \citet{Chaouch2019}.

To enable a fair comparison between the estimation methods, we adopt a criterion that balances the trade-off between coverage rate and confidence interval length, ensuring a comprehensive evaluation of the interval estimation performance. The criterion is defined as follows: 
$$
\text{Coverage Efficiency} = \dfrac{\text{Coverage rate}}{\text{Average CI length}}\times 100.
$$

\begin{figure}[h!]
    \centering
    \begin{minipage}{0.5\textwidth}
        \centering
        \includegraphics[width=\linewidth, height=5cm]{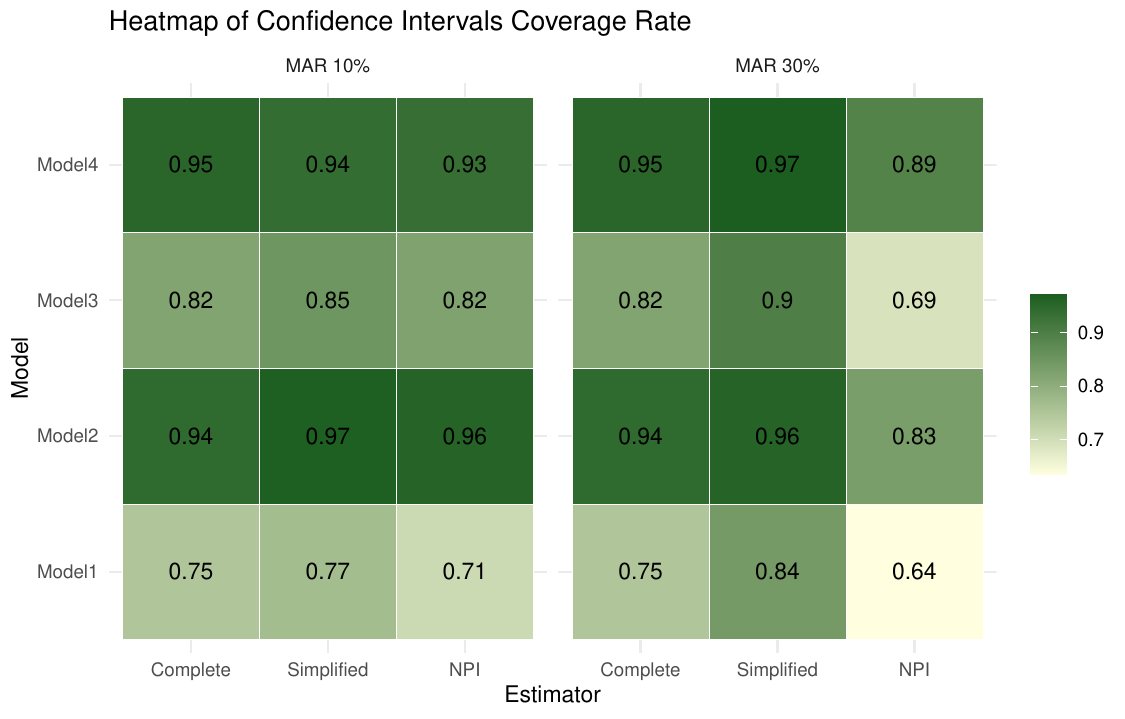}
    \end{minipage}\hfill
    \begin{minipage}{0.5\textwidth}
        \centering
        \includegraphics[width=\linewidth, height=5cm]{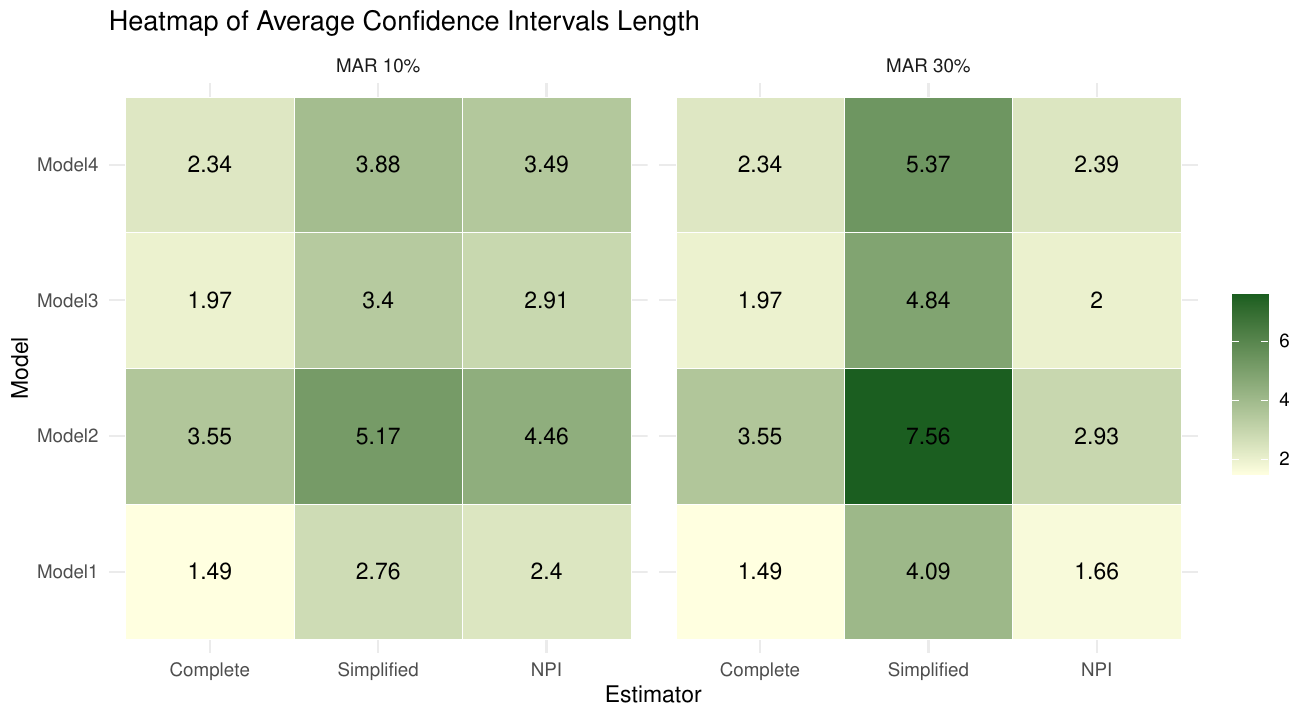}
    \end{minipage}
    
    \vspace{1em}  
    \begin{minipage}{0.6\textwidth}
        \centering
        \includegraphics[width=\linewidth, height=5cm]{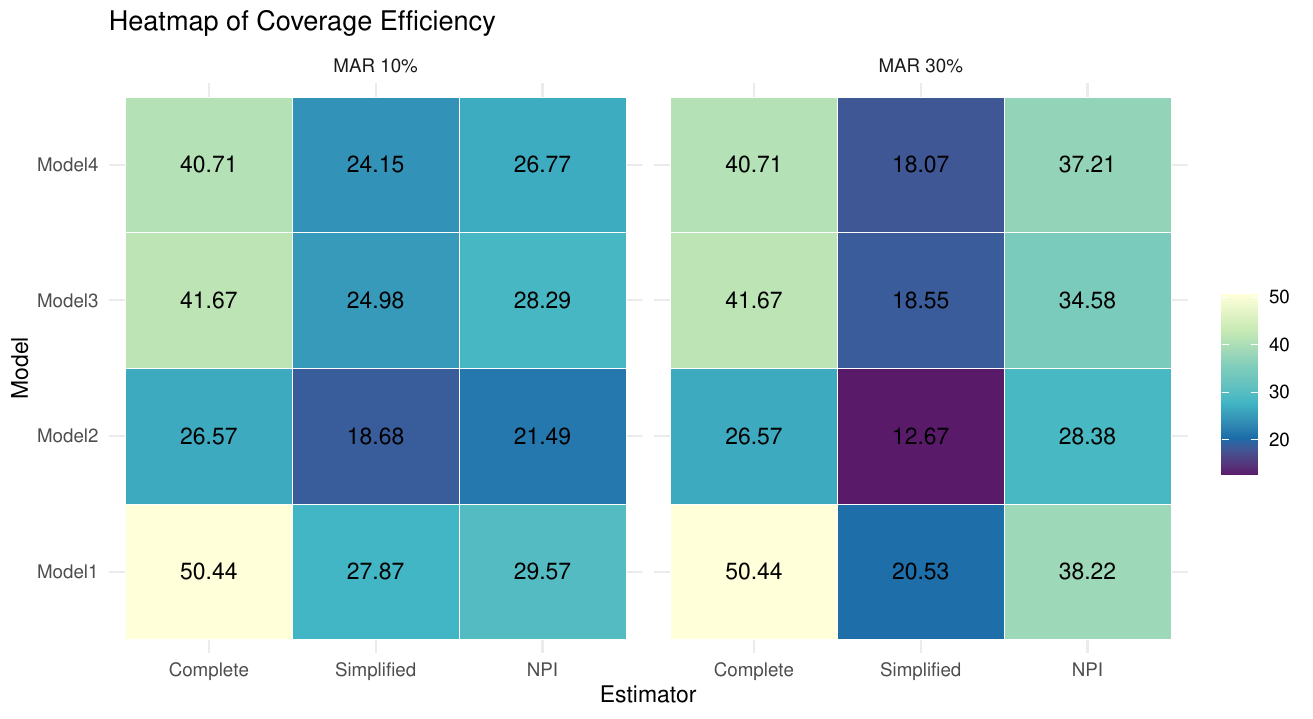}
    \end{minipage}
    \caption{Coverage rate (top-left), average confidence intervals length (top-right), and coverage efficiency (bottom) across models and MAR rates, based on 500 simulations with $n=300.$}
    \label{CIres}
\end{figure}

Figure \ref{CIres} (bottom) illustrates the coverage efficiency across various models and MAR rates. The results indicate that, for all models, the highest coverage efficiency is achieved when the data are completely observed. Additionally, the figure demonstrates that $\text{CI}_\nu^{\text{NPI}}$ outperforms $\text{CI}_\nu^{\text{S}}$ in terms of coverage efficiency, particularly as the missing rate increases. Similar findings are observed for $n = 50$, with the corresponding results displayed in Figure \ref{CIres_50}. Finally, as expected, a comparison between Figures \ref{CIres} and \ref{CIres_50} reveals that the performance of the estimated confidence intervals declines as the sample size decreases.

\begin{figure}[h!]
    \centering
    \begin{minipage}{0.5\textwidth}
        \centering
        \includegraphics[width=\linewidth, height=5cm]{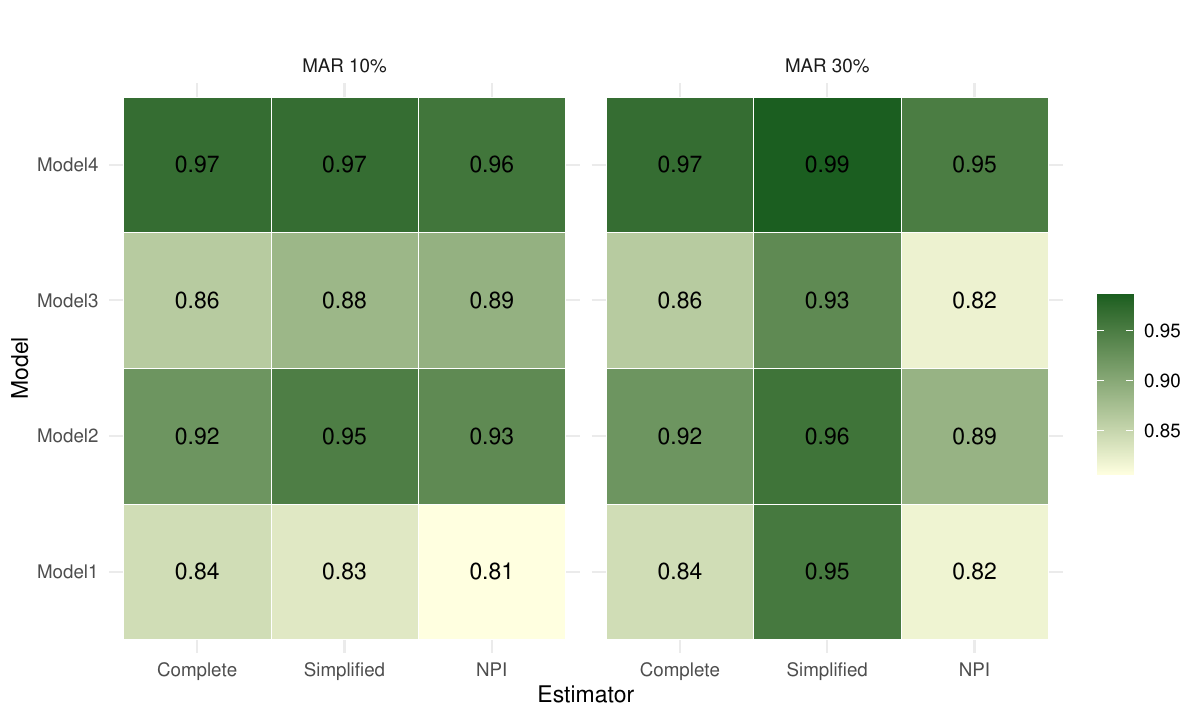}
    \end{minipage}\hfill
    \begin{minipage}{0.5\textwidth}
        \centering
        \includegraphics[width=\linewidth, height=5cm]{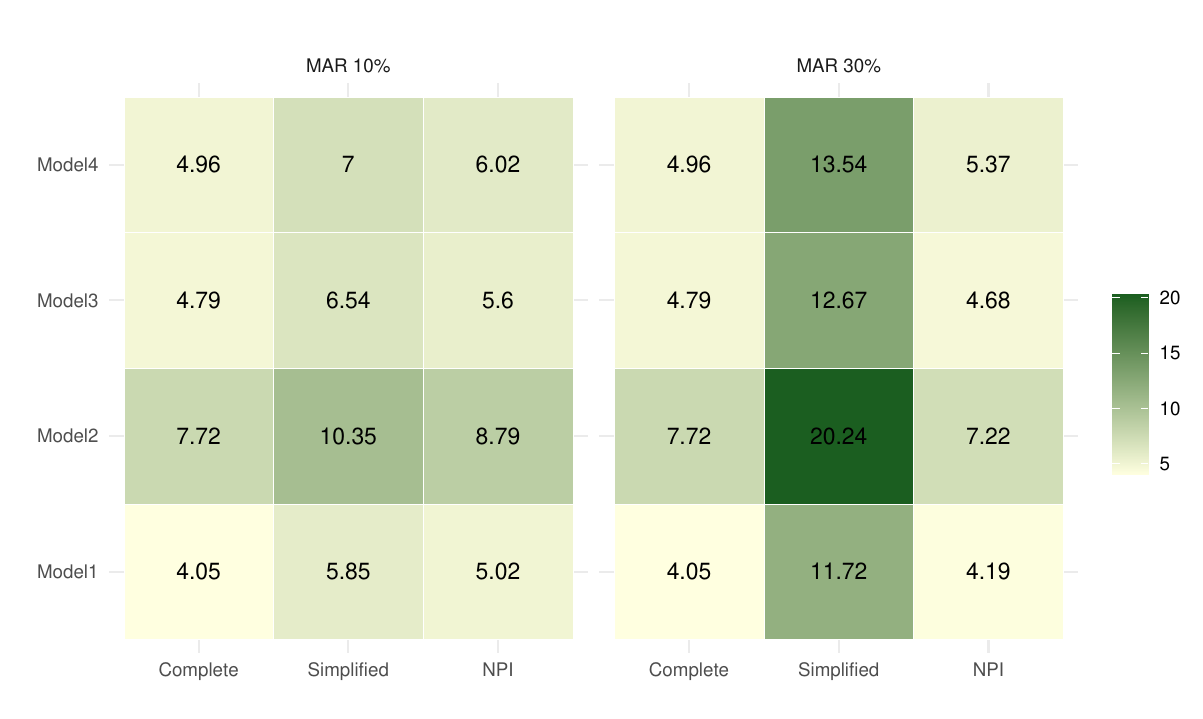}
    \end{minipage}
    
    \vspace{1em}  
    \begin{minipage}{0.6\textwidth}
        \centering
        \includegraphics[width=\linewidth, height=5cm]{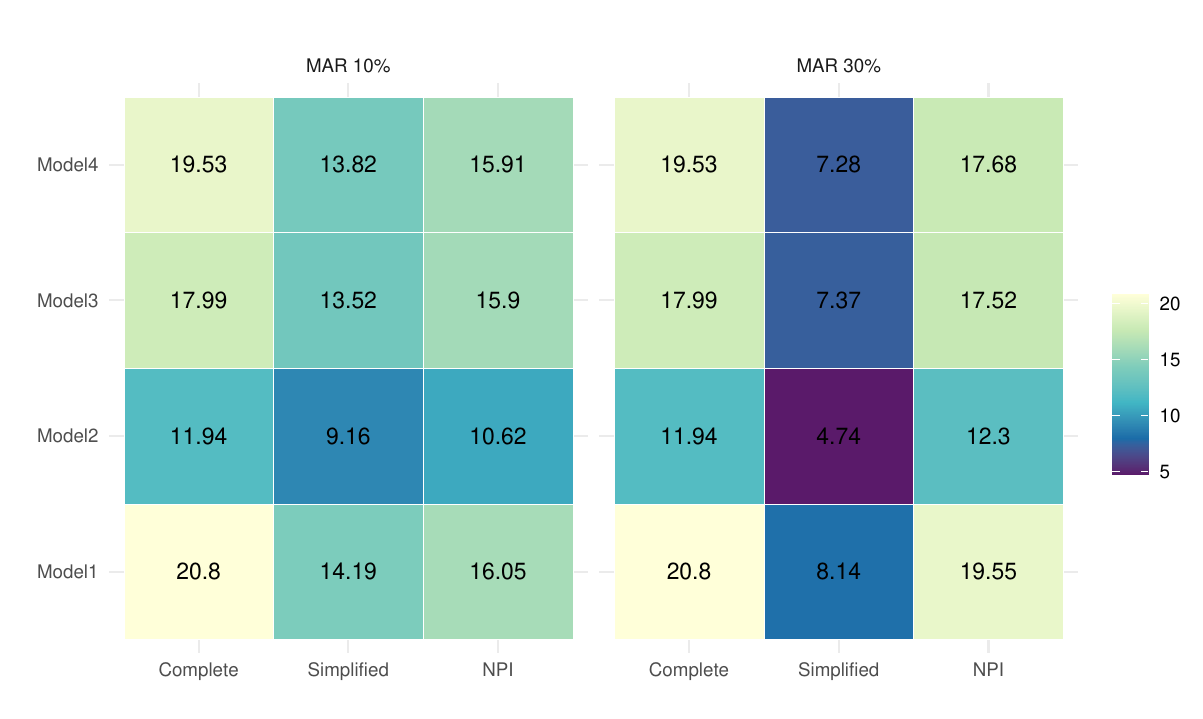}
    \end{minipage}
    \caption{Coverage rate (top-left), average confidence intervals length (top-right), and coverage efficiency (bottom) across models and MAR rates, based on 500 simulations with $n=50.$}
    \label{CIres_50}
\end{figure}

\section{Daily Gas price volatility modeling using intraday EU/USD exchange rate curves}\label{sec5}
In recent decades, natural energy sources have become a prominent area of research due to their critical role in industrial and economic systems. Energy sources are essential to numerous industrial sectors, significantly affecting export revenues, exchange rates, and stock markets. Among these, natural gas is particularly noteworthy as a vital resource, with its price having a substantial impact on the economies of many countries. Beyond its primary use in generating heat and electricity, natural gas plays a pivotal role in the chemical, heavy, and food industries, where it is a key component in the production of plastics, detergents, and paints.

The central role of natural gas in the global economy necessitates an investigation into its relationship with other macroeconomic and financial factors. For instance, using the Granger causality test, \citet{Sul} identified a strong connection between natural gas prices and exchange rates. This relationship is attributed to the fact that natural gas, like other commodities, is quoted and settled in U.S. dollars. When the dollar depreciates relative to other currencies, the price of natural gas tends to rise, whereas when the dollar appreciates, natural gas prices typically decline (see Figure \ref{Gas_EUUSD}).

\begin{figure}[h!]%
    \centering
    \includegraphics[width=15cm,height=7.5cm]{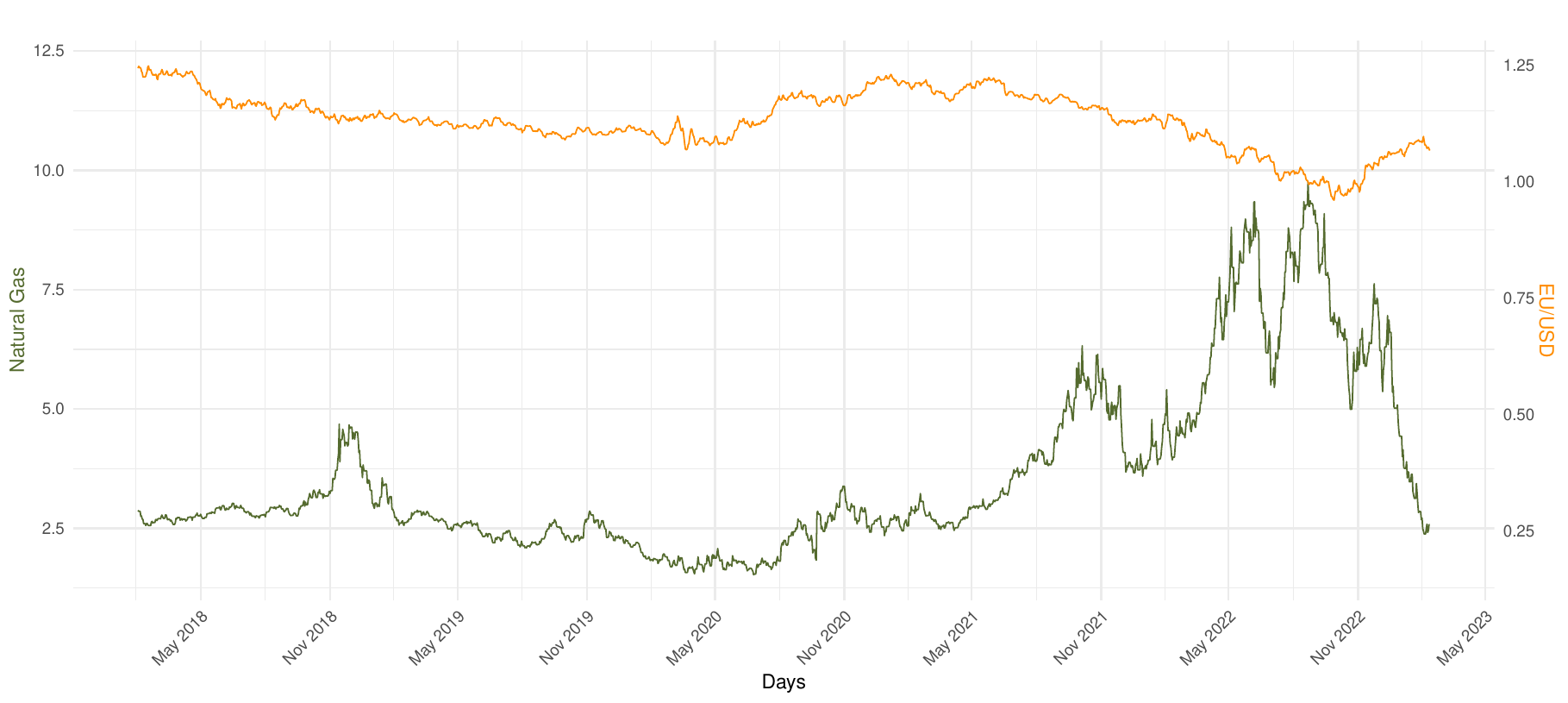} %
    \caption{Closing price of natural gas and EU/USD exchange rate.}%
    
    \label{Gas_EUUSD}%
\end{figure}

In this section, we aim to estimate the conditional volatility of daily natural gas log returns, utilizing the intraday curve of EU/USD exchange rate log returns as the predictor.

\subsection{Data description and random sample construction}

The dataset encompasses trading days from February 1, 2018, to February 1, 2023, totaling 1837 trading days. Figure \ref{brentGasClosingPrice} illustrates a 1-day frequency time series of natural gas closing prices. The log returns of natural gas prices are computed as $r_{t}^g:=\log(P_{t}^g/P_{t-1}^g) \times 100$, where $P_t^g$ represents the natural gas price observed on day $t.$ The results of the KPSS test indicate that the process  $r_{t}^g$ is stationary, with a $\text{p-value}$ of 0.1. Figure \ref{brentGasClosingPrice} displays both the natural gas prices $(P^g_t)_t$ and their corresponding log returns $(r_{t}^g)_t.$ 

\begin{rem}
    The theoretical framework presented in this paper assumes that the process is stationary and ergodic. When dealing with real-world data, the stationarity assumption can be evaluated using statistical tools such as the KPSS or Augmented Dickey-Fuller tests. However, to the best of our knowledge, there are no established methods for directly testing the ergodicity assumption. A similar challenge arises when developing theory for mixing processes, where this assumption cannot be verified unless a well-known mixing model is successfully fitted to the data. In practice, researchers typically proceed by implicitly assuming that the data meet the mixing or ergodicity assumptions and instead focus on assessing the performance of the proposed methodology. 
\end{rem}

Figure \ref{brentGasClosingPrice} also highlights periods of pronounced volatility in natural gas prices, as evidenced by sharp spikes in log returns. Notable periods include the end of 2018 to early 2019, driven by high winter demand and supply constraints; March 2020, marked by extreme fluctuations triggered by the onset of the COVID-19 pandemic and collapsing global demand; and late 2021 to 2022, characterized by sustained volatility resulting from the Russian-Ukrainian war and subsequent disruptions in energy supply. Even during mid-2022 to early 2023, volatility persisted due to geopolitical uncertainties and inflationary pressures, despite a gradual decline in prices. These fluctuations underscore the critical impact of economic, geopolitical, and seasonal factors on the dynamics of the natural gas market.

\begin{figure}[h!]%
    \centering
    \includegraphics[width=15cm,height=7.5cm]{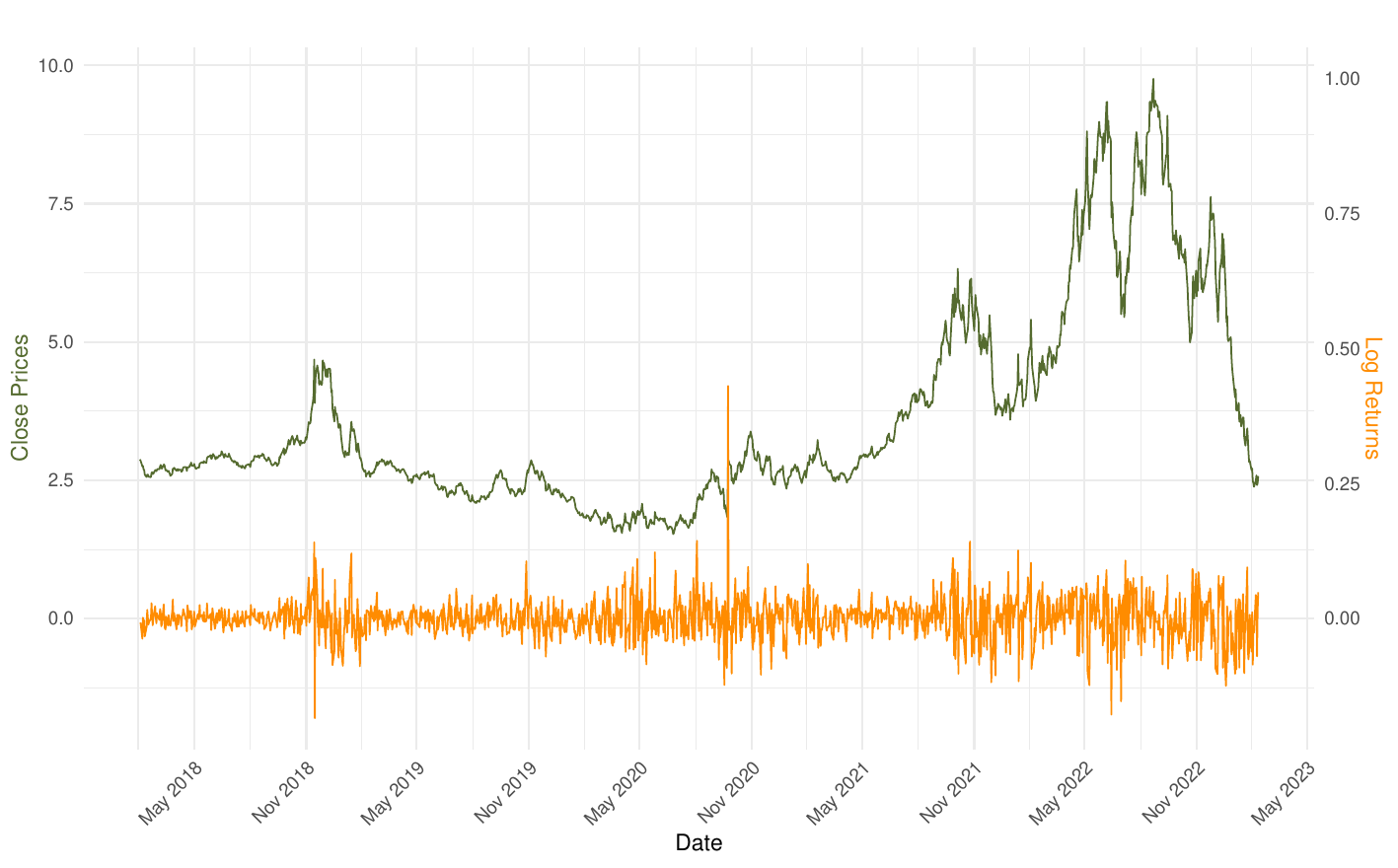} %
    \caption{Closing Price of Natural Gas as well as its log return.}%
    
    \label{brentGasClosingPrice}%
\end{figure}

Over the same period, we have access to 1-hour frequency data for the EU/USD exchange rate, denoted as $P^{\text{xr}}_t$. The corresponding log-returns are calculated as $r^{\text{xr}}_t := \log(P_{t}^{\text{xr}}/P_{t-1}^{\text{xr}}) \times 100$. Our random sample $(X_t, Y_t)_{t=1,\cdots, 1836}$ is defined such that $Y_t = r_t^g$ and $X_t(\lambda) = r_t^{\text{xr}}\left(\lambda + (t-1)\times 1836 \right),$ for $t= 1, \dots, n = 1836,$ and $\lambda\in [1, 23].$ Figure \ref{EUUSD_price_return}(a) illustrates the intraday (1-hour frequency) EU/USD log-return curves, while Figure \ref{EUUSD_price_return}(b) presents the corresponding closing price curves.


    


\begin{figure}[h!]%
    \centering
    \subfloat[\centering ]{{\includegraphics[width=7cm,height=5cm]{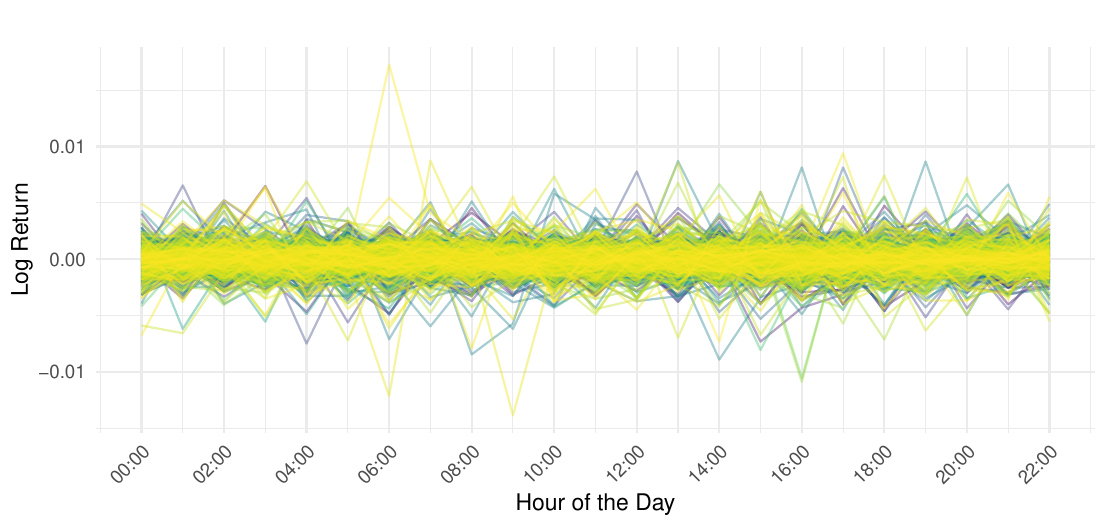} }}%
     \qquad
    \subfloat[\centering ]{{\includegraphics[width=7cm,height=5cm]{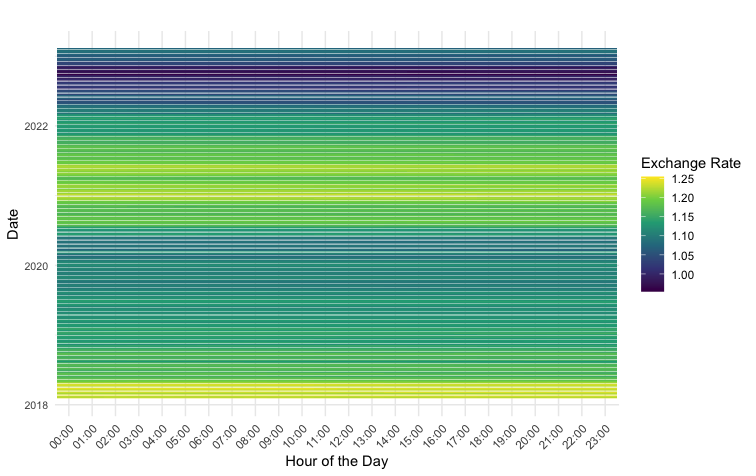} }}%
 
    \caption{(a) Intraday (1-hour frequency) EU/USD log-return curves. (b) HeatMaps of intraday EUR/USD exchange rates.}%
    \label{EUUSD_price_return}%
\end{figure}

Before proceeding with modeling the conditional volatility of gas returns using intraday exchange rate returns, it is essential to verify that the conditional volatility component is significant. To this end, we consider the model $Y_t = m(X_t) + \eta_t$ for $t\in \{1, \dots, n \}$. Figure \ref{etavsxbar} displays the residuals $\widehat{\eta}_t := Y_t - m_n(X_t)$ plotted against the average daily exchange rate log-returns $(\bar{X}_t)$. The plot clearly indicates a non-constant conditional variance in the residuals. Consequently, we adopt a model of the form $Y_t = m(X_t) + \sqrt{U(X_t)}\; \varepsilon_t,$ which incorporates a conditional volatility component.

 \begin{figure}[h!]%
    \centering
    \includegraphics[width=7cm,height=5cm]{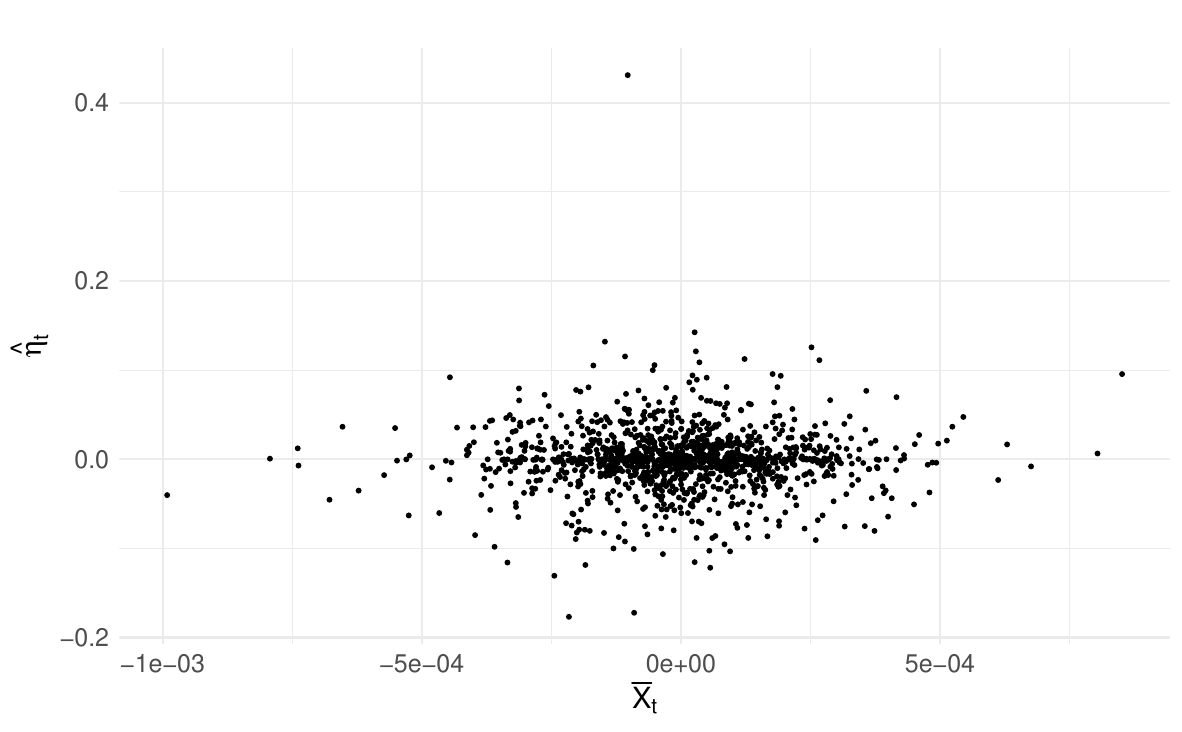} %
    \caption{Residuals $\widehat{\eta}_t$ plotted against the average daily log return of the EU/USD exchange rate.}%
    
    \label{etavsxbar}%
\end{figure}

 \begin{figure}[h!]
    \centering
    \begin{subfigure}[b]{\textwidth}
        \centering
        \includegraphics[width=\textwidth]{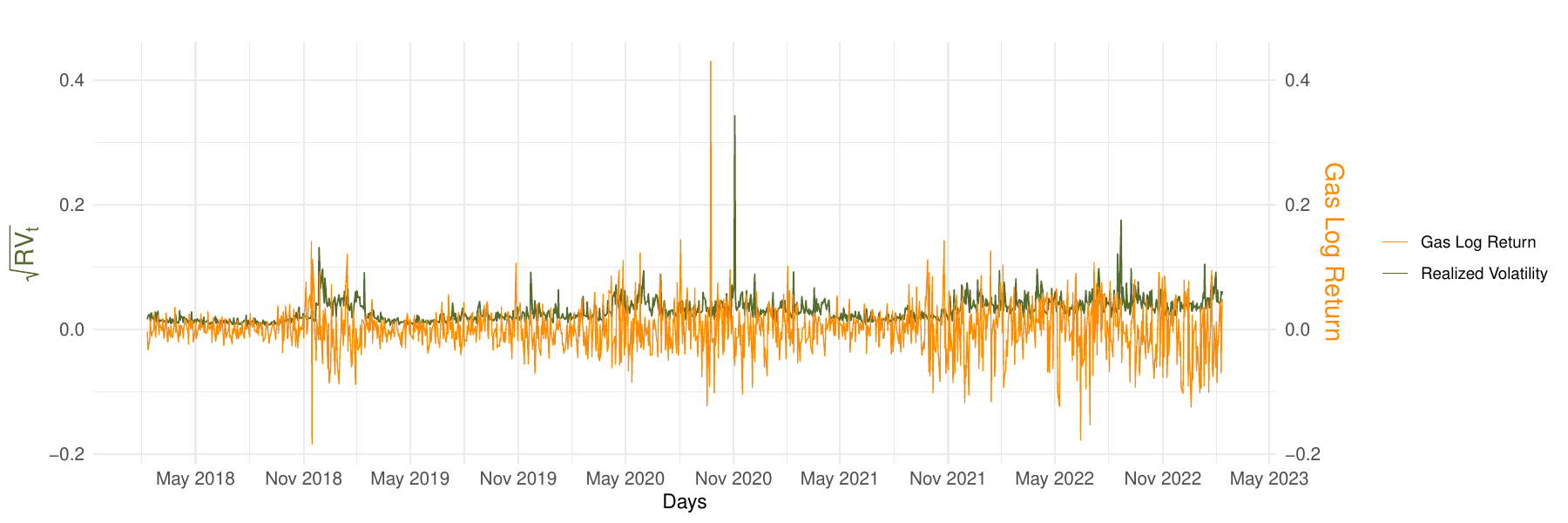}
        \caption{}
        \label{fig:subfig_a}
    \end{subfigure}

    \begin{subfigure}[b]{\textwidth}
        \centering
        \includegraphics[width=\textwidth]{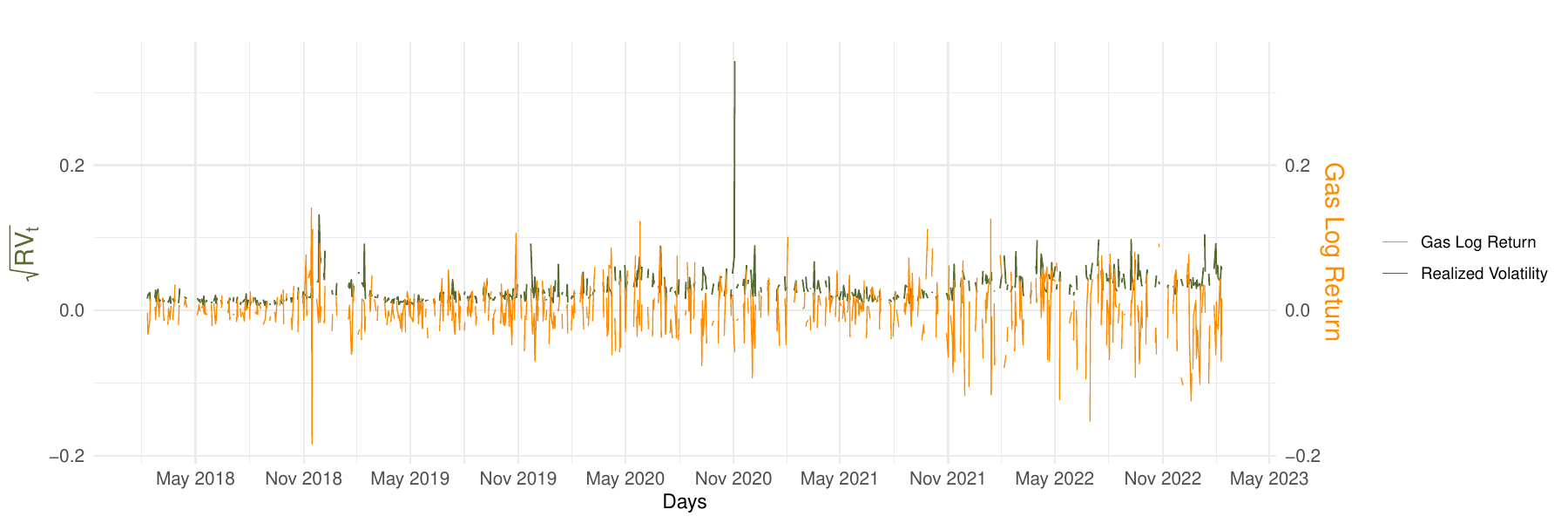}
        \caption{}
        \label{fig:subfig_b}
    \end{subfigure}

    \caption{(a) Log Return and Realized Volatility of natural gas with complete data. (b) Log Return and Realized Volatility of natural gas with $35\%$ missing data.}
    \label{RV_with_MAR}
\end{figure}

\begin{rem}
It is worth emphasizing that the dataset analyzed in this study is fully observed. This does not limit the applicability of our methodology to real-world scenarios, as the proposed approach remains valid even when the data are originally incomplete. The completeness of the dataset allows us to compute the true realized volatility of gas returns (as outlined in Section \ref{realizedvol}). In contrast, if the data were incomplete, it would not be possible to calculate the true conditional volatility. As a result, essential performance metrics such as the mean squared error and confidence interval coverage rate could not be evaluated, thereby constraining the ability to assess the accuracy and reliability of the point and interval estimations.
 \end{rem}

We consider that the missing at random mechanism is governed by the following conditional probability distribution:
 $$\pi(x)=\mathbb{P}(\delta=1|X=x)=\text{expit}\bigg(2\zeta\int_{1}^{23} x^2(\lambda) d\lambda \bigg),$$
where $\text{expit}(u)=\dfrac{e^u}{1+e^u}$ and $\zeta \in \{0.8, 2\}$. 

The parameter $\zeta$ controls the missing at random rate. Specifically, when  $\zeta = 0.8$ (resp. $\zeta=2$), approximately $35\%$ (resp. $15\%$) of the gas returns are missing. Figure \ref{RV_with_MAR}(a) shows the gas log-returns alongside the true realized volatility when data are fully observed. In contrast, Figure \ref{RV_with_MAR}(b) depicts the natural gas log-returns and the corresponding realized volatility with $35\%$ of the data missing. The objective of this study is to reconstruct the gaps in the realized volatility curve in Figure \ref{RV_with_MAR}(b) using the MAR gas log-returns as the response variable and the intraday exchange rate log-returns as the predictor.

\subsection{Realized volatility and performance measures}\label{realizedvol}
Volatility refers to a latent variable that cannot be directly observed, making the evaluation of its prediction a challenging task. However, the literature has demonstrated that true volatility can be approximated. For example, the daily realized volatility (see \cite{Merton1980}), computed from intraday return values, is widely regarded as a practical approximation of true volatility.

Using log-return data sampled at a 1-hour frequency, the realized volatility for a given day $t$ is calculated as:
$$ RV_t=\sum_{h=1}^{24} r_{t,h}^2,$$
where $r_{t,h}$ represents the log-return of natural gas observed at hour $h$ on day $t.$

\cite{Kara} established that the realized volatility converges to the true volatility as the sampling frequency of the log-returns from the original path increases. This implies that, for instance, using 1-minute frequency data for intraday gas log-returns provides a more accurate approximation of the true volatility compared to lower-frequency data.

To assess the accuracy of the estimators in estimating conditional volatility, we calculate the daily squared error, defined as
$$\text{SE}_t:=\left(\sqrt{\mathcal{U}_t(X_t)}-\sqrt{RV_t}\right)^2,$$
where $\mathcal{U}_t(X_t)$ represents the conditional variance estimator derived from either complete data ($U_{n,c}(\cdot)$), missing data ($U_{n,0}(\cdot)$), or imputed data ($U_{n,1}(\cdot)$). Additionally, rather than reporting the mean squared error (MSE) as a performance metric, we present the first, second, and third quartiles of the sequence $(\text{SE}_t)_t$. This approach provides not only an insight into the average squared error but also a clearer understanding of the variability in the errors.



\subsection{Point and interval estimation of daily conditional volatility }

The tuning parameters used to construct the conditional volatility estimators are defined as follows. We employed the quadratic kernel for $K, W$ and $H$, the bandwidths $(h_i)_{i=1, \dots, 3}$ are selected using the cross-validation technique described in Subsection \ref{tunpar}. For the semi-metric, given that the curves of exchange rate log-returns are not smooth, we utilized the PCA-based semi-metric, denoted as $d_4^{\mathrm{PCA}}(\cdot, \cdot)$. This semi-metric is based on the projection onto the four eigenfunctions, $v_1(\cdot), \dots, v_4(\cdot)$, associated with the four largest eigenvalues of the empirical covariance operator of the functional predictor $X.$ The PCA-based semi-metric is defined as:
 \begin{equation}
d_4^{\mathrm{PCA}}\left(X_{t}, X_{s}\right)=\sqrt{\sum_{k=1}^4\left(\int_1^{1439}\left(X_{t}(m)-X_{s}(m)\right) v_k(m) dm\right)^2} .
\end{equation}

Table \ref{Tab:AE} summarizes the squared errors $(\text{SE}_t)_t$, in percentage, for estimating the conditional volatility of gas log-returns using intraday EU/USD exchange rate returns under missing at random (MAR) scenarios of $15\%$ and $35\%.$ The estimator with complete data demonstrates stable performance across both MAR levels, with consistent error percentiles and mean squared error (MSE), highlighting its robustness when sufficient complete cases are available. The simplified estimator performs comparably to the complete estimator at MAR = 15\% but shows a significant increase in error metrics at MAR = 35\%, suggesting its reduced effectiveness under higher missingness rates. The nonparametric imputed estimator, while slightly less accurate at MAR = 15\% due to wider error dispersion, performs relatively better than the simplified estimator at MAR = 35\%, indicating its adaptability to higher missingness. Overall, the estimator with complete data is the most robust, while the nonparametric  imputed estimator offers a viable alternative for handling substantial missing data.

\begin{table}[h!]
\centering
\caption{Summary statistics of $(\text{SE}_t)$ (in $\%$) for each estimator under MAR rates of $15\%$ and $35\%.$}
    \begin{tabular}{l l l l l |l l l l}
    \hline
   Estimators & \multicolumn{4}{l}{$\text{MAR} = 15\%$} &
 \multicolumn{4}{|l}{$\text{MAR} = 35\%$}\\ \cline{2-9} & $\small Q_{25\%}$   & $\small Q_{50\%}$   & $\small Q_{75\%}$  & \small MSE & $\small Q_{25\%}$ & $\small Q_{50\%}$ & $\small Q_{75\%}$ & \small MSE\\
 \hline
 $\text{Complete}$ & 0.37 & 1.37 & 3.48 & 2.72 & 0.37 & 1.37 & 3.48  & 2.72 \\
 
 $\text{Simplified}$ & 0.33 & 1.34 & 3.51  & 2.78 & 0.62 & 2.15 & 4.43 &3.57  \\
   $\text{NP Imp.}$ & 0.40 & 1.49 & 3.71  & 2.80 & 0.55 & 2.09 & 4.14  & 3.39  \\
  
 \hline
    \end{tabular}
\label{Tab:AE}
\end{table}

Figure \ref{CI} illustrates the trade-off between coverage rate and confidence interval (CI) length for the complete, simplified, and nonparametric imputed estimators under MAR of 15\% and 35\%. The estimator with completely observed data consistently achieves a coverage rate of approximately 96\% with the shortest CI lengths ($\approx 0.08$), highlighting its efficiency in maintaining both narrow intervals and reliable coverage. The simplified estimator achieves the highest coverage rate ($\approx 99\%$) at both MAR levels but at the cost of significantly longer confidence intervals, particularly for MAR = 35\%, where the CI length exceeds 0.14. The nonparametric imputed estimator strikes a balance between the two, achieving slightly higher coverage rates than the complete estimator (around 97\%) while maintaining moderate CI lengths ($\approx 0.1$). In conclusion, the estimator with complete data is offering compact intervals, the simplified estimator prioritizing high coverage, and the nonparametric imputed estimator is providing a balance between coverage and CI length.

 \begin{figure}[h!]%
    \centering
    \includegraphics[width=12cm,height=9cm]{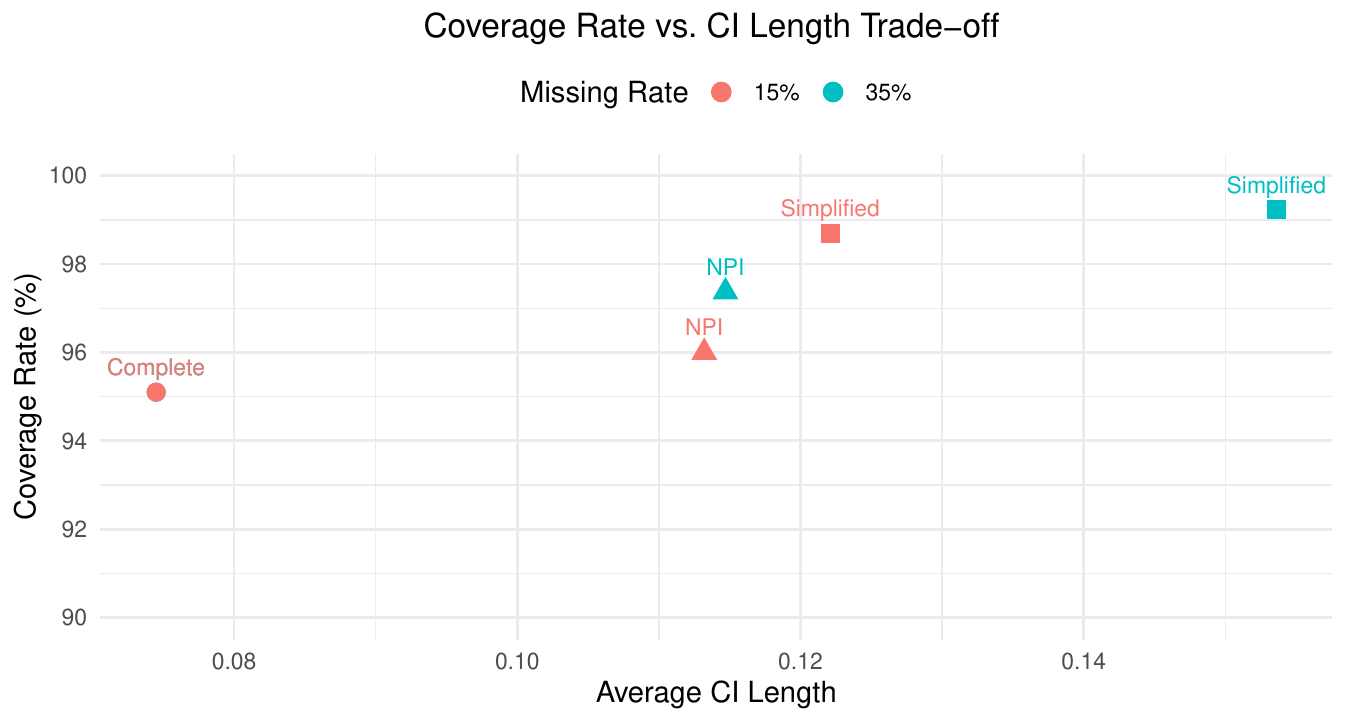} %
    \caption{Coverage rate versus CI length trade-off for volatility prediction under MAR rates of $15\%$ and $35\%.$}%
    
    \label{CI}%
\end{figure}

    
   


\section{Discussion and future work}\label{sec6}
In this paper, we investigated the nonparametric estimation of the conditional volatility of a real-valued response variable using a functional random variable as predictors, with a focus on addressing missing data challenges under a missing at random mechanism. Our methodology extended existing functional regression frameworks by incorporating both simplified and nonparametric imputed estimators for the regression and volatility operators. Detailed asymptotic properties were established for each estimator. Through simulation studies, we demonstrated the superior accuracy and efficiency of the imputed estimators, particularly under high MAR rates, highlighting their robustness in challenging scenarios. Real-world applications to energy commodities confirmed the practical utility of the proposed approach, effectively modeling conditional volatility while accounting for critical economic and geopolitical factors. Moreover, by constructing confidence intervals, we provided a comprehensive framework to evaluate the reliability of our estimators. Overall, this work contributes to advancing nonparametric techniques for volatility modeling, offering valuable insights and tools for both theoretical research and applied financial econometrics.
\\
This work can be extended in several directions. One potential avenue is to reduce the dimensionality of the functional predictor by employing a single functional index model (SFIM) for volatility estimation. SFIM has demonstrated its effectiveness in enhancing the consistency of regression operator estimators. However, to the best of our knowledge, its application to conditional variance estimation remains unexplored. Furthermore, in many contexts, particularly in finance and economics, volatility may depend on multiple predictors rather than a single functional covariate. Including additional real-valued covariates that capture economic or geopolitical conditions could provide a more comprehensive explanation of high-volatility periods. In this regard, it would be valuable to extend the current results to the semi-functional partial linear regression model (as introduced by \citet{Aneiros2006}), specifically addressing scenarios where the error terms are heteroscedastic. Such advancements could significantly broaden the applicability and robustness of the proposed methodology.
\section*{Proof of main results}
In order to prove the main results of this paper, we introduce additional notations and necessary Lemmas. Let $j \in \{1, 2 \}.$ Let
\begin{eqnarray}
  m^{[j]}_{n,0}(x) &=& \frac{1}{n\mathbb{E}(K_{1}(x))}\sum_{t=1}^{n}Y^{j-1}_{t}\delta_{t}K_{t}(x),\label{*} \\
  \overline{m}^{[j]}_{n,0}(x) &=& \frac{1}{n\mathbb{E}(K_{1}(x))}\sum_{t=1}^{n}\mathbb{E}[Y^{j-1}_{t}\delta_{t}K_{t}(x)|\F_{t-1}],\notag \label{**}
\end{eqnarray}
where
$K_{t}(x)=K\left(\frac{d_{1}(x,X_{t})}{h_{1}}\right)$,\\
and
\begin{equation}\label{Vn0}
V_{n,0}(x)=\frac{1}{n \mathbb{E}\left[K_{1}(x)\right]} \sum_{t=1}^{n} \delta_{t}\sqrt{U(X_{t})}\varepsilon_{t}K_{t}(x)\text{.}\color{black}
\end{equation}
\begin{lem} [\cite{LL16}]\label{Lemma 6.1}
Let $(M_{n})_{n\geq 1} $ be a sequence of martingale differences with respect to the sequence of $\sigma$-fields $\F_{n}=\{\sigma(M_{1},\ldots, M_{n}):
 n \geq 1 \} $, where $\sigma (M_{1},\ldots, M_{n})$ is the $\sigma$-field generated by the random variables $M_{1},\ldots, M_{n}$. Set $S_{n}=\displaystyle\sum_{t=1}^{n}M_{t}$. Assume that there exist for any $ t \in \mathbb{N}$ some nonnegative constants $C$ and $d_{t}$ such that $|M_{t}| \leq C$ almost surely, and
 $\mathbb{E}(M_{t}^{2}| \F_{t-1}) \leq d_{t}^{2} $
almost surely.\\ Then, for any $\varepsilon >0$,
\begin{equation*}
\mathbb{P}(|S_{n}|>\varepsilon)\leq 2\exp\left\lbrace\frac{-\varepsilon^{2}}{(4D_{n}+2C\varepsilon)}\right\rbrace,
\end{equation*}
where $D_{n}=\sum_{t=1}^{n}d_{t}^{2}.$
\end{lem}
\begin{lem}[\cite{LL10}]\label{Lemma 6.2}
Suppose that assumptions (A1)(i), (A2)(i)-(ii) and (A2)(iv) hold true. For $k \in \{1, 2, 3\} $, let
$\mathcal{K}_{t,k}(x)=\mathcal{K}\left(\frac{d_{k}(x,X_{t})}{h_{n,k}}\right)$ and  $M_{j,\mathcal{K},k}=\mathcal{K}^{j}(1)-\displaystyle\int_{0}^{1}(\mathcal{K}^{j})^{\prime}\tau_{0,k}(u)du$, where $1\leq j \leq 2+\kappa$, with $\kappa >0$, we have
\begin{itemize}
\item[$(i)$] $(\phi_{k}(h_{n,k}))^{-1}\mathbb{E}(\mathcal{K}_{t,k}^{j}(x)|\F _{t-1}) = M_{j,\mathcal{K} ,k}f_{t,1}(x) + \mathcal{O}_{a.s.}\left(\psi_{t,x}(h_{n,k})/\phi_{k}(h_{n,k})\right)$;
\item[$(ii)$] $(\phi_{k}(h_{n,k}))^{-1}\mathbb{E}(\mathcal{K}_{1,k}^{j}(x)) = M_{j,\mathcal{K} ,k}f_{1}(x) + o(1)$.
\end{itemize}
\end{lem}
\begin{lem}\label{Lemma 7.1.}
Suppose that assumptions (A1),(A2)(i)-(vi),(A3)(v),  and condition (\ref{2}) is satisfied.  
Then, 
\begin{equation}\label{sup1}
\underset{n\rightarrow\infty}{\lim}\sup_{x\in \mathcal{C}}\left|m_{n,0}^{[1]}(x)-\pi(x)\right|=0 \quad  a.s.
\end{equation}
\end{lem}
\begin{proof}
Note that
\begin{equation}\label{sup****}
\sup_{x\in \mathcal{C}}| m^{[1]}_{n,0}(x)|\leq\sup_{x\in \mathcal{C}}|m^{[1]}_{n,0}(x)-\overline{m}^{[1]}_{n,0}(x)\color{black}|+\sup_{x\in \mathcal{C}}|\overline{m}^{[1]}_{n,0}(x)|.
\end{equation}
Then, the first term in right-hand side of (\ref{sup****}) could be decomposed as follows.  For $\eta >0$ and $B(c_{k}, \eta):=\{x \in \mathcal{C}: d_{\mathcal{C}}(x, c_{k})<\eta \}$, we have
\begin{eqnarray}\label{decomp1}
\displaystyle{\sup_{x\in \mathcal{C}}}\left|m_{n,0}^{[1]}(x)-\overline{m}_{n,0}^{[1]}(x)\right|&\leq& \displaystyle{\max_{1\leq k \leq N(\eta, \mathcal{C}, d_{\mathcal{C}})}\sup_{x \in B(c_{k}, \eta)}}\left|m_{n,0}^{[1]}(x)-\overline{m}_{n,0}^{[1]}(x) \right| \nonumber\\
& \leq&  \displaystyle{\max_{1\leq k \leq N(\eta, \mathcal{C}, d_{\mathcal{C}})}\sup_{x \in B(c_{k}, \eta)}}\left|m_{n,0}^{[1]}(x)-m_{n,0}^{[1]}(c_{k}) \right|+\displaystyle{\max_{1\leq k \leq N(\eta, \mathcal{C}, d_{\mathcal{C}})}}\left|m_{n,0}^{[1]}(c_{k})-\overline{m}_{n,0}^{[1]}(c_{k}) \right|\nonumber\\
& & \quad +\displaystyle{\max_{1\leq k \leq N(\eta, \mathcal{C}, d_{\mathcal{C}})}\sup_{x \in B(c_{k}, \eta)}}\left|\overline{m}_{n,0}^{[1]}(x)-\overline{m}_{n,0}^{[1]}(c_{k}) \right|  \nonumber\\
&=: &\mathcal{H}_{1}+\mathcal{H}_{2}+\mathcal{H}_{3}.
\end{eqnarray}
We start with studying the term $\mathcal{H}_{1}.$
Observe that, for any $x \in B(c_{k}$, $\eta)$ and making use of the definition of $m_{n,0}^{[1]}(x)$ in \eqref{*}, one gets
$$\begin{aligned}
m_{n,0}^{[1]}(x)-m_{n,0}^{[1]}(c_{k})&=\frac{1}{n \mathbb{E}(K_{1}(x))\mathbb{E}(K_{1}(c_{k}))}\displaystyle{\sum _{t=1}^{n}}\delta_{t}\left[K_{t}(x)\mathbb{E}(K_{1}(c_{k}))-K_{t}(c_{k})\mathbb{E}(K_{1}(x)\right]\\
&=\frac{1}{n \mathbb{E}(K_{1}(x))}\displaystyle{\sum _{t=1}^{n}}\delta_{t}\left[K_{t}(x)-K_{t}(c_{k})\right]\\
& \quad +\frac{1}{n \mathbb{E}(K_{1}(x))\mathbb{E}(K_{1}(c_{k}))}\displaystyle{\sum _{t=1}^{n}}\delta_{t}K_{t}(c_{k})\left[\mathbb{E}(K_{1}(c_{k}))-\mathbb{E}(K_{1}(x))\right]\\
&= \mathcal{H}_{1,1}+\mathcal{H}_{1,2}.
\end{aligned}
$$
Condition (A1)(ii) and the fact that $|\delta_{t}|<1$ almost surely, one can bound the term $\mathcal{H}_{1,1}$ from above as follows: \\
$\begin{aligned}
\left|\mathcal{H}_{1,1}\right|
&\leq \frac{1}{n \left|\mathbb{E}(K_{1}(x))\right|}\displaystyle{\sum _{t=1}^{n}}\left| \delta_{t}\right| \left|K_{t}(x)-K_{t}(c_{k})\right|\\
& \leq \frac{1\color{black}}{n \left|\mathbb{E} (K_{1}(x))\right|} \displaystyle{\sum _{t=1}^{n}}|K_{t}(x)-K_{t}(c_{k})|\\
& \leq \frac{ 1\color{black}}{n \left|\mathbb{E} (K_{1}(x))\right|}  \left( a_{0}\displaystyle{\sum _{t=1}^{n}}\left|\frac{d_{1}\color{black}(x,X_{t})-d_{1}\color{black}(c_{k},X_{t})}{h_{1}}\right|^{\gamma}\right)\\
&\leq \frac{ 1\color{black}}{n \left|\mathbb{E}(K_{1}(x))\right|}  \left( a_{0} n \left[\frac{\eta\color{black}}{h_{1}}\right]^{\gamma}\right)\\
&\leq \frac{  \eta^{\gamma}a_{0}}{h_{1}^{\gamma}|\mathbb{E}(K_{1}(x))|}.
 \end{aligned}$\\
On the other hand, by condition (A1)(iii) and Lemma \ref{Lemma 6.2}, we have
$
\frac{K_{t}(x)}{|\mathbb{E}(K_{1}(x))|}\leq \frac{a_{2}}{a_{1}}=:a_{3}, \; \forall x \in \mathcal{E}.
$
Hence
\begin{equation*}
\left|\mathcal{H}_{1,2}\right|
\leq \frac{1}{n \mathbb{E}(K_{1}(x))}\displaystyle{\sum _{t=1}^{n}}\left|\delta_{t}\right| \left|\frac{K_{t}(c_{k})}{\mathbb{E}(K_{1}(c_{k}))}\right|  \left|\mathbb{E}(K_{1}(c_{k}))-\mathbb{E}(K_{1}(x))\right|.
\end{equation*}
Using condition (A1)(ii)  and the almost sure boundedness of $\delta$, we get 
$
\left|\mathcal{H}_{1,2}\right|
 \leq \frac{ a_{3}\eta\color{black}^{\gamma}a_{0}}{h_{1}^{\gamma}\left|\mathbb{E}(K_{1}(x))\right|}.$
Therefore, for any $1\leq k \leq N(\eta, \mathcal{C}, d_{\mathcal{C}})$, we have $|m_{n,0}^{[1]}(x)-m_{n,0}^{[1]}(c_{k}) | \leq \frac{(1+a_{3})a_{0}\eta^{\gamma}}{h_{1}^{\gamma}a_{1}}.$
Then, by taking $\eta=\eta_{n}=o(h_{1})$, one gets
\begin{equation}\label{h1}
\mathcal{H}_{1}=\displaystyle{\max_{1\leq k \leq N(\eta, \mathcal{C}, d_{\mathcal{C}})}\sup _{x \in B(c_{k}, \eta)}}\left|m_{n,0}^{[1]}(x)-m_{n,0}^{[1]}(c_{k})\right|=o(1).
\end{equation}
Similar to $\mathcal{H}_{1}$, one can show that when (A3)(v) holds true, the term  $\mathcal{H}_{3}$   in \eqref{decomp1} is negligible. That is
\begin{equation}\label{h3}
\mathcal{H}_{3}=\displaystyle{\max_{1\leq k \leq N(\eta, \mathcal{C}, d_{\mathcal{C}})}\sup _{x \in B(c_{k}, \eta\color{black})}}\left|\overline{m}_{n,0}^{[1]}(x)-\overline{m}_{n,0}^{[1]}(c_{k})\right|=o(1).
\end{equation}
Next, we turn our attention to the study of the term $\mathcal{H}_{2}.$    Note that
\begin{equation*}
 m_{n,0}^{[1]}(c_{k})-\overline{m}_{n,0}^{[1]}(c_{k})=\frac{1}{n\mathbb{E}(K_{1}(c_{k}))}\sum _{t=1}^{n}L_{n,t}(c_{k}),
\end{equation*}
where $L_{n,t}(c_{k})=\delta_{t}K_{t}(c_{k})-\mathbb{E}(\delta_{t}K_{t}(c_{k})|\F_{t-1})$ is a martingale difference with respect to the $\sigma$-field $\F_{t-1}$, for $t= 1, 2,\ldots, n.$  Then, we can use Lemma \ref{Lemma 6.1} to  handle  the convergence of $m_{n,0}^{[1]}(c_{k})-\overline{m}_{n,0}^{[1]}(c_{k})$.  To be able to use this  Lemma,  we need first to check its conditions. \\ 
First, since $|\delta_{t}|<1$ almost surely and  by using assumptions (A1)(iii), it follows that \\  
$$\begin{aligned}
|L_{n,t}(c_{k})| 
& \leq 2\left|\delta_{t}K_{t}(c_{k})\right|\\ & \leq 2a_{2}=:C.
\end{aligned}$$
In addition, making use of $C_{r}$-inequality,  assumptions (A1)(i), (A2)(iv), (A3)(v), Lemma \ref{Lemma 6.2}(i)  and 
 the fact that $f_{t,1} $ is bounded
by a  deterministic quantity  $b_{t}(x)$, one gets \\
$\begin{aligned}
\mathbb{E}(L^{2}_{n,t}(c_{k})|\mathcal{F}_{t-1})
&\leq 2\mathbb{E}(\delta_{t}K^{2}_{t}(c_{k})|\F_{t-1}) \\
& \leq 2 (\pi(x)+o(1))\mathbb{E}(K^{2}_{t}(c_{k})|\F_{t-1}) \\
& \leq  (2\pi(x)+o(1))\phi_{1}(h_{1})(M_{2,K,1}b_{t}(x)+1)=:d^{2}_{t}. 
\end{aligned}$ \\
Then, by assumption (A2)(v), we can write  $n^{-1}D_{n}=(2\pi(x)+o(1)) \phi_{1}(h_{1})[M_{2,K,1}D(x)+o_{a.s.}(1)]$, which  means that $D_{n}=\mathcal{O}(n\phi_{1}(h_{1}))$. 
Thus, by applying Lemma \ref{Lemma 6.1}, we have 
$$\begin{aligned}
\mathbb{P}\left(\displaystyle{\max_{1\leq k \leq N(\eta, \mathcal{C}, d_{\mathcal{C}})}}\left|m_{n,0}^{[1]}(c_{k})-\overline{m}_{n,0}^{[1]}(c_{k})\right|>\overline{\lambda}_{n}\right)& \leq \displaystyle{\sum_{k=1}^{N(\eta, \mathcal{C}, d_{\mathcal{C}})}}\mathbb{P}\left(\left|m_{n,0}^{[1]}(c_{k})-\overline{m}_{n,0}^{[1]}(c_{k})\right|>\overline{\lambda}_{n}\right)\\
& \leq \displaystyle{\sum_{k=1}^{N(\eta, \mathcal{C}, d_{\mathcal{C}})}}\mathbb{P}\left(\left|\sum_{t=1}^n L_{n, t}(c_{k})\right|> \overline{\lambda}_{n} n\mathbb{E}[K_{1}(c_{k})]\right)\\
&\leq 2N(\eta, \mathcal{C}, d_{\mathcal{C}}) \exp \left( -\frac{\left(\overline{\lambda}_{n} n\mathbb{E}[K_{1}(c_{k})]\right)^{2}}{4D_{n}+2C \overline{\lambda}_{n} n\mathbb{E}[K_{1}(c_{k})] } \right) \\
&\leq 2N(\eta, \mathcal{C}, d_{\mathcal{C}}) \exp \left( -\frac{\mathcal{O}(n\phi_{1}(h_{1}))^{2}\overline{\lambda}_{n}^{2}}{\mathcal{O}(n\phi_{1}(h_{1}))+\overline{\lambda}_{n} \mathcal{O}(n\phi_{1}(h_{1})) } \right)\\
& \leq 2N(\eta, \mathcal{C}, d_{\mathcal{C}}) \exp \left( -\frac{\mathcal{O}(n\phi_{1}(h_{1})) \overline{\lambda}_{n}^{2}}{1+\overline{\lambda}_{n}} \right)\\
& \leq 2\exp \left( -\mathcal{O}(n\phi_{1}(h_{1}))\overline{\lambda}_{n}^{2} \left[1-\frac{\log N(\eta, \mathcal{C}, d_{\mathcal{C}}) }{\mathcal{O}(n\phi_{1}(h_{1}))\overline{\lambda}_{n}^{2}}\right]\right),
\end{aligned}
$$ 
where $\overline{\lambda}_{n}=\lambda_{n}\ell_n^{-1}$, $\lambda_{n}$ and $\ell_n$ are defined in Theorem \ref{T1}. Hence,  using condition (\ref{2}), one gets
$$
\displaystyle{\sum_{n\geq 1}}\mathbb{P}\left( \displaystyle{\max_{1\leq k \leq N(\eta, \mathcal{C}, d_{\mathcal{C}})}}\left|m_{n,0}^{[1]}(c_{k})-\overline{m}_{n,0}^{[1]}(c_{k})\right|>\overline{\lambda}_{n}\right)<\infty.
$$ 
Then, using Borel-Cantelli Lemma we have 
\begin{equation}\label{sup21++}
  \displaystyle{\max_{1\leq k \leq N(\eta, \mathcal{C}, d_{\mathcal{C}})}}\left|m_{n,0}^{[1]}(c_{k})-\overline{m}_{n,0}^{[1]}(c_{k})\right|= \mathcal{O}_{a.s.}(\overline{\lambda}_{n})=o_{a.s.}(1).
 \end{equation}
Combining \eqref{decomp1} with \eqref{h1}, \eqref{h3} and \eqref{sup21++}, one deduces that the first term in the right-hand side of inequality \eqref{sup****} converges to zero as $n$ goes to infinity. Regarding the second term  $\displaystyle\sup_{x\in \mathcal{C}}|\overline{m}^{[1]}_{n,0}(x)|$, note that assumption(A3)(v), along with a double conditioning with respect to $\mathcal{G}_{t-1}$, allows to obtain
$$\begin{aligned}\overline{m}^{[1]}_{n,0}(x)-\pi(x)&=\frac{1}{n \mathbb{E}( K_1(x))} \sum_{t=1}^n \mathbb{E}\left\{\mathbb{E}\left[\delta_t K_t(x) \mid \G_{t-1}\right] \mid \F_{t\color{black}-1}\right\}-\pi(x)\\
&=\frac{1}{n \mathbb{E} (K_1(x))} \sum_{t=1}^n \mathbb{E}[(\pi(x)+o(1)) K_t(x) \mid \F_{t-1}]-\pi(x)\\
&=(\pi(x)+o(1)) \frac{1}{n \mathbb{E} (K_1(x))} \sum_{t=1}^n \mathbb{E}\left(K_t(x) \mid \F_{t-1}\right)-\pi(x)\\
&=\pi(x)\left[\frac{1}{n \mathbb{E} (K_1(x))} \sum_{t=1}^n \mathbb{E}\left(K_t(x) \mid \F_{t-1}\right)-1\right]+o(1),
\end{aligned}$$
 where $o(1)$ is uniformly in $x$. Finally, Lemma 7 in \cite{LL11} allows to conclude the proof of this Lemma.
 \end{proof}
\begin{lem}\label{Lemmacvn}
Suppose that assumptions (A1), (A2)(i)-(ii),(iv)-(vi), (A3)(i),(v), and conditions (\ref{1}) and  (\ref{2}) are satisfied. Then, we have 
 \begin{equation*}
\displaystyle{\sup_{x\in \mathcal{C}}}\left|V_{n,0}(x)\right|=\mathcal{O}_{a.s.}\left(\lambda_{n}\right) \,\ \hbox{as} \quad n \rightarrow \infty \text{.}
\end{equation*}
\end{lem}
 \begin{proof} Let 
\begin{equation}\label{DecVn0}
V_{n,0}(x)=\frac{1}{n \mathbb{E}\left[K_{1}(x)\right]} \sum_{t=1}^{n} \mathcal{L}_{t} K_{t}(x)= \left(  V_{n,0}(x)-V^{\top}_{n,0}(x)  \right)  + \widetilde{V}_{n,0}(x) + V^{-}_{n,0}(x),
\end{equation}
where \\
 $ V^{\top}_{n,0}(x)=\displaystyle{\frac{1}{n \mathbb{E}\left[K_{1}(x)\right]} \sum_{t=1}^{n} \mathcal{L}_{t} \1_{\left(|\mathcal{L}_{t}|\leq\ell_{n}\right)} K_{t}(x)}, \quad
\widetilde{V}_{n,0}(x)= \displaystyle{\frac{1}{n \mathbb{E}\left[K_{1}(x)\right]} \sum_{t=1}^{n}\mathbb{E}( \mathcal{L}_{t} \1_{\left(|\mathcal{L}_{t}|\leq\ell_{n}\right)}\mid \mathcal{F}_{t-1}) K_{t}(x)},
$ \\
$V^{-}_{n,0}(x)=\displaystyle{\frac{1}{n \mathbb{E}\left[K_{1}(x)\right]} \sum_{t=1}^{n} \left(  \mathcal{L}_{t} \1_{\left(|\mathcal{L}_{t}|\leq\ell_{n} \right)}-\mathbb{E}( \mathcal{L}_{t} \1_{\left(|\mathcal{L}_{t}|\leq\ell_{n}\right)}\mid \mathcal{F}_{t-1})\right)\color{black} K_{t}(x)},$ \\ such that $\mathcal{L}_{t}=\delta_{t}\sqrt{U(X_{t})}\varepsilon_{t}$  and $\ell_{n}$ is as defined in Theorem \ref{T1}.
\\
Lemma \ref{LemmaA} allows to conclude that the first term in (\ref{DecVn0}) equals zero almost surely as $n\rightarrow \infty.$ Moreover, making use of Lemma \ref{LemmaB} and Lemma \ref{LemmaB+} below, the second term is $\mathcal{O}\left(\{\ell^{\rho-1}_{n}\phi_{1}(h_{1})^{\rho-1/\rho}\}^{-1}\right)$ and the third term is $\mathcal{O}_{a.s.}\left(\lambda_{n}\right),$ respectively. Finally, since $(\lambda_{n}\ell^{\rho-1}_{n}\phi_{1}(h_{1})^{\rho-1/\rho})^{-1}=o(1)$ for n large enough, the proof of this lemma is achieved.
 \end{proof}
\begin{lem}\label{LemmaB}
Suppose that assumptions  (A1)(i), (A2)(i)-(ii),(iv),(vi), (A3)(i),(v) and condition (\ref{1}) are satisfied. Then
\begin{equation*}
\displaystyle{\sup_{x\in \mathcal{C}}}|\widetilde{V}_{n,0}(x)| = \mathcal{O}\left(\{\ell^{\rho-1}_{n}\phi_{1}(h_{1})^{\rho-1/\rho}\}^{-1}\right).
\end{equation*}
\end{lem}
The proof of this Lemma is detailed in the Appendix.
\begin{lem}\label{LemmaB+}   Under assumptions  (A1), (A2)(i)-(ii),(iv)-(vi), (A3)(i),(v)   and condition (\ref{2}), we have, as $n \rightarrow \infty,$
\begin{equation*}
\displaystyle{\sup_{x\in \mathcal{C}}}|V^{-}_{n,0}(x)| =\mathcal{O}_{a.s.}(\lambda_{n}).
\end{equation*}
\end{lem}
\begin{proof} The proof is similar to the proof of Lemma C in \cite{Chaouch2019}. \end{proof}

\subsection*{Proof of Theorem \ref{T1}}
By Model (\ref{model}), we have
\begin{equation*}
m_{n,0}(x)-m(x)=\frac{1}{m_{n,0}^{[1]}(x)}\left[\frac{1}{n \mathbb{E} (K_1(x))}\displaystyle{\sum_{t=1}^{n}}
\delta_{t}\{m(X_{t})-m(x)+\sqrt{U(X_{t})}\varepsilon_{t}\}K_{t}(x)\right].
\end{equation*}
Then,
\begin{equation}\label{Decosup2}
\displaystyle{\sup_{x\in \mathcal{C}}}|m_{n,0}(x)-m(x)|\leq \displaystyle{\sup_{u\in B(x,h_{1})}}|m(u)-m(x)|+\{\displaystyle{\inf_{x\in \mathcal{C}}}m_{n,0}^{[1]}(x)\}^{-1}\times \displaystyle{\sup_{x\in \mathcal{C}}}|V_{n,0}(x)|.
\end{equation}
Further, observe that
\begin{equation*}
\displaystyle{\inf_{x\in \mathcal{C}}}\left|m_{n,0}^{[1]}(x)\right|>\displaystyle{\inf_{x\in \mathcal{C}}}\left|\pi(x)\right|-\displaystyle{\sup_{x\in \mathcal{C}}}\left|m_{n,0}^{[1]}(x)-\pi(x)\right|.
\end{equation*}
Thus, making use of Lemma \ref{Lemma 7.1.} and  assumption (A2)(vii), we get
 \begin{equation}\label{***+}
 \displaystyle{\inf_{x\in \mathcal{C}}}\left|m_{n,0}^{[1]}(x)\color{black}\right|>\theta_{1}\,\ a.s.,
\end{equation}
where $\theta_{1}$ is defined in assumption (A2)(vii).\\
Finally, using Lemma \ref{Lemmacvn}, combined with  assumption (A4)(i), equation (\ref{Decosup2}) and equation (\ref{***+}), one concludes the proof of this Theorem.  

 \subsection*{Proof of Theorem \ref{T2}}
 \quad  Let 
 \begin{equation}\label{decompvar}
U_{n,0}(x)-U(x)=\left(\vartheta_{n,0}^{[1]}(x)+\vartheta_{n,0}^{[2]}(x)+\vartheta_{n,0}^{[3]}(x)+\vartheta_{n,0}^{[4]}(x)\right)/U_{n,0}^{[1]}(x),
 \end{equation}
 where
 \begin{eqnarray}
\vartheta_{n,0}^{[1]}(x)&=&\displaystyle\frac{1}{n \mathbb{E}\left[W_{1}(x)\right]} \sum_{t=1}^{n}\delta_{t}\left(m\left(X_{t}\right)-m_{n,0}\left(X_{t}\right)\right)^{2} W_{t}(x), \label{Vn1}\\
\vartheta_{n,0}^{[2]}(x)&=&\displaystyle\frac{2}{n \mathbb{E}\left[W_{1}(x)\right]} \sum_{t=1}^{n}\delta_{t}\left(m\left(X_{t}\right)-m_{n,0}\left(X_{t}\right)\right) \sqrt{U\left(X_{t}\right)} \varepsilon_{t} W_{t}(x), \label{Vn2} \\
\vartheta_{n,0}^{[3]}(x)&=&\displaystyle\frac{1}{n \mathbb{E}\left[W_{1}(x)\right]} \sum_{t=1}^{n} \delta_{t}U\left(X_{t}\right) W_{t}(x)\left(\varepsilon_{t}^{2}-1\right), \label{Decoimp}\\
\vartheta_{n,0}^{[4]}(x) &=&\displaystyle\frac{1}{n \mathbb{E}\left[W_{1}(x)\right]} \sum_{t=1}^{n}\delta_{t}\left(U\left(X_{t}\right)-U(x)\right) W_{t}(x), \notag\\
U_{n,0}^{[1]}(x)&=&\frac{1}{n \mathbb{E}\left[W_{1}(x)\right]} \sum_{t=1}^{n} \delta_{t} W_{t}(x),\notag
\end{eqnarray}
with $W_{t}(x)=W\left(\frac{d_{2}(x,X)}{h_{2}}\right).$ \color{black} Our main objective is to establish the uniform consistency rate, with respect to $x,$ of $U_{n,0}(x)$. For this purpose, let us consider
\begin{equation}\label{15*}
\displaystyle{\sup_{x\in \mathcal{C}}}\left| U_{n,0}(x)-U(x)\right| \leq  \big\{\displaystyle{\sup_{x\in \mathcal{C}}}|\vartheta_{n,0}^{[1]}(x)|+\displaystyle{\sup_{x\in \mathcal{C}}}|\vartheta_{n,0}^{[2]}(x)|+\displaystyle{\sup_{x\in \mathcal{C}}}|\vartheta_{n,0}^{[3]}(x)|+\displaystyle{\sup_{x\in \mathcal{C}}}|\vartheta_{n,0}^{[4]}(x)| \big\} \times \big\{\displaystyle{\inf_{x\in \mathcal{C}}}|U_{n,0}^{[1]}(x)|\big\}^{-1}.
 \end{equation}
Note that $U_{n,0}^{[1]}(x)$ has similar form as $m_{n,0}^{[1]}(x)$ when $K$ is replaced by $W$. 
Thus, similar to the proof of (\ref{sup1}) and (\ref{***+}), one can show that, under assumptions  (A1), (A2), (A3)(v) and condition (\ref{4}), one  gets 
\begin{equation}\label{supU1}
\underset{n\rightarrow\infty}{\lim}\sup_{x\in \mathcal{C}}\left|U_{n,0}^{[1]}(x)-\pi(x)\right|=0 \quad  \hbox{a.s.}, \color{black}
\end{equation}
and 
\begin{equation}\label{***u}
\displaystyle{\inf_{x\in \mathcal{C}}}\left|U_{n,0}^{[1]}(x)\right|> \theta_{1} \quad \hbox{a.s.}\,\
\end{equation}
\textit{Study of the term $\vartheta_{n,0}^{[1]}(x)$.} We have
$$\left|\vartheta_{n,0}^{[1]}(x)\right| \leq \frac{1}{n \mathbb{E}\left[W_{1}(x)\right]} \sum_{t=1}^{n}\delta_{t}\left|m\left(X_{t}\right)-m_{n,0}\left(X_{t}\right)\right|^{2} W_{t}(x) \leq \sup _{x \in \mathcal{C}}\left|m_{n,0}(x)-m(x)\right|^{2} \times \frac{1}{n \mathbb{E}\left[W_{1}(x)\right]} \sum_{t=1}^{n}\delta_{t} W_{t}(x)\text{.}
$$
Then, in view of Theorem \ref{T1} and   equation (\ref{supU1}), we get
\begin{equation}\label{32}
\displaystyle{\sup_{x\in \mathcal{C}}}\left|\vartheta_{n,0}^{[1]}(x)\right|=\mathcal{O}_{a.s.}(h_{1}^{2\alpha})+
\mathcal{O}_{a.s.}\left(\lambda_{n}^{2}\right)\text{.}
\end{equation}
\textit{Study of the term $\vartheta_{n,0}^{[2]}(x)$.} Observe that
$$
\vartheta_{n,0}^{[2]}(x) \leq 2 \sup _{x \in \mathcal{C}}\left|m_{n,0}(x)-m(x)\right| \times \overline{V}_{n,0}(x),\,\
$$
where
\begin{equation*}\label{BarVn0}
\overline{V}_{n,0}(x)=\frac{1}{n \mathbb{E}\left[W_{1}(x)\right]} \sum_{t=1}^{n} \delta_{t}\sqrt{U(X_{t})}\varepsilon_{t}W_{t}(x)\text{.}
\end{equation*}
In addition, note that $\overline{V}_{n,0}(x)$ has similar form as $V_{n,0}(x)$ when $K$ is replaced by $W$. Therefore, by Lemma \ref{Lemmacvn}, it follows that
\begin{equation}\label{Vnbar}
\displaystyle{\sup_{x\in \mathcal{C}}}\left|\overline{V}_{n,0}(x)\right|=\mathcal{O}_{a.s.}\left(\lambda_{n}^{\prime}\right),\,\
\end{equation}
where $\lambda_{n}^{\prime}$ is defined in Theorem \ref{T2}.
\color{black}
\\Making use of Theorem \ref{T1} and   equation  (\ref{Vnbar}), we find
\begin{equation}\label{33}
\displaystyle{\sup_{x\in \mathcal{C}}}\left|\vartheta_{n,0}^{[2]}(x)\right|=\mathcal{O}_{a.s.}\left\{\left(h_{1}^{\alpha}+\lambda_{n}\right)\lambda_{n}^{\prime}\right\}\text{.}
\end{equation}
 \textit{Study of the term $\vartheta_{n,0}^{[3]}(x).$} 
Observe that $\vartheta_{n,0}^{[3]}(x)$ has the same form as $\overline{V}_{n,0}(x)$  when $\varepsilon$ and $\sqrt{U}$ are replaced by $\varepsilon^{2}-1$ and $U,$ respectively.
 Similar to the proof of Lemma \ref{Lemmacvn},  under assumptions 
  (A1), (A2)(i)-(ii), (iv)-(vi), (A3)(iv)-(v)  and condition (\ref{4}), we obtain
\begin{equation}\label{34}
\displaystyle{\sup_{x\in \mathcal{C}}}\left|\vartheta_{n,0}^{[3]}(x)\right|=\mathcal{O}_{a.s.}\left(\lambda_{n}^{\prime}\right)\text{.}
\end{equation}  \color{black}
\textit{Study of the term $\vartheta_{n,0}^{[4]}(x).$} One can easily show, using assumption (A4)(ii) and  equation   (\ref{supU1}), that
\begin{equation}\label{35}
\displaystyle{\sup_{x\in \mathcal{C}}}\left|\vartheta_{n,0}^{[4]}(x)\right|=\mathcal{O}_{a.s.}\left(h_{2}^{\beta}\right)\text{.}
\end{equation}
Finally,  using  equation  (\ref{15*}) combining   with  the results obtained in (\ref{***u})-(\ref{32}) and (\ref{33})-(\ref{35}), we conclude the proof of the theorem .

The following Lemma gives the asymptotic normality of $\vartheta_{n,0}^{[3]}(x)$ which is needed to prove Theorem \ref{T3}.

\begin{lem}\label{Lemma-T3}
Suppose that assumptions (A1)(i), (A2)(i)-(iv), (vi), (A3)(ii)-(iii), (v), (A4)(ii)-(iii) and condition (\ref{3}) are satisfied. Then, as $n \to \infty$, we have, 
\begin{equation*}
\sqrt{n \phi_{2}(h_{2})}\vartheta_{n,0}^{[3]}(x)  \overset{\mathcal{D}}
{\longrightarrow}  \mathcal{N}\left(0, \sigma^{2}_1(x)\right),
\end{equation*} 
where $\sigma^{2}_1(x)$ is defined in Theorem \ref{T7}. 
\end{lem}
\begin{proof}  Define
$\xi _{n, t}=\{\sqrt{\phi_{2}\left(h_{2}\right) / n}\} \left[\delta_{t}U\left(X_{t}\right) W_{t}(x)\left(\varepsilon_{t}^{2}-1\right) / \mathbb{E}\left\{W_{1}(x)\right\}\right],$ $t=1,...,n.$ Further, observe that
\begin{equation}\label{36}
\sqrt{n \phi_{2}\left(h_{2}\right) }\vartheta_{n,0}^{[3]}(x)=\sum_{i=1}^{n} \xi_{n, t},
\end{equation}
where, for any $x\in \mathcal{E},$ the summands in $(\ref{36})$ form a triangular array of stationary martingale differences with respect to the $\sigma$-field $\mathcal{F}_{t-1}$.  Then, similar to  the proof of Lemma 4 in \cite{LL10}, we apply the Central Limit Theorem for discrete-time arrays of real-valued martingales to provide the asymptotic normality of $\vartheta_{n,0}^{[3]}(x)$ (see \citet{Hall1980}).
For that, we have to establish the following statements:
\begin{itemize}
\item[(i)] $\underset{n \rightarrow \infty}{\lim} \sum_{t=1}^{n} \mathbb{E}\left(\xi_{n, t}^{2} \mid \mathcal{F}_{t-1}\right) \stackrel{\mathbb{P}}{=}  \sigma^{2}_1(x),$
\item[(ii)] $n \mathbb{E}\left\{\xi_{n, t}^{2} \1_{\left(\left|\xi_{n, t}\right|>\zeta\right)}\right\}=o(1)$ holds for any $\zeta>0$.
\end{itemize}
\textbf{Proof of part (i).} Notice that\\
$\begin{aligned} \sum_{t=1}^{n} \mathbb{E}\left(\xi_{n, t}^{2} \mid \mathcal{F}_{t-1}\right)=& \frac{\phi_{2}\left(h_{2}\right)}{n\left[\mathbb{E}\left(W_{1}(x)\right)\right]^{2}} \sum_{t=1}^{n} \mathbb{E}\left[\delta_{t}\left\{U\left(X_{t}\right)-U(x)\right\}^{2} W_{t}^{2}(x)\left(\varepsilon_{t}^{2}-1\right)^{2} \mid \mathcal{F}_{t-1}\right] \\ &+\frac{\phi_{2}\left(h_{2}\right)}{n\left[\mathbb{E}\left(W_{1}(x)\right)\right]^{2}} \sum_{t=1}^{n} \mathbb{E}\left[\delta_{t}U^{2}(x) W_{t}^{2}(x)\left(\varepsilon_{t}^{2}-1\right)^{2} \mid \mathcal{F}_{t-1}\right] \\ &+\frac{2 \phi_{2}\left(h_{2}\right)}{n\left[\mathbb{E}\left(W_{1}(x)\right)\right]^{2}} \sum_{t=1}^{n} \mathbb{E}\left[\delta_{t}\left\{U\left(X_{t}\right)-U(x)\right\} U(x) W_{t}^{2}(x)\left(\varepsilon_{t}^{2}-1\right)^{2} \mid \mathcal{F}_{t-1}\right] \\ &\equiv  \Upsilon_{n, 1}+\Upsilon_{n, 2}+\Upsilon_{n, 3}. \end{aligned}$\\
 \textit{Regarding the term $\Upsilon_{n,1}.$}   Due to \eqref{pieq}, $\delta_t$ and $\varepsilon_{t}$  are independent. Then, by using Lemma \ref{Lemma 6.2} and  assumptions (A1)(i), (A2)(iii)-(iv), (vi), (A3)(ii), (v), and (A4)(ii),  we can show that:

$
\begin{aligned} |\Upsilon_{n,1}| &=\mathcal{O}_{a.s.}(h_{2}^{2\beta})\times\left\{\frac{\phi_{2}\left(h_{2}\right)}{n\left[\mathbb{E}\left\{W_{1}(x)\right\}\right]^{2}} \sum_{t=1}^{n} \mathbb{E}\left[\mathbb{E}\left\{\delta_{t} W_{t}^{2}(x)\left(\varepsilon_{t}^{2}-1\right)^{2} \mid \mathcal{G}_{t-1}\right\} \mid \mathcal{F}_{t-1}\right]\right\} \\ & \leq \mathcal{O}_{a.s.}(h_{2}^{2\beta})\times \left\{\pi(x)+\sup _{  u \in B(x,h)}|\pi(u)-\pi(x)|\right\}\left\{\frac{M_{2, W, 2}}{M_{1, W, 2}^{2}} \frac{1}{f_{1}(x)}+o_{a.s.}(1)\right\}\\
& \times\left\{\omega(x)+\sup _{u \in B(x,h)}|\omega(u)-\omega(x)|\right\}
 \longrightarrow 0 \,\ \hbox{as} \,\  n \to \infty. \\
 \end{aligned}
 $\\
\textit{Regarding the term $\Upsilon_{n,2}.$} By considering the independence of $\delta$ and $\varepsilon_{t},$ using  Lemma \ref{Lemma 6.2} and  assumptions (A1)(i), (A2)(iii)-(iv),(vi), (A3)(ii),(v),  we obtain

$
\begin{aligned} |\Upsilon_{n,2}| &=U^{2}(x)\times\left\{\frac{\phi_{2}\left(h_{2}\right)}{n\left[\mathbb{E}\left\{W_{1}(x)\right\}\right]^{2}} \sum_{t=1}^{n} \mathbb{E}\left[\mathbb{E}\left\{\delta_{t} W_{t}^{2}(x)\left(\varepsilon_{t}^{2}-1\right)^{2} \mid \mathcal{G}_{t-1}\right\} \mid \mathcal{F}_{t-1}\right]\right\} \\ & \leq U^{2}(x)\times\left\{\pi(x)+\sup _{u \in B(x,h)}|\pi(u)-\pi(x)|\right\}\left\{\frac{M_{2, W, 2}}{M_{1, W, 2}^{2}} \frac{1}{f_{1}(x)}+o_{a.s.}(1)\right\}\left\{\omega(x)+\sup _{u \in B(x,h)}|\omega(u)-\omega(x)|\right\}\\
&\longrightarrow  U^{2}(x)\pi(x) \left\{\frac{M_{2, W, 2}}{M_{1, W, 2}^{2}} \frac{1}{f_{1}(x)}\right\}\omega(x) \,\ \hbox{as} \,\ n \to \infty.
\end{aligned}$
\\
 \textit{Regarding \color{black} the term $\Upsilon_{n,3}$.}  Similarly as for $\Upsilon_{n,1},$ using assumptions  (A1)(i), (A2)(iii)-(iv),(vi), (A3)(ii),(v), and (A4)(ii) together with  Lemma \ref{Lemma 6.2},  we get

 $
\begin{aligned} |\Upsilon_{n,3}|&=\mathcal{O}_{a.s.}(h_{2}^{\beta})\times U(x)\left\{\frac{\phi_{2}\left(h_{2}\right)}{n\left[\mathbb{E}\left\{W_{1}(x)\right\}\right]^{2}} \sum_{t=1}^{n} \mathbb{E}\left[\mathbb{E}\left\{\delta_{t} W_{t}^{2}(x)\left(\varepsilon_{t}^{2}-1\right)^{2} \mid \mathcal{G}_{t-1}\right\} \mid \mathcal{F}_{t-1}\right]\right\} \\   &  \leq \mathcal{O}_{a.s.}(h_{2}^{\beta})\times U(x) \left\{\pi(x)+\sup _{u \in B(x,h)}|\pi(u)-\pi(x)|\right\}\left\{\frac{M_{2, W, 2}}{M_{1, W, 2}^{2}} \frac{1}{f_{1}(x)}+o_{a.s.}(1)\right\} \\
& \times \left\{\omega(x)+\sup _{u \in B(x,h)}|\omega(u)-\omega(x)|\right\}
\longrightarrow 0 \,\ \hbox{as} \,\ n \to \infty.
\end{aligned}$ 
\\
 Finally, we get
\begin{eqnarray*} 
\lim_{n\longrightarrow \infty} \sum_{t=1}^{n} \mathbb{E}\left(\xi_{n, t}^{2} \mid \mathcal{F}_{t-1}\right)=\lim_{n\longrightarrow \infty} \left( \Upsilon_{n,1}+\Upsilon_{n,2}+\Upsilon_{n,3}\right)
=\frac{M_{2, W, 2} U^{2}(x)\pi(x) \omega(x)}{M_{1, W, 2}^{2} f_{1}(x)}=:\sigma^{2}_1(x)\,\  \hbox{a.s.,}
\end{eqnarray*}
\hbox{whenever} \,\  $f_{1}(x) > 0.$

\textbf{Proof of part (ii).}
Consider $a>1$ and $b >1$ such that $\displaystyle{\frac{1}{a}+\frac{1}{b}=1}$, by using  H\"{o}lder's and Markov's inequalities one can
have, for all $\zeta > 0$,
$$
\mathbb{E}\left\{\xi_{n, t}^{2} \1_{\left(\left|\xi_{n, t}\right|>\zeta\right)}\right\}
\leq \frac{\mathbb{E}|\xi_{n, t}|^{2a}}{\zeta^{2a/b}}\text{.}
$$
Taking $2a=\kappa + 2$ (where $\kappa$ is given in assumption (A3)(iii)) and let $C_{0}$ be a positive constant. Under assumptions (A1)(i), (A2)(iii)-(iv),(vi), (A3)(iii), (v), (A4)(iii)  and by using  Lemma \ref{Lemma 6.2}, we obtain
$$\begin{aligned}
n\mathbb{E}\left\{\xi_{n, t}^{2} \1_{\left(\left|\xi_{n, t}\right|>\zeta\right)}\right\}& \leq C_{0}\left\{\phi_{2}\left(h_{2}\right) / n\right\}^{(2+\kappa) / 2} \times \frac{n}{\left[\mathbb{E}\left(W_{1}(x)\right\}\right]^{2+\kappa}} \mathbb{E}\left\{\delta_{t}U^{2+\kappa}\left(X_{t}\right) W_{t}^{2+\kappa}(x)\left(\varepsilon_{t}^{2}-1\right)^{2+\kappa}\right\}\\
&\leq C_{0}\left\{n \phi_{2}\left(h_{2}\right)\right\}^{-\kappa / 2} \frac{M_{2+\kappa, W, 2} f_{1}(x)+o(1)}{M_{1, W, 2}^{2+\kappa} f_{1}^{2+\kappa}(x)+o(1)}\left\{U^{2+\kappa}(x)+o(1)\right\} \left\{ \pi(x)+o(1)\right\} \\
&=\mathcal{O}\left[\left\{n \phi_{2}\left(h_{2}\right)\right\}^{-\kappa / 2}\right]\text{.}
\end{aligned}$$ \\
Finally,  since $n \phi_{2}(h_{2})\longrightarrow \infty$ as $n\longrightarrow \infty$, we get
\begin{equation*} 
n \mathbb{E} \left\{ \xi_{n, t}^{2} \1_{\left(\left|\xi_{n, t} \right| > \zeta \right)}\right\}=o(1).
\end{equation*}
Therefore,  Lemma \ref{Lemma-T3} holds.
\end{proof}

\subsection*{Proof of Theorem \ref{T3}.}
Using the decomposition \eqref{decompvar}, and denoting $B_{n,0}(x) := \vartheta_{n,0}^{[4]}(x)/U_{n,0}^{[1]}(x),$ we have
$$\sqrt{n\phi_2(h_2)}\left(U_{n,0}(x)-U(x) - B_{n,0}(x)\right)=\sqrt{n\phi_2(h_2)} \left(\vartheta_{n,0}^{[1]}(x)+\vartheta_{n,0}^{[2]}(x)+\vartheta_{n,0}^{[3]}(x)\right)/U_{n,0}^{[1]}(x).$$
Then, the combination of statements in \eqref{supU1}- \eqref{32}, \eqref{33},  conditions \eqref{extracond}, and Lemma \ref{Lemma-T3}, allow to obtain the first part Theorem \ref{T3}. Moreover, by using the assumption that $\sqrt{n\phi_2(h_2)} h_2^\beta \rightarrow 0\; \text{as}\; n\rightarrow \infty$ as well as \eqref{35}, the 'bias' term $B_{n,0}(x)$ becomes negligible as $n$ goes to infinity which allows to obtain the second part of Theorem \ref{T3}.
\subsection*{Proof of Corollary \ref{Cor}}
 Let us consider\\
 
 $
 \dfrac{ \widehat{M}_{1, W, 2}}{\sqrt{\widehat{M}_{2, W, 2}}} \sqrt{\dfrac{n \widehat{F}_{x, 2}(h)\pi_{n}(x)}{\omega_{n,0}(x)U^{2}_{n,0}(x)}}(U_{n,0}(x)-U(x)) \\ =\dfrac{ \widehat{M}_{1, W, 2}   \sqrt{M_{2, W, 2}}}{\sqrt{\widehat{M}_{2, W, 2}}M_{1, W, 2}} \sqrt{\dfrac{n \widehat{F}_{x, 2}(h)\pi_{n}(x)U^{2}(x) \omega(x)}{\omega_{n,0}(x)\pi(x)U^{2}_{n,0}(x)n\phi_{2}(h_{2})f_{1}(x)}} \dfrac{M_{1, W, 2}}{\sqrt{M_{2, W, 2}}}\sqrt{\dfrac{n\phi_{2}(h_{2})f_{1}(x)\pi(x)}{ U^{2}(x) \omega(x)   }} (U_{n,0}(x)-U(x)).
$\\

By   Theorem \ref{T3}, we have
$$\dfrac{M_{1, W, 2}}{\sqrt{M_{2, W, 2}}}\sqrt{\dfrac{n\phi_{2}(h_{2})f_{1}(x)\pi(x)}{ U^{2}(x) \omega(x)   }} (U_{n,0}(x)-U(x))  \overset{\mathcal{D}}
{\longrightarrow}  \mathcal{N}\left(0, 1\right).$$\\
Then, by (A1)(i), (A2)(i),(iv), and following the same steps as  the proof of Corollary 1 in \cite{LL10}, we get
\begin{equation}\label{P1}
\widehat{M}_{1, W, 2} \stackrel{\mathbb{P}}{\rightarrow} M_{1, W, 2}, \quad
\widehat{M}_{2, W, 2}\color{black}  \stackrel{\mathbb{P}}{\rightarrow} M_{2, W, 2}, \quad \frac{\widehat{F}_{x, 2}(h_2)}{\phi_{2}(h_{2}) f_1(x)} \stackrel{\mathbb{P}}{\rightarrow} 1 \quad \hbox{as} \quad  n \rightarrow \infty .
\end{equation}
In addition, from Theorem \ref{T2}, we  have  $U_{n,0}(x) \stackrel{\mathbb{P}}{\rightarrow} U(x) $  and by equation   (5.25) in \cite{LING15}(page 86), we have
\begin{equation}\label{P2*} \quad \pi_{n}(x) \stackrel{\mathbb{P}}{\rightarrow} \pi(x)  \quad \hbox{as} \quad  n \rightarrow \infty.
\end{equation}
Then, it remains to prove that
\begin{equation}\label{P22}
\quad \omega_{n,0}(x) \stackrel{\mathbb{P}}{\rightarrow} \omega(x)  \quad \hbox{as} \quad  n \rightarrow \infty.
\end{equation}
\color{black}
Observe that
\begin{eqnarray*}
\omega_{n,0}(x) &=& \frac{1}{G_{n}(x)}\frac{1}{n\mathbb{E}(H_{1}(x))} \sum_{t=1}^{n}\delta_t \left(\frac{(Y_t - m_{n,0}(X_t))^2 - U_{n,0}(X_t)}{U_{n,0}(X_t)}\right)^2  H_t(x),
 \label{eq:omega_n}
\end{eqnarray*}
where $G_{n}(x)=\left(n\mathbb{E}(H_{1}(x))\right)^{-1}\displaystyle\sum_{t=1}^{n} \delta_t H_t(x).$
\\
Making use of Lemma 5.2  in \cite{LING15}, we get
\begin{equation}\label{CGn}
G_{n}(x) \stackrel{\mathbb{P}}{\rightarrow} \pi(x)  \quad \hbox{as} \quad  n \rightarrow \infty. 
\end{equation}
Further, let \(\mathfrak{{L}}_t\) be defined as
\[
\mathfrak{{L}}_t = \delta_t\left(\frac{(Y_t - m_{n,0}(X_t))^2 - U_{n,0}(X_t)}{U_{n,0}(X_t)}\right)^2 H_t(x), \quad \hbox{for} \quad t \in \{1, \ldots, n\}.
\] 
Then $$
 \frac{1}{n\mathbb{E}(H_{1}(x)}\sum_{t=1}^{n} \mathfrak{{L}}_t = \mathfrak{T}_{1,n}(x) + \mathfrak{T}_{2,n}(x), \label{eq:T1_T2}
$$
where
\[
\mathfrak{T}_{1,n}(x) = \frac{1}{n\mathbb{E}(H_{1}(x)} \sum_{t=1}^{n} \left\{\mathfrak{{L}}_t - \mathbb{E}(\mathfrak{{L}}_t | \mathcal{F}_{t-1})\right\}, \quad \mathfrak{T}_{2,n}(x) = \frac{1}{n\mathbb{E}(H_{1}(x)} \sum_{i=1}^{n} \mathbb{E}(\mathfrak{{L}}_t | \mathcal{F}_{t-1}).
\]
Let us turn our attention to the study of the second term \(\mathfrak{T}_{2,n}(x)\). By employing a double conditioning with respect to the \(\sigma\)-field \(\mathcal{G}_{t-1}\) and the $C_{r}$-inequality, we derive
\begin{equation*}\label{T-4}
\mathfrak{T}_{2,n}(x) = \frac{1}{n\mathbb{E}(H_{1}(x))}\sum_{t=1}^{n} \mathbb{E}\left( H_t(x)\, \mathbb{E}\left(\delta_t\left.\left\{\frac{(Y_t - m_{n,0}(X_t))^2 - U_{n,0}(X_t)}{U_{n,0}(X_t)}\right\}^2  \bigg|\mathcal{G}_{t-1}\right)\right| \mathcal{F}_{t-1}\right). \color{black} 
\end{equation*} 
Note that model (\ref{model}) allows us to write
\begin{align}
\mathbb{E}\left[\delta_t\left\{\frac{(Y_t - m_{n,0}(X_t))^2 - U_{n,0}(X_t)}{U_{n,0}(X_t)}\right\}^2 \bigg|\mathcal{ G}_{t-1}\right] &= \mathbb{E}\left[\delta_t \left(\mathcal{A} + \mathcal{B}\right)^2 | \mathcal{G}_{t-1}\right], \notag
\end{align}
where
$
\mathcal{A} = \frac{\left\{m(X_t) - m_{n,0}(X_t)\right\}^2 + \left\{U(X_t) - U_{n,0}(X_t)\right\} \varepsilon_t^2 + 2 \left\{m(X_t) - m_{n,0}(X_t)\right\} \sqrt{ U(X_t) } \varepsilon_t}{ U_{n,0}(X_t)},
$
and
$
\mathcal{B} = \varepsilon_t^2 - 1.
$
\\
Using assumption (A3)(i), condition (\ref{epsilon}), Theorems \ref{T1} and \ref{T2}, and \eqref{***u}, we obtain
$$
\mathbb{E}\left(\delta_t \mathcal{A}^2 | \mathcal{G}_{t-1}\right) = o_{\mathbb{P}}(1), \quad \text{and} \quad \mathbb{E}(\delta_t \mathcal{A}\mathcal{B} \big| \mathcal{G}_{t-1}) = o_{\mathbb{P}}(1).
$$
Finally, assumptions (A3)(ii) and (A3)(v) ensure that, as \(n \to \infty\), one gets
\begin{equation*}\label{E-4}
\mathbb{E}\left[\delta_t \left\{\frac{(Y_t - m_{n,0}(X_t))^2 - U_{n,0}(X_t)}{U_{n,0}(X_t)}\right\}^2\bigg| \mathcal{G}_{t-1}\right] = \omega(x) \pi(x) + o(1).
\end{equation*}
Then, using Lemma 7 in \cite{LL11}, under assumptions (A1)(i) and (A2)(i)-(vi), it follows, after using \eqref{CGn}, that
$$
G_n^{-1}(x) \mathfrak{T}_{2,n}(x) \stackrel{\mathbb{P}}{\longrightarrow} \omega(x) \,\  \text{ as } \,\ n \to \infty.
$$ 
Now, we have to show that \(\mathfrak{T}_{1,n}(x)\) goes to zero in probability as \(n \to \infty\). Using Markov's, Burkholder's, and Jensen's inequalities, one obtains, for any $\eta>0$,
\[
\mathbb{P} \left\{|\mathfrak{T}_{1,n}(x)| > \eta\right\} \leq \frac{C \mathbb{E}(\mathfrak{{L}}_1^2)}{ {n \eta^2 \left(\mathbb{E}(H_{1}(x))\right)^2}},
\]
where \(C\) is a generic positive constant. \\
Let us now find an upper bound of $\mathbb{E}(\mathfrak{{L}}_1^2).$ For this, observe that 
\begin{eqnarray*}
    \mathbb{E}(\mathfrak{{L}}_1^2) &=& \mathbb{E}\left( \mathbb{E}(\mathfrak{{L}}_1^2|X_1)\right)\\
    &=& \mathbb{E}\left\{ H_{1}^{2}(x) \mathbb{E}\left( \delta_1 \left(U_{n,0}^{-1}(X_1) (Y_1 - m_{n,0}(X_1))^2 -1 \right)^4 |X_1 \right)\right\}.
\end{eqnarray*}
Making use of the $C_{r}$-inequality and \eqref{***u}, we obtain
$$
 \mathbb{E}\left( \delta_1 \left(\frac{(Y_1 - m_{n,0}(X_1))^2}{U_{n,0}(X_1)} -1 \right)^4 |X_1 \right) \leq \dfrac{8}{\theta_1^4} \times \bigg\{\mathbb{E}\left(\delta_1 (m(X_1) - m_{n,0}(X_1) + \sqrt{U(X_1)}\varepsilon_1)^8 | X_1\right) + \pi(X_1) \bigg\}.
$$


By using Binomial Theorem, we have 
\begin{eqnarray*}
 \left((m(X_{t}) - m_{n,0}(X_{t})) + \sqrt{U(X_{t})} \varepsilon_{t} \right)^8  
&=& \sum_{k=0}^{8} \binom{8}{k}  (m(X_{t}) - m_{n,0}(X_{t}))^k (\sqrt{U(X_{t})} \varepsilon_{t})^{8-k} \\
&=& (m(X_{t}) - m_{n,0}(X_{t}))^8 + U^{4}(X_{t}) \varepsilon_{t}^8 \\
&&+ 8 \left((m(X_{t}) - m_{n,0}(X_{t}))^7 (\sqrt{U(X_{t})} \varepsilon_{t})\right) \\
&&+ 28 \left((m(X_{t}) - m_{n,0}(X_{t}))^6 U(X_{t})\varepsilon_{t}^2\right) \\
&&+ 56 \left((m(X_{t}) - m_{n,0}(X_{t}))^5 \sqrt{U(X_{t})} U(X_{t}) \varepsilon_{t}^3\right) \\
&&+ 70 \left((m(X_{t}) - m_{n,0}(X_{t}))^4 U^{2}(X_{t}) \varepsilon_{t}^4\right] \\
&&+ 56 \left((m(X_{t}) - m_{n,0}(X_{t}))^3 \sqrt{U(X_{t})}U^{2}(X_{t}) \varepsilon_{t}^5\right) \\
&&+ 28 \left((m(X_{t}) - m_{n,0}(X_{t}))^2 U^{3}(X_{t})\varepsilon_{t}^6\right) \\
&&+ 8 \left((m(X_{t}) - m_{n,0}(X_{t})) \sqrt{U(X_{t})}U^{3}(X_{t}) \varepsilon_{t}^7\right).
\end{eqnarray*}

Making use of Theorem \ref{T1}, the independence between $\varepsilon$ and $\delta$, as well as assumptions (A3)(i),(v), and (A4)(ii)-(iii), we can prove that
$$
 \mathbb{E}(\mathfrak{{L}}_1^2) = \mathbb{E} (\mathbb{E}(\mathfrak{{L}}_1^2 | X_1))< C \mathbb{E}(H_1^2(x)).
$$
Thus, by Lemma \ref{Lemma 6.2}, we have
\[
P \left\{|\mathfrak{T}_{1,n}(x)| > \eta\right\} \frac{C \mathbb{E}(H_1^2(x))}{ {n \eta^2 (\mathbb{E}(H_{1}(x)))^2}}\leq \frac{C(M_{2,H,3}f_{1}(x) + o(1))}{n \eta^2 (M_{1,H,3}^{2}f_{1}^{2}(x) + o(1))},
\]
which goes to zero as \(n \to \infty\). Combining this with \eqref{CGn}, one obtains $G_n^{-1}(x) \mathfrak{T}_{1,n}(x) \stackrel{\mathbb{P}}{\rightarrow} 0$. Consequently, the proof of equation (\ref{P22}) is achieved.

Finally, we obtain
$$ \frac{\widehat{M}_{1, W, 2}}{\sqrt{\widehat{M}_{2, W, 2}}}\frac{ \sqrt{M_{2, W, 2}}}{M_{1, W, 2}}
\sqrt{\frac{n \widehat{F}_{x, 2}(h_2)\pi_{n}(x)U^{2}(x) \omega(x)}{\omega_{n,0}(x) U^{2}_{n,0}(x)\pi(x)n\phi_{2}(h_{2})f_{1}(x)}}
\stackrel{\mathbb{P}}{\rightarrow} 1 \quad  \hbox{as}  \quad  n \rightarrow \infty.$$
Hence, the proof of Corollary \ref{Cor} is achieved.

\subsection*{Proof of Theorem \ref{T5}.}
By using Model (\ref{model}), we have
\begin{equation*}\label{Deco2}
\begin{aligned}
m_{n,1}(x)-m(x)&=\frac{1}{m_{n,1}^{[1]}(x)}\left[\frac{1}{n \mathbb{E} (K_1(x))}\displaystyle{\sum_{t=1}^{n}}[\delta_{t}\{m(X_{t})+\sqrt{U(X_{t})}\varepsilon_{t}\}+(1-\delta_{t})m_{n,0}(X_{t})-m(x)]K_{t}(x)\right]\\
&=\frac{1}{m_{n,1}^{[1]}(x)}\left[\mathcal{I}_{n,1}(x)+\mathcal{I}_{n,2}(x)+V_{n,0}(x)\right],
\end{aligned}
\end{equation*}
where $V_{n,0}(x)$ is  defined in (\ref{Vn0}) and
$$
 \begin{aligned}
\mathcal{I}_{n,1}(x)&=\frac{1}{n \mathbb{E}\left[K_{1}(x)\right]} \sum_{t=1}^{n}(1-\delta_{t})\left[m_{n,0}\left(X_{t}\right)-m\left(X_{t}\right)\right] K_{t}(x), \\
\mathcal{I}_{n,2}(x)&=\frac{1}{n \mathbb{E}\left[K_{1}(x)\right]} \sum_{t=1}^{n} \left(m\left(X_{t}\right)-m(x)\right) K_{t}(x),\\
\end{aligned}
$$
\begin{equation*}\label{EGn'}
m_{n,1}^{[1]}(x):=\frac{1}{n \mathbb{E}\left[K_{1}(x)\right]} \sum_{t=1}^{n}  K_{t}(x).
\end{equation*}
Then, we find
\begin{equation}\label{Decm1}
\displaystyle{\sup_{x\in \mathcal{C}}}\left| m_{n,1}(x)-m(x) \right| \leq \frac{\displaystyle{\sup_{x\in \mathcal{C}}}\left|\mathcal{I}_{n,1}(x)\right|+\displaystyle{\sup_{x\in \mathcal{C}}}\left|\mathcal{I}_{n,2}(x)\right|+\displaystyle{\sup_{x\in \mathcal{C}}}\left|V_{n,0}(x)\right|}{\displaystyle{\inf_{x\in \mathcal{C}}}\left|m_{n,1}^{[1]}(x)\right|}.
 \end{equation}
\cite{LL11} (in page 371) have proved that, as $n \rightarrow \infty$,
\begin{equation}\label{infGn'}
\displaystyle{\inf_{x\in \mathcal{C}}}\left|m_{n,1}^{[1]}(x)\right|>1 \quad \text{a.s.}
\end{equation}
Next, by using assumption (A4)(i) and the almost sure convergence of $m_{n,1}^{[1]}(x)$ to 1 uniformly in $x$ (see Lemma 7 in \cite{LL11}), \color{black} it follows that
\begin{equation}\label{I2}
\sup_{x\in \mathcal{C}}|\mathcal{I}_{n,2}(x)|=\mathcal{O}_{a.s.}(h_{1}^{\alpha}).
\end{equation}
Furthermore,  by the use of  Theorem \ref{T1}, the almost sure boundedness of $\delta_t$ by 1 and the almost sure uniform convergence of $m_{n,1}^{[1]}(x)$ to 1 (see Lemma 7 in \cite{LL11}),  we get 
\begin{equation}\label{I1}
\displaystyle{\sup_{x\in \mathcal{C}}}\left|\mathcal{I}_{n,1}(x)\right|=\mathcal{O}_{a.s.}(h_{1}^{\alpha})+
\mathcal{O}_{a.s.}\left(\lambda_{n}\right)\text{.}
\end{equation}
Finally, using inequality (\ref{Decm1}) combined with Lemma \ref{Lemmacvn} and results in (\ref{infGn'}), (\ref{I2}), (\ref{I1}), to obtain
$$\displaystyle{\sup_{x\in \mathcal{C}}}\left| m_{n,1}(x)-m(x) \right|=\mathcal{O}_{a.s.}(h_{1}^{\alpha})+\mathcal{O}_{a.s.}\left(\lambda_{n}\right)\text{.}$$
This concluded the proof of this Theorem.

\subsection*{Proof of Theorem \ref{T6}.}
Let
 \begin{equation}\label{'Decu1}
\begin{aligned}
U_{n,1}(x)-U(x)&=\frac{\displaystyle\sum_{t=1}^n [\delta_t (Y_t - m_{n,0}(X_t))^2 + (1-\delta_t)U_{n,0}(X_t)]W_{t}(x)}{\displaystyle\sum_{t=1}^n W_{t}(x)}-U(x)\\
&=\frac{1}{ U_{n,1}^{[1]}(x)}\left(\vartheta_{n,0}^{[1]}(x)+\vartheta_{n,0}^{[2]}(x)+\vartheta_{n,0}^{[3]}(x)+\Theta_{n,1}(x)+\Theta_{n,2}(x)\right),
\end{aligned}
\end{equation}
where $\vartheta_{n,0}^{[1]}(x), \vartheta_{n,0}^{[2]}(x), \vartheta_{n,0}^{[3]}(x)$ are defined in (\ref{Vn1}), (\ref{Vn2}), (\ref{Decoimp}), respectively, and  


$\begin{array}{lll}
\Theta_{n,1}(x)&:=& \displaystyle\frac{1}{{n\mathbb{E}(W_{1}(x))}}\displaystyle \displaystyle\sum_{t=1}^n (1-\delta_t)[U_{n,0}(X_t)-U(X_t) ]W_{t}(x), \nonumber\\
\Theta_{n,2}(x)&:=&\displaystyle\frac{1}{{n\mathbb{E}(W_{1}(x))}}\displaystyle \displaystyle\sum_{t=1}^n [U(X_t)-U(x)]  W_{t}(x),\nonumber\\
 U_{n,1}^{[1]}(x) &:=& \displaystyle\frac{1}{n \mathbb{E}\left[W_{1}(x)\right]} \sum_{t=1}^{n}  W_{t}(x).\label{EGn'}
\end{array}$ \\
Then, 
\begin{equation}
\label{DecompUn1}
\displaystyle{\sup_{x\in \mathcal{C}}}\left|U_{n,1}(x)-U(x)\right|\leq \dfrac{ \displaystyle{\sup_{x\in \mathcal{C}}}\left|\vartheta_{n,0}^{[1]}(x)\right|+\displaystyle{\sup_{x\in \mathcal{C}}}\left|\vartheta_{n,0}^{[2]}(x)\right|+\displaystyle{\sup_{x\in \mathcal{C}}}\left|\vartheta_{n,0}^{[3]}(x)\right|+\displaystyle{\sup_{x\in \mathcal{C}}}\left|\Theta_{n,1}(x)\right|+\displaystyle{\sup_{x\in \mathcal{C}}}\left|\Theta_{n,2}(x)\right|}{ \displaystyle{\inf_{x\in \mathcal{C}}}\left| U_{n,1}^{[1]}(x)\right| }.   
\end{equation}

According to Lemma 7 in \cite{LL11}, we can deduce that:
\begin{equation}\label{CpsG"}
\lim_{n\rightarrow \infty}\sup_{x\in \mathcal{C}}| U_{n,1}^{[1]}(x)-1|= 0 \quad \text{a.s.}
\end{equation}
Similar as for (\ref{infGn'}), one has
\begin{equation}\label{infGn''}
\displaystyle{\inf_{x\in \mathcal{C}}}\left| U_{n,1}^{[1]}(x)\right|>1 \quad \text{a.s.}\ \quad \text{as} \quad n \rightarrow \infty.
\end{equation}
Making use of Theorem \ref{T2}, equation (\ref{CpsG"})  and the almost sure boundedness of $\delta_t$, one gets 
\begin{equation}\label{'Thetan1}
\displaystyle{\sup_{x\in \mathcal{C}}}\left|\Theta_{n,1}(x)\right| = \mathcal{O}_{a.s.}(h_{1}^{2\alpha}+h_{2}^{\beta})+\mathcal{O}_{a \cdot s .}(\lambda_{n}^{\prime}+\lambda_{n}^2).\end{equation}
In addition, by using assumption (A4)(ii) and equation (\ref{CpsG"}), we obtain  
\begin{equation}\label{'Thetan2}
\displaystyle{\sup_{x\in \mathcal{C}}}|\Theta_{n,2}(x)|=  \mathcal{O}_{a.s.}(h_{2}^{\beta}).
\end{equation}

Finally, decomposition (\ref{DecompUn1}), combined  with equations (\ref{32}),  (\ref{33}), (\ref{34}), (\ref{infGn''})-(\ref{'Thetan2}) allows to conclude that 
$$\displaystyle{\sup_{x\in \mathcal{C}}}|U_{n,1}(x)-U(x)|= \mathcal{O}_{a . s .}(h_{1}^{2\alpha}+h_{2}^{\beta})+\mathcal{O}_{a \cdot s .}(\lambda_{n}^{\prime}+\lambda_{n}^2).$$
 This ends the proof of the Theorem.

\subsection*{Proof of Theorem \ref{T7}.} 
Using  decomposition in (\ref{'Decu1}), equations  (\ref{32}), (\ref{33}) and (\ref{'Thetan1}), (\ref{'Thetan2})  allow to conclude that  $\vartheta_{n,0}^{[1]}(x)$, $\vartheta_{n,0}^{[2]}(x)$, $\Theta_{n,1}(x)$ and $\Theta_{n,2}(x)$ are negligible as $n  \to \infty.$ Furthermore, according to 
  equation  (\ref{CpsG"}), $U_{n,1}^{[1]}(x)$ converges almost surely to 1. Finally,  we can conclude that the asymptotic distribution of the   nonparametric imputed   conditional variance is determined by  the asymptotic variance of   the term $\vartheta_{n,0}^{[3]}(x)$  which is   specified in Lemma \ref{Lemma-T3}.

\subsection*{Proof of Corollary \ref{CorIm}}
First, we note that \\
$
 \dfrac{ \widehat{M}_{1, W, 2}}{\sqrt{\widehat{M}_{2, W, 2}}} \sqrt{\dfrac{n \widehat{F}_{x, 2}(h_2)}{\omega_{n,1}(x)\pi_{n}(x) U^{2}_{n,1}(x)}}(U_{n,1}(x)-U(x)) \\ =\dfrac{ \widehat{M}_{1, W, 2}   \sqrt{M_{2, W, 2}}}{\sqrt{\widehat{M}_{2, W, 2}}M_{1, W, 2}} \sqrt{\dfrac{n \widehat{F}_{x, 2}(h_2)\omega(x)\pi(x) U^{2}(x)}{\omega_{n,1}(x)\pi_{n}(x) U^{2}_{n,1}(x)n\phi_{2}(h_{2})f_{1}(x)}} \dfrac{M_{1, W, 2}}{\sqrt{M_{2, W, 2}}}\sqrt{\dfrac{n\phi_{2}(h_{2})f_{1}(x)}{\omega(x)\pi(x)U^{2}(x)  }} (U_{n,1}(x)-U(x)).
$\\

By   Theorem \ref{T7}, we find
$$\frac{M_{1, W, 2}}{\sqrt{M_{2, W, 2}}}\sqrt{\frac{n\phi_{2}(h_{2})f_{1}(x)}{ \omega(x)\pi(x) U^{2}(x)   }} (U_{n,1}(x)-U(x))  \overset{\mathcal{D}}
{\longrightarrow} \mathcal{N}\left(0, 1\right).$$\\
Then, following the same steps as for (\ref{P22}), we get
\begin{equation}\label{P22+}
\quad \omega_{n,1}(x) \stackrel{\mathbb{P}}{\rightarrow} \omega(x)  \quad \hbox{as} \quad  n \rightarrow \infty.
\end{equation}

Thus, in view of Theorem \ref{T6} and equations (\ref{P1}), (\ref{P2*}) and (\ref{P22}), 
we obtain
$$ \frac{\widehat{M}_{1, W, 2}\color{black}}{\sqrt{\widehat{M}_{2, W, 2}}\color{black}}\frac{ \sqrt{M_{2, W, 2}}}{M_{1, W, 2}}
\sqrt{\frac{n \widehat{F}_{x, 2}(h)\omega(x)\pi(x) U^{2}(x)}{\omega_{n}(x)\pi_{n}(x) U^{2}_{n,1}(x)n\phi_{2}(h_{2})f_{1}(x)}}
\stackrel{\mathbb{P}}{\rightarrow} 1 \quad  \hbox{as}  \quad  n \rightarrow \infty.$$
Therefore, the proof of Corollary \ref{CorIm} is completed.


\color{black}

\section*{Appendix}
\begin{lem}\label{LemmaA}  Assume that $\left(X_{t}, \varepsilon_{t}\right)$ is a strictly stationary ergodic process and suppose that assumptions  (A3)(i),(v)  hold. Then, for each $\omega$ outside a null set $D$, there exists a positive integer $n_{0}(\omega)$ such that $V_{n,0}(x)=V_{n,0}^{\top}(x)$ for $n \geq n_{0}(\omega)$ and $x \in \mathcal{E}$.
\end{lem}
\begin{proof}

Recall that 
$\mathcal{L}_{t}=\delta_{t}\sqrt{U(X_{t})}\varepsilon_{t}.$ Then, for every $\eta > 0$ and by using Markov’s inequality, one has\\
$$ \begin{aligned}
\mathbb{P}\left(\displaystyle{\left|\frac{\mathcal{L}_{t}}{\ell_{n}}\right|}> \eta \right)=\mathbb{P}\left(|\mathcal{L}_{t}|> \eta \ell_{n}  \right) &\leq \displaystyle{\frac{\mathbb{E}(|\mathcal{L}_{t}|)}{\eta \ell_{n}}}\\
& \leq \displaystyle{\frac{\mathbb{E}(\delta_{t}|\sqrt{U(X_{t})}\varepsilon_{t}|)}{\eta \ell_{n}}}.\\
\end{aligned}
$$
Let $a > 1$ and $b > 1$ be real numbers such that $1/a + 1/b = 1$. By using the H\"{o}lder's inequality,  it follows that  
 $$ \begin{aligned}
\mathbb{P}\left(\displaystyle{\left|\frac{\mathcal{L}_{t}}{\ell_{n}}\right|}> \eta  \right)
& \leq \frac{1}{\eta  \ell_{n}}
\mathbb{E}^{1 / a}\left(\delta_{t}\left|\sqrt{U(X_{t})}\right|^{a} \right) \times \mathbb{E}^{1 / b}(|\varepsilon_{t}|^{b})
\end{aligned}.$$
By a double conditioning with respect to the $\sigma$-field $\mathcal{G}_{t-1}$ and by taking $a=
\rho /(\rho-1)$ and $b=\rho,$ we get 
 $$ \begin{aligned}
\mathbb{P}\left(\displaystyle{\left|\frac{\mathcal{L}_{t}}{\ell_{n}}\right|}> \eta \right)
& \leq \frac{1}{\eta \ell_{n}}
\mathbb{E}^{(\rho-1) / \rho }\left(\delta_{t}\left|\sqrt{U(X_{t})}\right|^{\rho /(\rho-1)} \right) \times  \mathbb{E}^{1 / \rho}[\mathbb{E}(|\varepsilon_{t}|^{\rho}|\mathcal{G}_{t-1})]. 
\end{aligned}$$
Making use  the second part of assumption (A3)(i), one finds
 $$ \begin{aligned}
\mathbb{P}\left(\displaystyle{\left|\frac{\mathcal{L}_{t}}{\ell_{n}}\right|}> \eta \right)
& \leq \frac{C}{\eta \ell_{n}}
\mathbb{E}^{(\rho-1) / \rho }\left(\delta_{t}\left|\sqrt{U(X_{t})}\right|^{\rho /(\rho-1)} \right).
\end{aligned} $$
Note that  another use of the H\"{o}lder inequality allows us to bound the quantity  $\mathbb{E}^{(\rho-1) / \rho }\left(\delta_{t}\left|\sqrt{U(X_{t})}\right|^{\rho /(\rho-1)} \right) $ as follows:

 $$ \begin{aligned}
\left|\mathbb{E}^{(\rho-1) / \rho }\left(\delta_{t}\left|\sqrt{U(X_{t})}\right|^{\rho /(\rho-1)} \right)\right|& \leq  \mathbb{E}^{(\rho-1) / b\rho }\left(\left|\sqrt{U(X_{t})}\right|^{b\rho /(\rho-1)} \right) \times \mathbb{E}^{(\rho-1) / a\rho }\left(|\delta_{t}|^{a} \right).  \\
\end{aligned} $$
By a double conditioning with respect to the $\sigma$-field $\mathcal{G}_{t-1}$  and  taking the same values of a and b   with the use of assumption (A3)(v), it follows that  
 $$ \begin{aligned}
\left|\mathbb{E}^{(\rho-1) / \rho }\left(\delta_{t}\left|\sqrt{U(X_{t})}\right|^{\rho /(\rho-1)} \right)\right|& \leq  \mathbb{E}^{(\rho-1) / b\rho }\left(\left|\sqrt{U(X_{t})}\right|^{b\rho /(\rho-1)} \right) \times  \mathbb{E}^{(\rho-1) / a\rho }[\mathbb{E}\left(\delta_{t}|\mathcal{G}_{t-1} \right)]   \\
&  \leq \mathbb{E}^{(\rho-1) / b\rho }\left(\left|\sqrt{U(X_{t})}\right|^{b\rho /(\rho-1)} \right) \times (\pi(x)+o(1))^{(\rho-1) / a\rho }.  \end{aligned}$$
 By putting $M= (\pi(x)+o(1))^{(\rho-1) / a\rho }$, one has   
 $$ \begin{aligned}
\left|\mathbb{E}^{(\rho-1) / \rho }\left(\delta_{t}\left|\sqrt{U(X_{t})}\right|^{\rho /(\rho-1)} \right)\right|& \leq  M\mathbb{E}^{(\rho-1) / \rho^{2} }\left(\left|\sqrt{U(X_{t})}\right|^{\rho^{2} /(\rho-1)} \right).   \\
\end{aligned} $$ 
which is finite by the first part of  assumption (A3)(i).\\ 
Thus,
 $$ \begin{aligned}
\mathbb{P}\left(\displaystyle{\left|\frac{\mathcal{L}_{t}}{\ell_{n}}\right|}> \eta \right)
& \leq \frac{\hat{C}}{\eta \ell_{n}}.
\end{aligned}$$
Since $0<\ell_{n} \uparrow \infty$ and $\eta >0$,  we get  
 \begin{equation*}
 \mathcal{L}_{t}/ \ell_{n} \longrightarrow 0 \quad \text{a.s.}
 \end{equation*}
Therefore, for $ \omega \in D^c$ with $\mathbb{P}(D^c)=1$ and some positive integer $n_{0}(\omega)$, it follows that 
$$|\mathcal{L}_{t}(\omega)| < \ell_{n},  \quad \quad t=1 \ldots n,\quad n\geq n_{0}(\omega).  $$
Finally,
$V_{n,0}(x)=V^{\top}_{n,0}(x).$

\end{proof}
\begin{flushleft}
\textbf{Proof of Lemma \ref{LemmaB}.}
\end{flushleft}
 Let $a>1$ and $b>1$ be real numbers such that $1 / a+1 / b=1$. By using assumption  A3(v)   and  the H\"{o}lder's and Markov's inequalities, we can get
$$
\begin{aligned}
\left|\mathbb{E}\left\{\mathcal{L}_{t} \1_{\left(\left|\mathcal{L}_{t}\right|  \leq \ell_{n}\right)} \mid \mathcal{F}_{t-1}\right\}\right| & \leq \mathbb{E}^{1 / a}\left(\left|\mathcal{L}_{t}\right|^{a} \mid \mathcal{F}_{t-1}\right) \times \mathbb{E}^{1 / b}\left\{\1_{\left(\left|\mathcal{L}_{t}\right| \leq \ell_{n}\right)} \mid \mathcal{F}_{t-1}\right\} \\
& \leq \mathbb{E}^{1 / a}\left(\left|\mathcal{L}_{t}\right|^{a} \mid \mathcal{F}_{t-1}\right) \times \frac{\mathbb{E}^{1 / b}\left(\left|\mathcal{L}_{t}\right| \mid \mathcal{F}_{t-1}\right)}{\ell_{n}^{1 / b}} \\ & \leq \frac{\mathbb{E}\left(\left|\mathcal{L}_{t}\right|^{a} \mid \mathcal{F}_{t-1}\right)}{\ell_{n}^{a / b}}=\frac{\mathbb{E}\left(|\delta_{t}|\left|\sqrt{U\left(X_{t}\right)}\right|^{a}\left|\varepsilon_{t}\right|^{a} \mid \mathcal{F}_{t-1}\right)}{\ell_{n}^{a / b}}\\
& \leq (\pi(x)+o(1))\frac{\mathbb{E}\left(\left|\sqrt{U\left(X_{t}\right)}\right|^{a}\left|\varepsilon_{t}\right|^{a} \mid \mathcal{F}_{t-1}\right)}{\ell_{n}^{a / b}}.
\end{aligned}
$$Then, by using Lemma \ref{Lemma 6.2}, assumptions (A1), (A2)(i), (iv), (vi) and (A3)(i), along with the first part of condition (\ref{1}) and by following similar steps as  the proof of Lemma B in \cite{Chaouch2019} with specific values for a and b  ($a=\rho$ and $ b=\rho /(\rho-1)$), we obtain

$$
\sup _{x \in \mathcal{C}}\left|\widetilde{V}_{n,0}(x)\right| \leq \frac{C}{\ell_{n}^{\rho-1} \phi_{1}\left(h_{1}\right)^{(\rho-1) / \rho}} \,\ \hbox{almost surely as} \,\ n \rightarrow \infty.
$$




\end{document}